\documentclass{article}
\pdfoutput=1

\usepackage[utf8]{inputenc} 
\usepackage[T1]{fontenc}    
\PassOptionsToPackage{hyphens}{url} 
\usepackage[hyperfootnotes=false]{hyperref}
\usepackage{booktabs}       
\usepackage{amsfonts}       
\usepackage{nicefrac}       
\usepackage{microtype}      
\usepackage[table]{xcolor}  
\usepackage{breqn}
\usepackage{seqsplit}
\usepackage[hang,flushmargin]{footmisc}
\usepackage{algorithm}
\usepackage{algpseudocode}
\usepackage{graphicx}
\usepackage{enumerate}
\usepackage{caption}
\usepackage{tcolorbox}      
\usepackage{array}           
\usepackage{subcaption}
\usepackage{comment}
\usepackage[export]{adjustbox}
\usepackage{mathtools}
\usepackage{amsthm}
\usepackage{authblk}
\usepackage{fullpage}
\usepackage{cancel}
\usepackage{nicematrix}
\usepackage[ruled,vlined,algo2e]{algorithm2e}
\usepackage{mwe}
\usepackage{seqsplit}
\usepackage{placeins}

\usepackage{enumitem}

\usepackage[round]{natbib}

\usepackage{tikz}
\usetikzlibrary{positioning}
\usetikzlibrary{calc}
\usetikzlibrary{fit}
\usepackage[HTML]{xcolor}
\usepackage{tcolorbox}

\usepackage{Definitions}
\usepackage{dsfont}

\definecolor{mycommentcolor}{HTML}{4f9739}
\algrenewcommand\algorithmiccomment[1]{\hfill\textcolor{mycommentcolor}{\(\triangleright\) #1}}

\newlength{\inlineheight}
\algnewcommand{\LineComment}[1]{\State \textcolor{mycommentcolor}{\(\triangleright\) #1}}

\title{Test-Time Compute Games}

\author{Ander Artola Velasco$^{\S}$}
\author{Dimitrios Rontogiannis$^{\S}$}
\author{Stratis Tsirtsis$^{\dagger}$}
\author{Manuel~Gomez-Rodriguez$^{\S}$}
\affil{$^{\S}$Max Planck Institute for Software Systems, Kaiserslautern, Germany \\
\{avelasco, drontogi, manuel\}@mpi-sws.org}

\affil{$^{\dagger}$Hasso Plattner Institute, Potsdam, Germany \\ stratis.tsirtsis@hpi.de}

\date{}

\begin{document}

\maketitle

\begin{abstract}
Test-time compute has emerged as a promising strategy to enhance the reasoning abilities of large language models (LLMs).
%
However, this strategy has in turn increased how much users pay cloud-based providers offering LLM-as-a-service, since providers charge users for the amount of test-time compute they use to generate an output.
%
In our work, we show that the market of LLM-as-a-service is socially ineffi\-cient---providers have a financial incentive to increase the amount of test-time compute, even if this increase contributes little to the quality of the outputs.
%
To address this inefficiency, we introduce a reverse second-price auction mechanism where providers bid their offered price and (expected) quality for the opportunity to serve a user, and users pay proportionally to the marginal value generated by the winning provider relative to the second-highest bidder.
%
To illustrate and complement our theoretical results, we conduct experiments with multiple instruct models from the \texttt{Llama} and \texttt{Qwen} families, as well as reasoning models distilled from \texttt{DeepSeek-R1}, on benchmark datasets covering mathematics, science, and open-ended question answering.
\end{abstract}

\section{Introduction}
\label{sec:intro}
The massive computational resources required to run large language models (LLMs) have led many users to rely on a growing market of cloud-based service providers that offer LLM-as-a-service. 
%
A key driver of this computational demand is the use of test-time compute (TTC)---additional computations performed by an LLM at inference time to improve its performance using techniques such as chain-of-thought~\citep{wei2022chain}, tree-of-thought~\citep{yao2023tree}, and best-of-n sampling~\citep{chow2024inference}. 
%
%
Importantly, in the 
LLM-as-a-service market, the cost of these additional computations ultimately falls on the users,\footnote{\label{footnote:apis}\url{https://ai.google.dev/gemini-api/docs/pricing}, \url{https://openai.com/api/pricing/}, \url{https://www.claude.com/pricing}.} who pay for all tokens generated by a model during inference, including intermediate tokens that are not visible in the final outputs observed by the users.
%

In this context, providers have the flexibility to (dynamically) adjust the amount of test-time compute an LLM uses to respond to a user'{}s query.\footnote{For example, GPT-5 uses real-time routing to adjust reasoning depth, and thus test-time compute. See  \url{https://openai.com/index/introducing-gpt-5/}.} However, in a competitive market, this flexibility creates a new strategic dimension beyond how providers price their services.
%
%
In particular, providers can strategically increase the amount of test-time compute allocated to a user'{}s query to maximize profit, even if the additional test-time compute contributes little to the quality of the response. 
Consider a simple illustrative example where, for a given set of queries with verifiable ground truth (\eg, diagnosing patients based on their medical records), two different providers can run their LLMs in either a low-TTC mode (\eg, standard generation) or a high-TTC mode (\eg, chain-of-thought), with average accuracies and generation costs for the providers given by:
\begin{table}[H]
\centering

\begin{minipage}{0.45\textwidth}
\centering
\textbf{Provider 1}\\[4pt]
\begin{tabular}{lcc}
\toprule
{}        & Avg.\ accuracy & Avg.\ gen. cost \\
\midrule
Low TTC   & 70\%           & \$0.25          \\
High TTC  & 90\%           & \$1             \\
\bottomrule
\end{tabular}
\end{minipage}
\hfill
\begin{minipage}{0.45\textwidth}
\centering
\textbf{Provider 2}\\[4pt]
\begin{tabular}{lcc}
\toprule
{}        & Avg.\ accuracy & Avg.\ gen. cost \\
\midrule
Low TTC   & 50\%           & \$0.5           \\
High TTC  & 95\%           & \$10            \\
\bottomrule
\end{tabular}
\end{minipage}

\end{table}
If both providers price their models with a $25\%$ profit margin over their generation costs, it is easy to see that a user 
who values each percentage point of accuracy as $\$0.02$ would always select the first provider, who offers them higher value (\ie, $\$0.02\times\text{accuracy} - \text{price}$) independently of the TTC mode chosen by the second provider.
%
However, to increase their profit, the first provider is financially incentivized to choose the high-TTC mode, even though the low-TTC mode would maximize the sum of the user value and provider profit, and would therefore be socially optimal.\footnote{We formally introduce user value and provider utilities in Section~\ref{sec:model}. For the exact derivations related to this example, refer to Appendix~\ref{app:example}.}

In our work, we formally characterize the above gap and show that, when providers act to maximize profit in a competitive market, they cannot, in general, be expected to use the socially optimal amount of test-time compute.
%
%
Further, we conceptualize and analyze a forward-looking market
where users and providers participate in a reverse second-price auction, and show that, in contrast to the current pay-for-compute market, this auction mechanism is guaranteed to be socially optimal.
%

%

\xhdr{Our contributions}
%
%
We start by modeling the current LLM-as-a-service market as a normal-form game~\citep{NISAN} between $N$ providers who choose their test-time compute allocations to maximize profit, and users pay according to the price per compute set by their selected provider. Building on this characterization, we make the following contributions:

\begin{enumerate}

    \item We show that the LLM-as-a-service market 
    admits a pure Nash equilibrium
    and that, under natural assumptions, 
    the market providers are guaranteed to reach this equilibrium in finite time.
    
    \item We characterize the test-time compute used by providers in an LLM-as-a-service market at the Nash equilibrium and show that, in general, it is not socially optimal.

    \item We introduce a reverse second-price auction mechanism where providers bid with their offered price and (expected) quality for the opportunity to serve a user, and the user's price is determined by the marginal value generated by the winning provider relative to the second-highest bidder.
    
    \item We show that the above auction mechanism is socially optimal, and it guarantees that users always obtain at least their second-best achievable value for the task and providers secure non-negative profits.

\end{enumerate}

To complement our theoretical results,  we conduct experiments with multiple instruct models from the \texttt{Llama} and \texttt{Qwen} families, as well as reasoning models distilled from \texttt{DeepSeek-R1}, on math, science, and open-ended question answering benchmark datasets.
Our results show that the existing pay-for-compute market is up to $19\%$ (socially) inefficient, as measured by the Price of Anarchy.\footnote{The code for our experiments is publicly available at \url{https://github.com/Human-Centric-Machine-Learning/strategic-ttc}.}

%
\xhdr{Further related work}
%
%
Our work connects to the rapidly growing literature on the economic and strategic aspects of machine learning systems~\citep{KleinbergMono,tsirtsisStrategic,einav2025a,saig2025evolutionary,rosenfeld2025machinelearningmaximizewelfare,gauthierstatistical}, and more specifically, LLMs-as-a-service~\citep{duetting2024mechanism,la2024language,mahmood2024pricing,laufer2024fine,cai2025are,saig2024incentivizing,bergemann2025economicslargelanguagemodels,velasco2025llmoverchargingyoutokenization, velasco2025auditingpaypertokenlargelanguage, iyer2025sellserviceuncertainoutcomes,cao2025pay,sun2025coincountinginvisiblereasoning}.
%
%
Therein, most closely related to ours is a strand of recent works that analyzes the incentives of providers of LLM-as-a-service when they act strategically. Specifically,~\citet{velasco2025llmoverchargingyoutokenization,velasco2025auditingpaypertokenlargelanguage} study how providers can overcharge users by misreporting the number of tokens used by the LLM to encode a given text, and introduce an auditing method to detect such behavior. In a similar vein,~\citet{sun2025coincountinginvisiblereasoning} and~\citet{cao2025pay} consider a setting where a provider injects additional reasoning tokens to inflate the user's payment, while~\cite{saig2024incentivizing} and~\citet{cai2025are} study a scenario in which providers have a financial incentive to be unfaithful to users by deploying cheaper-to-run models to maximize profit. 

%
%
Further, our work relates to a line of work that analyzes the 
capabilities enabled by test-time compute~\citep{wen2025parathinkernativeparallelthinking, wu2025lessunderstandingchainofthoughtlength, bi2025forestofthoughtscalingtesttimecompute, liu20251bllmsurpass405b, snell2024scalingllmtesttimecompute}. While most of the literature has focused primarily on performance aspects, there is increasing interest in analyzing the economic implications of test-time compute. In particular, recent works by~\citet{wang2024reasoning},~\citet{zellinger2025economicevaluationllms} and~\citet{erol2025cost} argue for incorporating the substantial costs of test-time compute into LLM evaluation and ranking, while~\citet{wan2025beacon},~\citet{komiyama2025best} and~\citet{kalayci2025optimal} develop methods for selecting test-time compute resources once performance gains become marginal under majority voting and best-of-n.
However, our work is the first to show that, in the LLM-as-a-service market, providers are incentivized to strategically decide about the amount of test-time compute used by the LLMs they serve, and to examine the resulting effects on the social welfare. 

%
%
Finally, our work also draws on the classic literature in auctions and mechanism design~\citep{krishna2009auction,milgrom2004putting, Milgrom1982,Che1993,Nisan2007-lx}, which has studied optimal auction design~\citep{kersten2014multiattribute, perry2002efficient, takanashi2019efficiency, myerson1981optimal,riley1981optimal, maskin1984optimal,bhawalkar2011welfare,syrgkanis2013composable, Bergemann_data}, mechanisms to incentivize truthful bidding~\citep{bar2002incentive, ledyard1987incentive,vickrey1961counterspeculation,clarke1971multipart,groves1973incentives,Akbarpour2020}, and the (in)efficiency of various markets when players act in their own interest~\citep{koutsoupias1999worst,roughgarden2002bad,roughgarden2015intrinsic,hartline2023robust}. 
%
%
Within this literature, the mechanism we propose for test-time compute markets is closely related to scoring auctions studied in the context of procurement contracts and supplier selection~\citep{Che1993,roughgarden2017price,CAI2026130083}, where bidders submit offers specifying not only a price but also additional attributes of their offer (\eg, quality), and the seller uses a scoring rule to determine the winner. In many settings, such scoring auctions have been shown to outperform first-price auctions, achieving higher efficiency~\citep{AskerScore,awaya2025quality}; yet, to the best of our knowledge, they have not been considered in the context of LLM-as-a-service.

\section{A Game-Theoretic Model of Test-Time Compute}
\label{sec:model}
We model the LLM-as-a-service market as a normal-form (test-time compute) game $\Gcal$~\citep{NISAN} between $N$ providers who simultaneously deploy their own pre-trained model and compete for user demand by strategically selecting the level of test-time compute they use.
In practice, each provider $i \in [N]$ can select different levels of test-time compute for different types of tasks (\eg, low for fact retrieval and high for mathematical reasoning or coding), and their choice of compute for one task does not restrict their choices for other tasks. 
As a result, providers compete to serve each task independently, 
and each task defines an independent game between providers. 

More formally, given a task $\Tcal$ characterized by a set of queries $\Qcal$, each provider can select a test-time compute level $\theta\in\Theta$ for the model they serve from a finite set $\Theta$ of test-time compute levels.\footnote{For ease of exposition, we consider that all providers share the same action space $\Theta$.}
For instance, in test-time compute methods such as majority voting~\citep{wang2023selfconsistency} and best-of-n sampling~\citep{chow2024inference}, each level of test-time compute $\theta$  naturally corresponds to the number of generated samples and, in chain-of-thought~\citep{wei2022chain}, it may correspond to different ``reasoning effort'' levels.
%
%
Then, given a test-time compute level $\theta$, the model served by provider $i$ produces an average output quality $q_i(\theta)\in\mathbb{R}_{+}$ for the task $\Tcal$, 
which may scale differently across providers, \ie, $q_i \neq q_j$ for $i\neq j$. 
In practice, the average output quality $q_i$ may correspond to average accuracy across queries if the task $\Tcal$ has a verifiable ground truth, or to an average user satisfaction score if the task $\Tcal$ is open-ended.
Moreover, there is extensive empirical evidence~\citep{wen2025parathinkernativeparallelthinking, wu2025lessunderstandingchainofthoughtlength, bi2025forestofthoughtscalingtesttimecompute} that the quality $q_i$ typically increases with the level of test-time compute, \ie,
\begin{equation}
    q_i(\theta) \leq  q_i(\theta')\quad \text{for}\quad \theta \prec \theta',
\end{equation}
where $\prec$ denotes a total order in the space $\Theta$ describing whether $\theta$ corresponds to a level of compute that is on average lower than $\theta'$ (\eg, as measured by the average number of tokens used by the model per query). 

Although provider $i$ can increase the average quality $q_i(\theta)$ they offer to users by selecting a higher level of test-time compute $\theta$, this comes at the expense of a higher (average) generation cost $c_i(\theta) \in \mathbb{R}_{+}$ for the provider.
%
As a result, to compensate for higher generation costs, the price $p_i(\theta)$ charged by each provider $i$ increases with the level of test-time compute $\theta$, \ie, $p_i(\theta) < p_i(\theta')$ for $\theta \prec \theta'$ (see footnote~\ref{footnote:apis}).
Here, since each provider $i$ is unlikely to set a price per unit of test-time compute lower than the cost,\footnote{In practice, a provider's generation cost and price for a model'{}s output is proportional to the number of tokens it used to generate the output.}
we assume any increase in test-time compute leads to a positive marginal profit for the provider, \ie,
\begin{equation}\label{eq:increasing_utility}
        p_i(\theta) -  c_i(\theta) < p_i(\theta') -  c_i(\theta') \quad \text{for}\quad \theta \prec \theta'.
\end{equation}

Further, given a test-time compute level $\theta$, each provider $i$ offers a certain average value $V_i(\theta)$ to users for the task $\Tcal$, which is given by the difference between the average output quality $q_i(\theta)$ produced by their model and the average price $p_i(\theta)$ charged by the provider for task $\Tcal$. 
That is,
\begin{equation}\label{eq:user-utility}
    V_i(\theta) = q_i(\theta) - p_i(\theta),
\end{equation}
where we assume no ties in the values offered by different providers, \ie, $V_i(\theta) \neq V_j(\theta')$ for all $i\neq j$ and all $\theta, \theta' \in \Theta$, since in practice both $q_i$ and $p_i$ refer to average quantities and thus take real values.\footnote{All results extend straightforwardly to settings with ties without altering the conclusions.}

Given the test-time compute levels selected by all providers, $\thetab = (\theta_1,\ldots,\theta_N) \in \Theta^N$, users allocate their demand across providers based on the value $V_i(\theta_i)$ offered by each provider, but they also have the option to abstain from using any provider if none of them offers a value higher than an abstention threshold value $V_0 \in \mathbb{R}$. 
%
More formally, we characterize the allocation of user demand across providers through a market share function $s: \mathbb{R}^{N+1} \to [0,1]^{N+1}$ that maps the values offered by all providers and the abstention value to the fraction of queries served by each provider and the fraction of queries where users abstain.
Moreover, as is typical in game theory, we use the notation $s_i(V_i(\theta_i);\Vb_{-i}(\thetab_{-i}))$ to denote the market share of provider $i$ as a function of their offered value $V_i(\theta_i)$ while keeping fixed all other values $\Vb_{-i}(\thetab_{-i})$ (\ie, the abstention value and the values offered by providers other than $i$).

In the remainder of the paper, we will focus on the two canonical market share functions shown in Table~\ref{tab:market-models}.
%
\begin{table}[t]
\centering
\caption{Market share function, $s_i(V_i(\theta_i);\Vb_{-i}(\thetab_{-i}))$} \label{tab:market-models}
\begin{NiceTabular}{cc}
\toprule
\Block[fill=gray!10]{}{\textbf{Perfect rationality}} & \Block[fill=gray!10]{}{\textbf{Bounded rationality}} \\
\midrule
$\displaystyle 
\mathds{1}\left\{  V_i(\theta_i) > \max\left\{\Vb_{-i}(\thetab_{-i}) \right\} \right\}$ & $\displaystyle 
\frac{e^{\beta V_i(\theta_i)}}{e^{\beta V_0} + \sum_{j=1}^N e^{\beta V_j(\theta_j)}}$ \\[2ex]
\bottomrule
\end{NiceTabular}
\end{table}
The first function corresponds to a Bertrand-style market~\citep{Tirole1988} where users are perfectly rational and always select the provider offering them the highest value.
%
The second function corresponds to a market where users are boundedly rational~\citep{MCKELVEY19956} and select providers probabilistically based on their offered values.\footnote{In the context of random utility models (RUMs), boundedly rational users are those who only have access to noisy estimations of the value offered by each provider~\citep{yao2023rethinking}.}
Here, the parameter $\beta > 0$ controls the degree of rationality; as $\beta$ increases, users become more rational and, in the limit $\beta \to \infty$, users become perfectly rational.


%

Given the choices of test-time compute levels made by all providers and the resulting market share they attain, each provider $i$ obtains a utility $U_i(\theta_i;\thetab_{-i})$ given by
\begin{equation}\label{eq:provider-utility}
    U_i(\theta_i;\thetab_{-i}) = \underbrace{s_i(V_i(\theta_i);\Vb_{-i}(\thetab_{-i}))}_{\text{Provider's market share}} \cdot \underbrace{\left(p_i(\theta_i) - c_i(\theta_i)\right)}_{\text{Provider's profit}}.
\end{equation}

%
%
In the next section, we will analyze the (social) welfare $\SW(\thetab)$ 
generated by an LLM-as-a-service market, 
which naturally corresponds to the sum of the total value obtained by users and the total utility obtained by all providers. More formally:
\begin{equation}\label{eq:social-welfare-definition}
\begin{aligned}
\SW(\thetab) 
  &\coloneqq   \sum_{i=1}^N U_i(\thetab_i;\thetab_{-i}) + \sum_{i=1}^{N} V_i(\thetab_i) \cdot s_i(V_i(\theta_i);\Vb_{-i}(\thetab_{-i})) +  V_0\cdot s_0(V_0;\Vb(\thetab))\\ 
  &= \sum_{i=1}^{N} s_i(V_i(\theta_i);\Vb_{-i}(\thetab_{-i}))\cdot  \underbrace{(q_i(\theta_i)-c_i(\theta_i))}_{\coloneqq \SW_i(\theta_i)}  +  V_0\cdot s_0(V_0;\Vb(\thetab)), 
\end{aligned}
\end{equation}
where prices $p_i(\theta)$ cancel out, and $\SW_i(\theta_i)$ denotes each provider’s contribution to social welfare, reflecting the quality they offer relative to the generation cost incurred by their chosen level of test-time compute.
%

%
%
%

\section{Social Welfare of Test-Time Compute Games}
\label{sec:equilibrium}
In this section, we will show that, in the LLM-as-a-service market, we can expect the test-time compute levels selected by the market providers to be suboptimal in terms of social welfare.
To this end, as is typical in game theory, we will analyze the (pure) Nash equilibria of the test-time compute game $\Gcal$ introduced in Section~\ref{sec:model}.

A choice of test-time compute levels $\thetab^\dagger\in\Theta^N$ is a pure Nash equilibrium of
$\Gcal$ if no provider can improve their utility through a unilateral change in test-time compute, \ie,
\begin{equation}\label{eq:NE-definition}
    U_i(\theta^\dagger_i;\thetab^\dagger_{-i}) \geq U_i(\theta'_i;\thetab^\dagger_{-i}) \,\,\,\text{for all}\,\,\, i\in[N], \theta'_i \in \Theta.
\end{equation}
%

To determine the existence of pure Nash equilibria, we resort to the theory of potential games~\citep{MONDERER1996124}.
Specifically, we start by showing that every test-time compute game is a \emph{generalized ordinal potential game}, that is, there exists a function $\Phi \,\colon\, \Theta^N \to \RR$, called a potential, such that any unilateral deviation of some provider $i$ that strictly increases their utility also strictly increases the potential. More specifically, let 
\begin{equation}\label{eq:potential}
        \Phi(\thetab)=
        \begin{dcases}
                \log \left(p_{\pi(1)}(\theta_{\pi(1)})-c_{\pi(1)}(\theta_{\pi(1)})\right) + C\cdot V_{\pi(2)}(\theta_{\pi(2)}) & (\beta=\infty) \\ 
                \sum_{i=1}^N \log\left(p_i(\theta_i)-c_i(\theta_i)\right) +  \beta\cdot\sum_{i=1}^N  V_i(\theta_i) - \log \left(e^{\beta V_0} + \sum_{i=1}^N e^{\beta V_i(\theta_i)} \right) & (\beta < \infty) 
        \end{dcases}
\end{equation}
where $C>0$ is a constant whose value is explicitly given in Appendix~\ref{app:proof-potential}, and $\pi$ is the ordering of providers in decreasing order of the value they offer under $\thetab$. Then, we have the following theorem:
\begin{theorem}\label{thm:potential-game}
    The function $\Phi$ defined in Eq.~\ref{eq:potential} is a potential for the test-time compute game $\Gcal$, \ie, for all $\thetab\in\Theta^N$, $i\in[N]$, and $\theta'_i\in\Theta$, it holds that $U_i\left(\theta'_i ; \thetab_{-i}\right) > U_i\left(\theta_i ; \thetab_{-i}\right)  \Rightarrow \Phi\left(\theta'_i ; \thetab_{-i}\right) > \Phi\left(\theta_i ; \thetab_{-i}\right)$.   
\end{theorem}

Intuitively, Theorem~\ref{thm:potential-game} reveals that, although each provider acts selfishly to increase their own utility, the structure of the market is such that their actions jointly optimize the potential $\Phi(\thetab)$,
which captures the overall state of the market.
In this context, it is worth noting that competition between providers drives up both the value offered to the users and (some of) the providers' profits. For example, under perfect rationality ($\beta=\infty$), maximizing the potential increases the second-highest value offered to the users (and hence, also the highest) as well as the profit of the provider offering the highest value.

Following well-known results in game theory, the characterization of the test-time compute game $\Gcal$ as an ordinal potential game readily implies that the game admits a (not necessarily unique) pure Nash equilibrium~\citep{MONDERER1996124}.\footnote{Refer to Chapters $18$ and $19$ of~\cite{NISAN} for a comprehensive discussion of potential games.}
To understand why this holds, consider that the levels of test-time compute selected by the providers start at some arbitrary point $\thetab^1\in\Theta^N$ and follow \emph{better-response dynamics} over time.
Formally, this means that, at each time step $t$, the levels of test-time compute $\thetab^t$ are such that either (i) there exists a single provider $i\in[N]$ who changes their level of test-time compute to $\theta^{t+1}_i \neq \theta^t_i$ with $U_i(\theta^{t+1}_i;\thetab_{-i}^{t}) > U_i(\theta^{t}_i;\thetab^{t}_{-i})$ or (ii) $\thetab^{t+1} = \thetab^{t}$ if no such change is possible by any provider.
Then, since the domain $\Theta^N$ is finite, it is guaranteed that the sequence $\thetab^t$ eventually becomes constant---otherwise, the better-response dynamics would keep increasing the potential function $\Phi$ indefinitely---at which point no provider can further increase their utility and the market reaches a pure Nash equilibrium.
In this context, note that better-response dynamics 
are particularly natural in practice.
One can think of the better-response of a provider as the (minor) release of a model that has not undergone additional pre-training or post-training, but whose performance has changed due to optimizations in test-time compute methods.

Further, we proceed by identifying properties that hold in the equilibrium of a test-time compute game $\Gcal$, starting from an explicit characterization of the levels of test-time compute selected by providers under perfect rationality.
When users are perfectly rational,
they always opt for the provider that offers them the highest value (see Table~\ref{tab:market-models}). 
As a consequence, each provider's incentive is to select a level of test-time compute that offers higher value than what their competitors are offering. 
This observation allows us to determine the equilibrium reached by the providers using the maximum value $V_i^*$ that each provider can offer, \ie,
\begin{equation}\label{eq:max-user-value}
        V^*_i := \max_{\theta_i \in \Theta}\Big\{q_i(\theta_i) - p_i(\theta_i)\Big\},
\end{equation}
where, without loss of generality, we assume that providers are indexed such that \mbox{$V^*_{1} \geq V^*_{2} \geq \dots \geq V^*_{N}$} and we will refer to the provider who can offer the highest value $V^*_{1}$ as the \emph{dominant} provider.
In particular, we have the following theorem:
\begin{theorem}\label{thm:rational-equilibrium}
     In any pure Nash equilibrium $\thetab^\dagger$ of a test-time compute game $\Gcal$ with perfectly rational users,
     (i) the dominant provider serves all queries with $\theta^\dagger_{1}
            =\max_{\theta_1\in\Theta} \left\{ \theta_{1} \;\middle|\; V_{1}(\theta_{1}) > V^*_{2} \right\}$, 
    and (ii) there exists at least one provider $i\neq 1$ such that
    \begin{equation*}
        V_{i}(\theta_i^\dagger) > \max_{\theta_{1}\in\Theta}\left\{ V_{1}(\theta_{1}) \; \middle| \;  V_{1}(\theta_{1}) < V_{2}^* \quad\mathrm{and}\quad p_{1}(\theta_{1}) - c_{1}(\theta_{1}) > p_{1}(\theta^\dagger_{1}) - c_{1}(\theta^\dagger_{1}) \right\},
    \end{equation*}
        where the right-hand side of the inequality is $-\infty$ if the set is empty.
\end{theorem} 
The above result states that, in equilibrium, the dominant provider captures the market by selecting a level of test-time compute that offers higher value to the users than what their competitors can possibly offer and, once they achieve that, they increase their test-time compute as much as possible to maximize their utility (see also Eq.~\ref{eq:increasing_utility}).
%
At the same time, at least one competitor offers a sufficiently high-value alternative while gaining zero utility---otherwise, the dominant provider could also lower the value they offer, and the situation would not be an equilibrium.

For games $\Gcal$ with boundedly rational users ($\beta < \infty$), the levels of test-time compute in equilibrium cannot be characterized explicitly,
since all providers serve a positive fraction of user queries and their chosen levels of test-time compute depend on the costs, prices, and qualities $\{c_i, p_i, q_i\}_{i=1}^N$ of all providers.
%
Nonetheless, we next show that, as long as users are sufficiently rational, any Nash equilibrium of the game $\Gcal$ with $\beta<\infty$ is also an equilibrium of the game $\Gcal$ with perfectly rational users and, hence, also 
satisfies the properties of  Theorem~\ref{thm:rational-equilibrium}.
Formally, we have the following theorem:
\begin{theorem}\label{thm:high-rationality}
    Consider a set of providers specified by $\{ c_i, p_i,q_i \}_{i=1}^N$. Then, there exists a (finite) level of user rationality $\beta_0>0$ such that, for any $\beta > \beta_0$, any pure Nash equilibrium $\thetab^\dagger$ of the game $\Gcal$ with boundedly rational users 
    is also a Nash equilibrium of the game $\Gcal$ with perfectly rational users. 
\end{theorem}

Taken together, the above results provide some intuition regarding the provider utilities and user value in equilibrium. 
However, they do not elucidate to what extent we can expect the actions of providers in equilibrium 
to be optimal in terms of social welfare.
To shed light on this question, in what follows, we will analyze the \emph{price of anarchy} (PoA), a standard game-theoretic measure~\citep{koutsoupias1999worst,Nisan2007-lx} that compares the social welfare achieved by the test-time compute levels $\thetab^\dagger$ selected by the providers in equilibrium against the 
test-time compute levels that maximize the social welfare of the market.
More formally, the price of anarchy of a test-time compute game $\Gcal$ is given by
\begin{equation}\label{eq:PoA}
    \text{PoA}\left(\Gcal\right) \coloneqq
    \frac{ \max_{\theta \in \Theta^{N}} \SW(\thetab) }{\SW(\thetab^\dagger)}. 
\end{equation}
%
By definition, the price of anarchy is always at least $1$, and it equals $1$ only when the compute levels $\thetab^\dagger$ selected by providers at equilibrium also optimize the social welfare. However, this is unlikely to happen, because rational providers act to maximize their individual utilities rather than coordinating to maximize the social welfare of the market.
To better understand what conditions can lead to $\text{PoA}(\Gcal) > 1$, and since explicitly characterizing the 
test-time compute level $\argmax_{\theta \in \Theta^{N}} \SW(\thetab)$ in an arbitrary game $\Gcal$ is generally intractable, we compare the test-time compute $\thetab^\dagger$ selected by the providers at equilibrium against the test-time compute $\thetab^{*}$ selected by the providers in an idealized situation in which they maximize their individual contribution $\SW_i(\theta_i^*)$ to the welfare.
By doing so, we obtain the following lower-bound on the price of anarchy:\footnote{For the case of perfect rationality ($\beta = \infty$), we adopt the standard convention $e^{-\infty}\coloneqq 0$.}

\begin{theorem}\label{thm:poa}
    The price of anarchy of the test-time compute game $\Gcal$ is lower-bounded, up to higher order terms, as
\begin{equation*}\label{eq:poa-bound}
\textnormal{PoA}\left(\Gcal\right)
    \geq 1 + \frac{ \textnormal{W}_{\pi^*(1)}(\theta^*_{\pi^*(1)}) - \textnormal{W}_{\pi(1)}(\theta_{\pi(1)}^\dagger)}{\textnormal{W}_{\pi^*(1)}(\theta^*_{\pi^*(1)})} + f(\beta)
    ,
\end{equation*}
where $f(\beta) = \Ocal\left(e^{-\beta \Delta V} \right)$, $\Delta V>0$ is a constant depending on the game instance, and $\pi$, $\pi^*$ are permutations of providers in decreasing order of the value they offer at $\thetab^\dagger$ and $\thetab^*$, respectively. Refer to the proof in Appendix~\ref{app:proof-poa} for the exact expression of $f$, $\Delta V$, and the higher order corrections to the above lower bound.
\end{theorem}

The first term in the lower-bound of the price of anarchy given by Theorem~\ref{thm:poa} quantifies the difference between the contributions to the social welfare by the providers who offer the highest value in an idealized situation in which they select a test-time compute level that maximizes their individual contributions to social welfare (\ie, $\thetab^*$) and in equilibrium (\ie, $\thetab^\dagger$), while the term $f(\beta)$ quantifies the effect of users' bounded rationality, which vanishes exponentially fast as users become fully rational ($\beta \to \infty$).

%
Although Theorem~\ref{thm:poa} illustrates the fact that competition between providers can lead to levels of test-time compute that are inefficient in terms of social welfare, it is worth highlighting that, even in the absence of such inefficiencies, there is 
still 
room for the social welfare to improve further. 
Specifically, as can be seen in Eq.~\ref{eq:social-welfare-definition}, the social welfare depends not only on the providers' test-time compute choices, but also on the way market shares are allocated across providers, which is determined by the providers' prices and the way users select which provider serves their queries (see Table~\ref{tab:market-models}).
In an ideal scenario, each provider would select the level of test-time compute $\theta^*_i$ that maximizes their individual contribution $W_i(\theta^*_i)$ to social welfare, and users would select the provider delivering the highest such contribution.
However, even under perfect rationality, this ideal scenario may not be achievable because the provider delivering the highest contribution $W_i(\theta^*_i)$ may not be the one delivering the highest value to the users.
In the next section, we conceptualize and analyze an alternative forward-looking market that realizes such an ideal scenario. 

\section{An Auction Mechanism for Test-Time Compute}
\label{sec:auction}
To design a market in which social welfare is maximized, we draw inspiration from mechanism design~\citep{NISAN2001166,krishna2009auction} and propose a reverse second-price auction that, by design, incentivizes each provider to choose the level of test-time compute $\theta_i^*$ that maximizes their individual contribution $W_i(\theta_{i}^{*})$ to social welfare.
Within this auction, a user takes the role of a seller who provides a task they would like to be served and providers take the role of bidders who declare a quality and price they can offer for serving the user's task based on the representative queries.
Moreover, the auction is conducted by a third-party platform that allocates the user's task to the provider who offers the best quality versus price trade-off and determines the payment of the provider based on the stated qualities and prices of all providers.

Formally, we denote the game induced by the aforementioned auction as $\widetilde{\Gcal}$ and, for simplicity of notation, we overload the notation used in Sections~\ref{sec:model} and~\ref{sec:equilibrium}, with the understanding that the corresponding definitions may actually differ.
The process begins with a user submitting a (small) set of queries $\Qcal$ representative of a task $\Tcal$ they wish to solve to the third-party platform conducting the auction.
All $N$ providers then bid simultaneously for the opportunity to serve the task $\Tcal$.
Specifically, each provider $i$ selects a level of test-time compute $\theta_i$ and uses the set $\Qcal$ of queries to estimate the average quality $q_i(\theta_i)\in\RR_+$ of their model on the task $\Tcal$.\footnote{In many domains, such as coding or information retrieval, estimating the quality $q_i(\theta_i)$ can be done automatically with limited overhead. We further discuss this in Section~\ref{sec:discussion}.} Then, each provider selects a price $p_i\in\RR_+$ at which they are willing to serve the task, and submits a bid $(q_i(\theta_i), p_i)$ to the platform, as is common in scoring auctions~\citep{Che1993,AskerScore}.\footnote{Given a selected test-time compute level $\theta_i$, providers have no incentive to misreport their model's quality $q_i(\theta_i)$, since the platform can verify the quality of their responses, as we formally show in Appendix~\ref{app:auction-verification}.}
%
Finally, the platform assigns the task $\Tcal$ to the provider who offers the highest value $V_i(\theta_i, p_i) \coloneqq q_i(\theta_i) - p_i$ where, similarly to Section~\ref{sec:equilibrium}, we assume there are no ties between providers.

To align the incentives of the providers with social welfare, the third-party platform sets the payment made by the user to the winning provider using a second-price payment rule---the user pays an amount that does not depend on the price $p_{\pi(1)}$ bid by the winning provider, but instead equals the difference between the quality delivered by the winning provider and the value offered by the second-best provider.
Formally, the payment is given by
\begin{equation}\label{eq:second-price-payment}
    P(\thetab,\mathbf{p}) = q_{\pi(1)}(\theta_{\pi(1)}) - V_{\pi(2)}(\theta_{\pi(2)}, p_{\pi(2)}),
\end{equation}
where $\pi$ is the ordering of providers
in decreasing order of the value they offer, and $\thetab$ and $\mathbf{p}$ denote the vectors of test-time compute levels and prices selected by all providers, respectively.
Based on the payment in Eq.~\ref{eq:second-price-payment}, the utility obtained by each provider is then given by
\begin{equation}\label{eq:utility-auction}
    U_i(\theta_i, p_i; \thetab_{-i}, \mathbf{p}_{-i}) = \left( P(\thetab,\mathbf{p}) - c_i(\theta_i) \right) \cdot \mathds{1}\left\{ V_i (\theta_i,p_i) > \max_{j\neq i} V_j(\theta_j,p_j) \right\},
\end{equation}
where $c_i(\theta) \in \mathbb{R}_{+}$ is the generation cost, and the utility of all providers except for the winning one equals zero.

Similar to the competitive market setting in Section~\ref{sec:model}, we define social welfare as the sum of the total utility obtained by all providers and the net 
value obtained by the user, which is the difference between the quality offered by the winning provider and the payment set by the third-party platform.
More concretely, since the only provider with non-zero utility is the winning one, social welfare is given by
\begin{multline*}\label{eq:SW-auction}
    \SW(\thetab,\mathbf{p})
    = \underbrace{U_{\pi(1)}(\theta_{{\pi(1)}}, p_{{\pi(1)}}; \thetab_{-{\pi(1)}}, \mathbf{p}_{-{\pi(1)}})}_{\text{Winning provider's utility}}
    + \underbrace{q_{\pi(1)}(\theta_{\pi(1)}) - P(\thetab, \mathbf{p})}_{\text{Net user value}} \\
    = q_{\pi(1)}(\theta_{\pi(1)}) - c_{\pi(1)}(\theta_{\pi(1)})
    = W_{\pi(1)}\left(\theta_{\pi(1)}\right),
\end{multline*}
where the second equality follows from Eqs.~\ref{eq:second-price-payment} and ~\ref{eq:utility-auction} and highlights that, due to the structure of the payment rule, the social welfare depends entirely on the quality, generation cost, and level of test-time compute of the winning provider $\pi(1)$.

Next, we proceed to analyze the equilibrium properties of the game $\widetilde{G}$ induced by the auction.
Our starting point is the observation that, under the payment rule given by Eq.~\ref{eq:second-price-payment}, each provider $i$ has a dominant strategy, that is, a combination of test-time compute $\theta^\dagger_i$ and price $p^\dagger_i$ 
that maximizes their utility regardless of how the other providers act.
Formally, a dominant strategy for provider $i$ satisfies
\begin{equation}U_i(\theta_i^\dagger, p_i^\dagger; \thetab_{-i}, \mathbf{p}_{-i}) \geq U_i(\theta_i, p_i; \thetab_{-i}, \mathbf{p}_{-i}),\quad \text{for all}\quad (\theta_i,\thetab_{-i}) \in \Theta^{N}\quad \text{and}\quad  (p_i,\mathbf{p}_{-i}) \in \RR_+^{N}
\end{equation}
and, consequently, the choice of test-time compute levels $\thetab^\dagger$ and prices $\mathbf{p}^\dagger$ is a pure Nash equilibrium of $\widetilde{\Gcal}$.
The following theorem shows that such strategies (and hence, equilibria) always exist and characterizes their associated level of test-time compute and price:
\begin{theorem}\label{thm:second-price-dominant}
    In the game $\widetilde{\Gcal}$, the choice of test-time compute level $\theta^\dagger_i = \theta^*_i = \arg\max_{\theta\in\Theta} W_i(\theta)$ and price $p^\dagger_i = c_i\left(\theta^*_i\right)$ is a dominant strategy for provider $i$.
\end{theorem}
The above result follows from the fact that, conditioned on winning the auction, the payment defined in Eq.~\ref{eq:second-price-payment} does not depend on the price $p_{\pi(1)}$. Thus, a provider can simultaneously increase their utility and the value they offer to the user by selecting the compute level $\theta_i^*$ to maximize $V_i(\theta_i, c_i(\theta_i)) = \SW_i(\theta_i)$, and bidding the corresponding quality $q_i(\theta_i^*)$ and the lowest possible price---their generation cost $c_i$.

As an immediate consequence, in the Nash equilibrium $(\thetab^\dagger, \mathbf{p}^\dagger)$ of the game $\widetilde{\Gcal}$, the value $V_i(\theta_i, p_i)$ offered by each provider $i$ is equal to the maximum individual contribution to social welfare $\SW_i(\theta_i^*)$ they can make.
Since the platform selects the provider who offers the highest value to serve the task $\Tcal$, in equilibrium, the winning provider is the one who can make the highest individual contribution to social welfare, \ie, $\pi_{(1)} = \arg\max_{i\in[N]} \SW_i(\theta_i^*)$.
Then, the price of anarchy of the game $\widetilde{\Gcal}$ satisfies
\begin{equation}\label{eq:auction-efficient}
\textnormal{PoA}\left(\widetilde{\Gcal}\right) \coloneqq \frac{\max_{\thetab, \mathbf{p}} \textnormal{W}(\thetab, \mathbf{p})}{\textnormal{W}(\thetab^\dagger, \mathbf{p}^\dagger)}
\stackrel{(*)}{=} \frac{\max_i \max_{\theta_i \in \Theta} \SW_i(\theta_i)}{\SW_{\pi(1)}(\theta_{\pi(1)}^*)}
= 1,
\end{equation}
where $(*)$ holds because it is possible to 
set the prices $\mathbf{p}$ such that any provider wins the auction.
Together, the above results show that, at equilibrium, the game $\widetilde{\Gcal}$ achieves the maximum possible value of social welfare.

In this context, it is worth noting that, due to the choice of payment rule in Eq.~\ref{eq:second-price-payment}, the game $\widetilde{\Gcal}$ enjoys several additional desirable properties.
First, one can easily verify that the winning provider receives a payment strictly higher than the price they bid (\ie, $P(\thetab,\mathbf{p})>p_{\pi(1)}$), and therefore always obtains positive utility.
Second, the value received by the user matches the value offered by the second-best provider, \ie, $q_{\pi(1)}(\theta_{\pi(1)}) - P(\thetab,\mathbf{p}) =  V_{\pi(2)}(\theta_{\pi(2)}, p_{\pi(2)})$.
We further explore these properties, as well as their implications for users and providers, in our experiments in Section~\ref{sec:experiments}.

\section{Experiments}
\label{sec:experiments}
In this section, we conduct experiments to analyze the outcome of test-time compute games. Here, our goal is twofold: first, to validate our theoretical results in Section~\ref{sec:equilibrium} characterizing the dynamics and equilibria of test-time compute games; and second, to empirically demonstrate that, under realistic conditions, current test-time compute markets do not, in general, maximize social welfare. We begin by describing our experimental setup.\footnote{Refer to Appendix~\ref{app:experimental-details} for further details about our experimental setup.}

\begin{figure}[ht]
    \centering
    \begin{subfigure}{0.45\textwidth}
        \centering
        \includegraphics[width=0.95\linewidth]{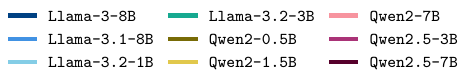} \\
        \vspace{0.1cm}
        \includegraphics[width=\linewidth]{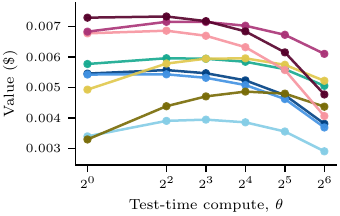}
        \caption{Non-reasoning models}
        \label{fig:main-values:a}
    \end{subfigure}
    \qquad
    \begin{subfigure}{0.45\textwidth}
        \centering
        \includegraphics[width=0.95\linewidth]{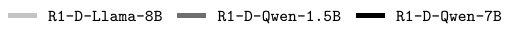}
        \vspace{0.1cm}
        \includegraphics[width=\linewidth]{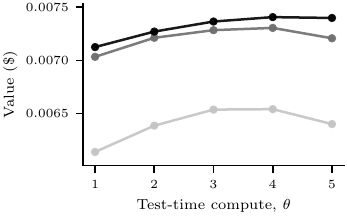}
        \caption{Reasoning models}
        \label{fig:main-values:b}
    \end{subfigure}
    \caption{{\bf User values offered by providers in a test-time compute game.}
     The figure shows the user values $V_i(\theta)$ offered by providers in a test-time compute game $\Gcal$, as a function of their test-time compute $\theta$.
     Panel~(a) corresponds to a test-time compute game with $N=9$ providers serving non-reasoning models from the \texttt{Llama} and \texttt{Qwen} families, where providers use best-of-n across $\theta$ samples as their test-time compute method. Panel (b) corresponds to a game with $N=3$ providers serving reasoning models distilled from \texttt{DeepSeek-R1}, where $\theta$ represents reasoning effort, defined by binning the model outputs into quantiles based on the number of reasoning tokens (see Appendix~\ref{app:experimental-details}).
     In both games, providers serve queries $Q$ from the \texttt{GSM8K} dataset, and we consider that each (average) percentage point of accuracy offers a value of $\$0.008$ to the users.
    }
    \label{fig:main-values}
\end{figure}

\xhdr{Experimental setup}
We consider test-time compute games, as defined in Section~\ref{sec:model}, in which each provider serves a different LLM, which we use to identify providers. More concretely, we study games $\Gcal$ with either: (i) $N=9$ providers, each serving a non-reasoning model selected from one of two families---four models from the \texttt{Llama} family, \texttt{\seqsplit{Llama-\{3-8B,3.1-8B,3.2-1B,3.2-3B\}-Instruct}}, and five models from the \texttt{Qwen} family, \texttt{\seqsplit{Qwen-\{2-0.5B,2-1.5B,2-7B,2.5-3B,2.5-7B\}-Instruct}}---or (ii) $N=3$ providers serving reasoning models distilled from \texttt{DeepSeek-R1}, namely \texttt{\seqsplit{DeepSeek-R1-Distill-\{Llama-8B, Qwen-1.5B, Qwen-7B\}}}.

In each game, the task $\Tcal$ that providers compete to serve consists of queries drawn from one of four widely used benchmarking datasets, with each dataset defining a distinct game: \texttt{GSM8K}~\citep{cobbe2021training}, \texttt{GPQA}~\citep{rein2023gpqagraduatelevelgoogleproofqa} and \texttt{AIME}~\citep{aime25}, which consist of STEM questions with verifiable ground-truth answers, and \texttt{TruthfulQA}~\citep{lin2022truthfulqameasuringmodelsmimic}, which evaluates models on their ability to generate truthful answers to questions reflecting common misconceptions.

To serve these tasks, all providers in a game use the same test-time compute method. Specifically, we consider either best-of-n~\citep{chow2024inference} or majority voting in games where providers serve non-reasoning models, and chain-of-thought~\citep{wei2022chain} in games where providers serve reasoning models. For best-of-n, providers can generate $\theta\in\{2^0, \ldots, 2^6\}$ independent model outputs for a given query, and they select as the final response the highest-scored output according to a fixed reward model, namely \texttt{ArmoRM-Llama3-8B-v0.1}.
Similarly, for majority voting, providers can generate $\theta\in\{2^0, \ldots, 2^6\}$ model outputs for a given query, and take as the final response the most frequent response among the generated outputs. In contrast, for chain-of-thought, since the reasoning models we consider simply generate a sequence of reasoning tokens but do not allow explicit control over how much test-time compute is used, we define discrete compute levels by generating $32$ outputs per question, and binning (for each question) the outputs into $5$ quantiles based on the number of reasoning tokens in the output. Then, we consider the outputs of each of the $5$ bins as being generated using a different test-time compute level $\theta$. Since \texttt{TruthfulQA} consists of open-ended questions where majority voting is not directly applicable, and involves short responses that do not benefit from extended chain-of-thought reasoning, we restrict the experiments on this dataset to non-reasoning models using best-of-n.

Lastly, to completely specify the game $\Gcal$, we need to define how the providers' compute choices determine both the value $V_i(\theta)$ they offer to users, and their own utilities $U_i(\thetab)$. To this end, as we report in Appendix~\ref{app:model-evaluation}, we first measure the average output quality and the average number of generated tokens of all twelve models across all datasets and test-time compute levels. More precisely, the output quality simply corresponds to the accuracy for \texttt{GSM8K}, \texttt{GPQA}, and \texttt{AIME}, whereas for the open-ended benchmark \texttt{TruthfulQA}, we use a \texttt{BLEURT}-based score~\citep{sellam2020bleurtlearningrobustmetrics} defined as the difference between the similarity to a true reference answer and the similarity to a false reference answer. Based on these measurements, we compute, for each provider and compute level $\theta$, the prices $p_i(\theta)$ by multiplying the average number of generated tokens per query by the average per-token price. These token prices are obtained from the Hugging Face list of providers (see Appendix~\ref{app:experimental-details}), where typical prices are on the order of $\sim$$\$0.1$ per million tokens. Then, we infer the provider costs $c_i(\theta)$ by assuming a fixed profit margin of $25\%$, defined as the percentage increase of the per-token price over the per-token cost, and for simplicity, we set this margin to be the same across providers. Finally, we linearly map model qualities into user values by considering that each (average) percentage point of accuracy is worth $\{\$0.008, \$0.02, \$0.05, \$0.008 \}$, respectively, for \texttt{GSM8K}, \texttt{GPQA}, \texttt{AIME}, and \texttt{TruthfulQA}. This choice is motivated by assigning a higher monetary value to quality, the harder the dataset is; as seen in Appendix~\ref{app:model-evaluation}, models typically perform best on \texttt{GSM8K} and \texttt{TruthfulQA} and worst on \texttt{GPQA} and \texttt{AIME}.

\begin{figure}[t]
    \centering
    
    \begin{subfigure}{0.45\textwidth}
        \centering
        \includegraphics[width=0.95\linewidth]{figures/var/legend_unreasoning.pdf} \\
        \vspace{0.1cm}
        \includegraphics[width=\linewidth]{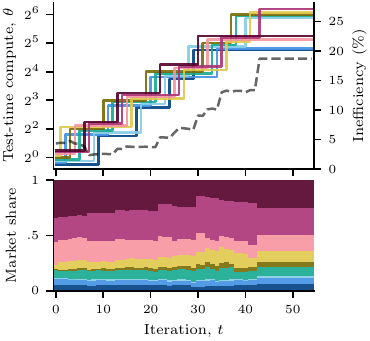}
        \caption{Non-reasoning models}
        \label{fig:main-dynamics:a}
    \end{subfigure}
    \qquad
    \begin{subfigure}{0.45\textwidth}
        \centering
        \includegraphics[width=0.95\linewidth]{figures/var/legend_reasoning.pdf}
        \vspace{0.1cm}
        \includegraphics[width=\linewidth]{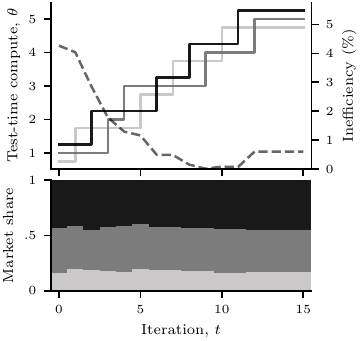}
        \caption{Reasoning models}
        \label{fig:main-dynamics:b}
    \end{subfigure}
    \caption{{\bf Dynamics of a test-time compute game.}
     The figure illustrates the better-response dynamics of a test-time compute game $\Gcal$ when providers sequentially select a test-time compute level that increases their utility. The upper panels show the compute levels $\theta^t$ selected by each provider and the resulting market inefficiency (dashed black curve), defined as $\max_\thetab \SW(\thetab)/\SW(\thetab^t)$, and the lower panels show the market share of each provider.
     Panel~(a) and Panel~(b) correspond to the games described in Figure~\ref{fig:main-values}, where providers serve, respectively, non-reasoning from the \texttt{Llama} and \texttt{Qwen} families, and reasoning models distilled from \texttt{DeepSeek-R1}. Here, we set $\beta=1000$.
}
    \label{fig:main-dynamics}
\end{figure}

\xhdr{Results}
For concreteness, here we focus on games where providers serve queries from \texttt{GSM8K} using best-of-n for non-reasoning models and chain-of-thought for reasoning models. We find qualitatively similar results across different levels of user rationality $\beta$, other combinations of test-time compute methods and datasets, and a range of provider profit margins and accuracy-to-value conversion parameter values (see Appendix~\ref{app:dynamics-results}).

Figure~\ref{fig:main-values} shows that the values offered by providers are approximately concave with respect to the compute. Specifically, users initially benefit from a slight increase in compute; however, as shown in Figure~\ref{fig:gsm8k_reasoning_models} in Appendix~\ref{app:model-eval-gsm8k}, the number of tokens---and thus the price paid by the user---increases rapidly with compute, eventually diminishing their value. In games with non-reasoning models (Figure~\ref{fig:main-values:a}), most providers offer comparable values to users, and no provider offers a value higher than their competitors simultaneously for all computes. In contrast, for reasoning models (Figure~\ref{fig:main-values:b}), the provider serving \texttt{R1-D-Qwen-7B} offers a higher value than their competitors across all computes.

To visualize how these differences in offered value translate into specific market outcomes, we show in Figure~\ref{fig:main-dynamics} the dynamics of a test-time compute game with boundedly rational users ($\beta=1000$), where providers sequentially adjust their compute levels to increase their own utility. Consistent with the value versus compute curves, when providers serve non-reasoning models, \texttt{Qwen2.5-7B} and \texttt{Qwen2.5-3B} have the majority of the market share, while \texttt{R1-D-Qwen-7B} and \texttt{R1-D-Qwen-1.5B} have the majority of the market share when providers serve reasoning models.
%
Furthermore, in accordance with Theorem~\ref{thm:potential-game}, since $\Gcal$ is a potential game, the better-response dynamics converge to a pure Nash equilibrium. This occurs after $t=44$ and $t=12$ iterations when serving non-reasoning models (Figure~\ref{fig:main-dynamics:a}) and when serving reasoning models (Figure~\ref{fig:main-dynamics:b}), respectively. 
Importantly, we find that, once providers have reached an equilibrium, the market is socially inefficient. In particular, when providers serve non-reasoning models, the price of anarchy is \mbox{$ 1.19$}, representing a $16\%$ loss in social welfare at the equilibrium. For reasoning models, the inefficiency is approximately $0.6\%$. However, it is noteworthy that even in this latter case---which features a more constrained strategy space with only $N=3$ providers---social welfare is not maximized. Altogether, our result suggests that competition in test-time compute markets does not naturally align with the social optimum, regardless of the number of competing providers.

Lastly, in Table~\ref{tab:auction-main}, we compare the outcome of the above test-time compute games $\Gcal$ with the outcomes if the providers and users had instead run the second-price auction described in Section~\ref{sec:auction}. We find that, compared to the test-time compute game $\Gcal$, the auction mechanism $\widetilde{\Gcal}$ increases the value of users by $29\%$ and $2.8\%$, and the social welfare by $25\%$ and $4\%$, respectively, for non-reasoning models and reasoning models. Note that as shown in Figure~\ref{fig:main-dynamics}, we have $\text{PoA}(\Gcal)>1$, while it is guaranteed that $\text{PoA}(\widetilde{\Gcal})=1$ (Eq~\ref{eq:auction-efficient}). However, we also find that, under the auction mechanism, providers can either decrease the utility ($-28\%$ for non-reasoning models) or increase it ($+200\%$ for reasoning models). This effect is mainly driven by the fact that the second-price payment in Eq.~\ref{eq:second-price-payment} depends on how competitive the market is, which varies across game instances (see Figure~\ref{fig:main-values}).

\begin{table}[H]
\centering
\caption{\textbf{Comparison of the equilibrium between $\Gcal$ and the auction $\widetilde{\Gcal}$.}
The table compares the equilibrium outcomes of the test-time compute games $\Gcal$ ($\beta=1000$) in Figure~\ref{fig:main-values} with the outcomes of the auction $\widetilde{\Gcal}$ for the same set of providers using best-of-n with non-reasoning models, and chain-of-thought with reasoning models.
For $\Gcal$, the table shows the providers' average utility, the users' average value and price, all weighted by their equilibrium market shares, together with the social welfare. 
For $\widetilde{\Gcal}$, where the auction winner is \texttt{Qwen2.5-7B}, followed by \texttt{Qwen2.5-3B} and
\texttt{R1-D-Qwen-7B,} followed by \texttt{R1-D-Qwen-1.5B}, respectively for non-reasoning and reasoning models, the table shows the value received by the user (once the auction is conducted), the price they pay (according to Eq.~\ref{eq:second-price-payment}), the provider's utility, and the social welfare.
All quantities (except the compute level) have units $(\$ \times 10^{-3})$. 
}
\begin{tabular}{lcccc}
\toprule
{} & \multicolumn{2}{c}{Non-reasoning} & \multicolumn{2}{c}{Reasoning} \\
\cmidrule(lr){2-3} \cmidrule(lr){4-5}
{} & Game $\Gcal$  & Auction $\widetilde{\Gcal}$  & Game $\Gcal$ & Auction $\widetilde{\Gcal}$ \\
\midrule
Compute level ($\theta$)   &  16    &   2     &   4    &   4  \\
User value           & 5.5  & 7.1   & 7.1  & 7.3 \\
Price                & 1.2  & 0.32   & 0.18 & 0.30    \\
Provider(s') utility & 0.25 & 0.18 & 0.04 & 0.12    \\
Social welfare       & 5.8  & 7.3   & 7.1  & 7.4    \\
\bottomrule
\end{tabular}
\label{tab:auction-main}
\end{table}


\section{Discussion and Limitations}
\label{sec:discussion}
In this section, we highlight several limitations of our work and discuss avenues for future research.

\xhdr{Model}
Our definition of test-time compute games considers only factors related to the quality of LLM outputs and financial aspects, such as the pricing of LLMs by providers and their generation costs.
In practice, however, other factors may influence users' choices among providers, leading to more complex behavior than described by the canonical market share functions in Table~\ref{tab:market-models}. These could include, for example, a provider's popularity or the support and reliability of their API service.
Consequently, it would be interesting to extend our definition of test-time compute games to account for these additional factors and develop auction mechanisms for such (complex) games.
Moreover, in our work, we assume that all users derive the same (average) value from a given LLM and test-time compute level; however, in practice, different users may value output quality differently, even for the same query and output. 
Extending test-time compute games to heterogeneous users is an interesting direction for future work.

\xhdr{Evaluation}
We have studied the equilibria of test-time compute games on simulated markets using several state-of-the-art LLMs, including both reasoning and non-reasoning models, and test-time compute methods.
In each game, a set of simulated providers competes to serve queries drawn from a benchmarking dataset, either with verifiable ground truth (as in \texttt{GSM8K}, \texttt{GPQA}, and \texttt{AIME}) or with reference correct and incorrect answers, as in the open-ended \texttt{TruthfulQA}.
However, it would be interesting to study the equilibria of test-time compute games in settings where providers compete to serve open-ended queries and quality is measured automatically without relying on predefined correct outputs.
Further, it would be important to analyze the dynamics of the LLM-as-a-service market using real-world data, and deploy and evaluate our proposed auction mechanism in a real-world deployment.

\xhdr{Practical considerations}
Our second-price auction requires that providers accurately estimate the (average) quality that the model they serve can deliver for a given task.
In many domains where LLMs are currently used, such as coding or information retrieval, automated evaluation procedures make this feasible with limited overhead, especially when user queries overlap with standardized benchmarks on which providers have already evaluated their models. At the same time, the platform implementing the auction can leverage these procedures to verify and evaluate the quality of a provider’s responses. As we show in Appendix~\ref{app:auction-verification}, our auction can be naturally extended with such evaluations to ensure that providers maximize their utility by truthfully reporting their quality estimates.
Looking ahead, the feasibility of an auction in settings with less well-established quality metrics and benchmarks will depend on continued progress in the LLM evaluation literature.
%
%
%
%
Lastly, we note that, while second-price auctions possess many desirable properties, they also make a provider’s payment dependent on the bids of other providers, leaving the mechanism potentially vulnerable to manipulation or collusion~\citep{Hendricks1989, MAILATH1991467, CHE2018398}.

\xhdr{Adoption incentives}
Our empirical results in Appendix~\ref{app:auction-tables} suggest that, while the auction mechanism introduced in Section~\ref{sec:auction} is guaranteed to increase social welfare, in some settings this may come at the expense of provider utility. However, the same results also show that users typically obtain higher value when participating in the auction, especially compared to settings in which they have limited information about the value offered by providers (\ie, when the parameter $\beta$ is low). As a result, users may be more likely to participate in a platform implementing the auction, which could, in turn, incentivize provider participation, as it offers access to a larger share of the user base.
%


\section{Conclusions}
\label{sec:conclusions}
In this work, we introduced \emph{test-time compute games}, a game-theoretic model in which LLM providers in a market of LLMs-as-a-service strategically select the level of test-time compute used by their model to compete for user queries and maximize their profit.
Based on this model, we have demonstrated, first theoretically and then empirically, that current pricing approaches in markets of LLMs-as-a-service incentivize providers to set their level of test-time compute in a way that is not aligned with social welfare, as it does not optimally balance the tradeoff between generation cost and output quality.
%
%
To address this, we proposed a forward-looking market based on a reverse second-price auction that provably aligns provider incentives with social welfare.
We hope that our work inspires users, LLM providers, and online platforms to explore alternative market structures for generative AI that place social welfare at the center of their design choices.

\vspace{2mm}
\xhdr{Acknowledgements} 
Gomez-Rodriguez acknowledges support from the European Research Council (ERC) under the European Union'{}s Horizon 2020 research and innovation programme (grant agreement No. 101169607).
Tsirtsis acknowledges support from the Alexander von Humboldt Foundation in the framework of the Alexander von Humboldt Professorship (Humboldt Professor of Technology and Regulation awarded to Sandra Wachter) endowed by the Federal Ministry of Education and Research via the Hasso Plattner Institute.

{ 
\small
\bibliographystyle{unsrtnat}
\bibliography{llm-inference-incentives}
}

\clearpage
\newpage

\appendix

\section{Additional Details for the Introductory Example}\label{app:example}

Here, we provide additional details regarding the introductory example in Section~\ref{sec:intro}. Therein, two LLM providers compete to serve a user on a fixed task and may each choose between a low- and a high-TTC mode for their respective models. The resulting average output accuracy and generation cost associated with each mode are given by:

\begin{table}[H]
\centering

\begin{minipage}{0.49\textwidth}
\centering
\textbf{Provider 1}\\[4pt]
\begin{tabular}{lcc}
\toprule
{}        & Avg.\ accuracy & Avg.\ gen. cost \\
\midrule
Low TTC   & 70\%           & \$0.25          \\
High TTC  & 90\%           & \$1             \\
\bottomrule
\end{tabular}
\end{minipage}
\hfill
\begin{minipage}{0.49\textwidth}
\centering
\textbf{Provider 2}\\[4pt]
\begin{tabular}{lcc}
\toprule
{}        & Avg.\ accuracy & Avg.\ gen. cost \\
\midrule
Low TTC   & 50\%           & \$0.5           \\
High TTC  & 95\%           & \$10            \\
\bottomrule
\end{tabular}
\end{minipage}

\end{table}

Suppose that providers have fixed their prices to obtain a margin of $25\%$ over their generation costs for both low and high TTC modes, and that the user values each percentage of accuracy as $\$ 0.02$. Then, using the notation in Section~\ref{sec:model}, for each compute mode $\theta\in\{\textnormal{Low}, \textnormal{High} \}$ used by each provider $i=1,2$, the value they offer to the user is
\begin{equation}
    V_i(\theta) = \underbrace{\$0.02 \times a_i(\theta)}_{q_i(\theta)} - \underbrace{1.25 \times c_i(\theta)}_{p_i(\theta)},
\end{equation}
where $a_i(\theta)\in[0,100]$ are the average accuracy when selecting compute $\theta$, and $c_i$ are the generation costs. Using the specific values in the above table, we obtain:

\begin{equation}
    \begin{dcases}
        V_1(\textnormal{Low}) = \$1.0875\\
        V_1(\textnormal{High}) = \$0.55
    \end{dcases}
    \quad \textnormal{and} \quad
        \begin{dcases}
        V_2(\textnormal{Low}) = \$0.375\\
        V_2(\textnormal{High}) = -\$10.6.
    \end{dcases} 
\end{equation}
Consequently, a perfectly rational user would select the first provider to serve the queries, no matter which compute level they use, and the market (for this particular task) is monopolistically dominated by the first provider. The total user value and provider utility of this market (\ie, the social welfare $\SW$ in Section~\ref{sec:equilibrium}) is, depending on the compute level selected by the first provider:
\begin{equation}
    \begin{dcases}
        \SW(\theta_1 = \textnormal{Low},\theta_2) = V_1(\textnormal{Low}) + p_1(\textnormal{Low}) - c_1(\textnormal{Low}) = \$ 1.15\\
        \SW(\theta_1 = \textnormal{High},\theta_2) = V_1(\textnormal{High}) + p_1(\textnormal{High}) - c_1(\textnormal{High}) = \$ 0.8,
    \end{dcases}
\end{equation}
irrespectively of the compute level $\theta_2$ of the second provider. That is, the low compute mode is socially optimal. However, the utility obtained by the first provider is
\begin{equation}
   \begin{dcases}
        U_1(\textnormal{Low};\theta_2) = 0.25\times \$0.25=\$0.0625\\
        U_1(\textnormal{High};\theta_2) = 0.25\times \$1=\$0.25,
   \end{dcases}
\end{equation}
which strictly incentivizes the provider to select the high-TTC mode despite its lower social welfare. In this stylized example, the price of anarchy is:
\begin{equation}
    \text{PoA} = \frac{\$1.15}{\$0.8} = 1.4375.
\end{equation}

\clearpage
\newpage

\section{Auction with Quality Verification}\label{app:auction-verification}

In this section, we show that the auction mechanism introduced in Section~\ref{sec:auction} can be naturally extended to allow the platform to verify the average quality of the models deployed by providers. As we will now demonstrate, this simple extension incentivizes providers to truthfully report their quality estimates.

Consider an extension of the auction $\widetilde{\Gcal}$ that proceeds in multiple stages as follows (we overload notation for simplicity and also refer to it as $\widetilde{\Gcal}$). A total of $N$ providers bid simultaneously for the opportunity to serve a task $\Tcal$ for a user.  The auction begins with the user submitting to the platform a small set $\Qcal$ of queries representative of the task $\Tcal$, which are then shared with all providers. Providers can use this information to estimate the quality of their models and determine the price they wish to offer (\eg, by directly evaluating their model on $\Qcal$, using private historical data, \etc).\footnote{As in Section~\ref{sec:auction}, quality refers to any agreed-upon metric between the user, the platform, and the providers, such as pass@k for coding tasks.}

In the first stage of the auction, each provider $i$ submits a bid consisting of a pair $(\tilde{q}_i, p_i)$ to the platform, where $\tilde{q}_i$ is a reported estimate of their model quality on $\Tcal$ and $p_i$ is their bid price. The platform assigns the task $\Tcal$ to the provider offering the highest value $V_i = \tilde{q}_i - p_i$. 
In the second stage, the winning provider, denoted by $\pi(1)$, selects a test-time compute level $\theta_{\pi(1)}$ (which need not be known to the platform or the user) and generates responses to the queries in $\Tcal$. These responses are shared with both the user and the platform. Based on the observed responses, the platform computes an unbiased estimator $\hat{q}_{\pi(1)}$ of the true model quality on task $\Tcal$ at compute level $\theta_{\pi(1)}$, \ie,
\begin{equation}
    \EE[\hat{q}_{\pi(1)} \given \theta_{\pi(1)}] = q_{\pi(1)}(\theta_{\pi(1)}).
\end{equation}
For example, the platform may randomly sample a subset of the responses and evaluate their quality.\footnote{There already exist online platforms that continuously monitor the quality of proprietary models on selected tasks, such as \url{https://marginlab.ai/trackers/claude-code-historical-performance/}.} The platform then sets the payment to be made by the user to the provider for serving task $\Tcal$ as
\begin{equation}
    P(\mathbf{\tilde{q}}, \mathbf{p}) = \hat{q}_{\pi(1)} - V_{\pi(2)},
\end{equation}
where $\mathbf{\tilde{q}}$ and $\mathbf{p}$ denote the vectors of reported qualities and prices, respectively. Consequently, the (average) utility obtained by each provider is simply:
\begin{equation}
    U_i(\tilde{q}_i, \theta_i, p_i; \mathbf{\tilde{q}}_{-i}, \thetab_{-i}, \mathbf{p}_{-i}) 
    = \EE\left[ \left( P(\mathbf{\tilde{q}}, \mathbf{p}) - c_i(\theta_i) \right) 
    \cdot \mathbf{1}\left\{ V_i > \max_{j \neq i} V_j \right\} \Bigg| \theta_{\pi(1)} \right],
\end{equation}
where the expectation is taken over the randomness in the platform's estimator $\hat{q}_{\pi(1)}$, and note that only the provider serving the task obtains non-zero utility. Thus, the auction $\widetilde{\Gcal}$ can be viewed as an extensive-form game in which each provider’s action is a triplet $(\tilde{q}_i, \theta_i, p_i)$ and, analogously to Section~\ref{sec:auction}, the social welfare is
\begin{equation}
    \SW(\mathbf{\tilde{q}}, \thetab, \mathbf{p}) 
    = q_{\pi(1)}(\theta_{\pi(1)}) - c_{\pi(1)}(\theta_{\pi(1)})
    = W_{\pi(1)}(\theta_{\pi(1)}).
\end{equation}

The next result formalizes the fact that, in this extended auction $\widetilde{\Gcal}$, truthfully reporting quality constitutes a (weakly) dominant strategy:

\begin{theorem}\label{thm:second-price-extended-dominant}
In the extended auction $\widetilde{\Gcal}$, the strategy
\begin{equation}
    (\tilde{q}_i, \theta_i, p_i) = (q_i(\theta_i^*), \theta_i^*, c_i(\theta_i^*))
\end{equation}
is weakly dominant for each provider $i$, where $\theta_i^* \in \arg\max_{\theta\in\Theta} \{ q_i(\theta) - c_i(\theta) \}$.
\end{theorem}

\begin{proof}
Fix an extended game $\widetilde{\mathcal{G}}$, a provider $i$, and any compute levels ${\thetab}_{-i}$, bid prices $\mathbf{p}_{-i}$, and bid qualities $\tilde{\mathbf{q}}_{-i}$. We will prove that, for provider $i$, selecting $\theta_i = \theta^*_i$, $p_i = c_i(\theta^*_i)$, and $\tilde{q}_i = q_i(\theta^*_i)$ is a weakly dominant strategy, that is, any deviation does not strictly increase their utility. Let $V^{(1)}_{-i} \geq V^{(2)}_{-i} \geq \cdots$ denote the values offered by the other providers in decreasing order, which are fixed since $(\tilde{\mathbf{q}}_{-i}, \thetab_{-i}, \mathbf{p}_{-i})$ are fixed. Consider any possible deviation $(\tilde{q}_i, \theta_i, p_i)$ for provider $i$, and distinguish the following two cases.

\textbf{Case 1: $V^{(1)}_{-i} > q_i(\theta^*_i) - c_i(\theta^*_i)$.} Under the candidate strategy $(q_i(\theta^*_i), \theta^*_i, c_i(\theta^*_i))$, provider $i$ loses the auction and receives utility zero. If the deviation also loses, the utility is likewise zero. If the deviation wins, the second-best value is $V^{(1)}_{-i}$ and the utility is
\begin{align*}
    U_i(\tilde{q}_i, \theta_i, p_i;\, \tilde{\mathbf{q}}_{-i}, {\thetab}_{-i},
    \mathbf{p}_{-i}) &= \mathbb{E}[\hat{q}_i - c_i(\theta_i) - V^{(1)}_{-i} \mid \theta_i]
    \\
    &= q_i(\theta_i) - c_i(\theta_i) - V^{(1)}_{-i} \\
    &= W_i(\theta_i) - V^{(1)}_{-i} \\
    &\leq W_i(\theta^*_i) - V^{(1)}_{-i} \leq 0,
\end{align*}
where the first inequality uses $W_i(\theta_i) \leq W_i(\theta^*_i)$ by definition of $\theta^*_i$, and the second uses $V^{(1)}_{-i} \geq q_i(\theta^*_i) - c_i(\theta^*_i) = W_i(\theta^*_i)$. Hence, no deviation strictly increases the utility of provider $i$ above zero.

\textbf{Case 2: $V^{(1)}_{-i} < q_i(\theta^*_i) - c_i(\theta^*_i)$.} Under the candidate
strategy, provider $i$ wins and the second-best value is $V^{(1)}_{-i}$, yielding utility
$U_i(q_i(\theta^*_i), \theta^*_i, c_i(\theta^*_i);\, \tilde{\mathbf{q}}_{-i},
\boldsymbol{\theta}_{-i}, \mathbf{p}_{-i}) = W_i(\theta^*_i) - V^{(1)}_{-i} > 0$. If the
deviation loses, the resulting utility is $0 < W_i(\theta^*_i) - V^{(1)}_{-i}$, which is not profitable.
If the deviation wins, the second-best value is again $V^{(1)}_{-i}$ (since the other
providers' strategies are unchanged), and the utility is
\begin{align*}
    U_i(\tilde{q}_i, \theta_i, p_i;\, \tilde{\mathbf{q}}_{-i}, \boldsymbol{\theta}_{-i},
    \mathbf{p}_{-i}) &= \mathbb{E}[\hat{q}_i - c_i(\theta_i) - V^{(1)}_{-i} \mid \theta_i]
    \\
    &= q_i(\theta_i) - c_i(\theta_i) - V^{(1)}_{-i} \\
    &= W_i(\theta_i) - V^{(1)}_{-i} \\
    &\leq W_i(\theta^*_i) - V^{(1)}_{-i} \\
    &= U_i(q_i(\theta^*_i), \theta^*_i, c_i(\theta^*_i);\, \tilde{\mathbf{q}}_{-i},
    \boldsymbol{\theta}_{-i}, \mathbf{p}_{-i}).
\end{align*}
Hence, no deviation strictly increases the utility of provider $i$ above $W_i(\theta^*_i) -
V^{(1)}_{-i}$.

In both cases, no deviation from $(q_i(\theta^*_i), \theta^*_i, c_i(\theta^*_i))$ strictly
increases provider $i$'s utility, so this strategy is weakly dominant.
\end{proof}

\clearpage
\newpage

\section{Deferred Proofs}\label{app:proofs}

\subsection{Proof of Theorem~\ref{thm:potential-game}}\label{app:proof-potential}

We prove Theorem~\ref{thm:potential-game} for the case of perfectly rational ($\beta=\infty$) and boundedly rational ($\beta<\infty$) separately.

\subsubsection{Potential game (perfect rationality)}

We will show that, in the case where users are perfectly rational, the game $\Gcal$ is a generalized ordinal potential game, \ie, there exists a function $\Phi: \Theta^N \to \RR$ such that for any provider $i$:
\begin{equation*}
    U_i\left(\theta'_i ; \thetab_{-i}\right) - U_i\left(\theta_i ; \thetab_{-i}\right) >0 \Rightarrow \Phi\left(\theta'_i ; \thetab_{-i}\right) - \Phi\left(\theta_i ; \thetab_{-i}\right) > 0 \quad \text{for all }\, \thetab\in\Theta^{N}, \theta_i'\in \Theta.
\end{equation*}

Given the test-time compute levels of all providers $\thetab$, let $\pi$ be the permutation ordering providers by the value they offer to users, \ie,

\begin{equation*}
    V_{\pi(1)}(\theta_{\pi(1)}) > V_{\pi(2)}(\theta_{\pi(2)}) > \dots > V_{\pi(N)}(\theta_{\pi(N)}).
\end{equation*}
and denote, for conciseness, $V_{\pi(1)} = V_{\pi(1)}(\theta_{\pi(1)})$ and $V_{\pi(2)} = V_{\pi(2)}(\theta_{\pi(2)})$.
We define the potential function as a weighted sum of the second largest value and the utility gained by the provider $\pi(1)$, \ie,
\begin{align}\label{eq:app-potential-argmax}
    \Phi(\thetab) & = C\cdot V_{\pi(2)} + \log   U_{\pi(1)}(\theta_{\pi(1)} ; \thetab_{-\pi(1)}) \\
    & = C\cdot V_{\pi(2)} + \log \left( p_{\pi(1)}(\theta_{\pi(1)}) - c_{\pi(1)}(\theta_{\pi(1)}) \right),\nonumber
\end{align}
where $C>0$ is a constant whose value we will specify later on.

To prove that the function given by Eq.~\ref{eq:app-potential-argmax} is a generalized ordinal potential function for the game $\Gcal$, we distinguish two cases, depending on whether the provider $i$ who unilaterally changes their test-time compute level to increase their utility (i) is the winning provider under $\thetab$, that is, $i=\pi(1)$, or (ii) is not the winning provider under $\thetab$.

In the first case, based on Eqs.~\ref{eq:increasing_utility},~\ref{eq:provider-utility} and the fact that the provider $i$ strictly increased their utility, it has to hold that $s_i(V_i(\theta_i);\Vb_{-i}(\thetab_{-i})) = s_i(V_i(\theta'_i);\Vb_{-i}(\thetab_{-i})) = 1$ and $\theta_i \prec \theta'_i$.
This implies that the ordering of providers in terms of the value they offer to the users remains the same before and after the change of test-time compute level by provider $i$, and the second highest value under $\thetab$ and $\thetab'$ coincide. Thus,
\begin{align*}
    \Phi\left(\theta'_i ; \thetab_{-i}\right) - \Phi\left(\theta_i ; \thetab_{-i}\right) &= 
    \log U_{\pi(1)}(\theta'_{\pi(1)} ; \thetab_{-\pi(1)}) - \log U_{\pi(1)}(\theta_{\pi(1)} ; \thetab_{-\pi(1)})\\
    &= \log U_{i}(\theta'_{i} ; \thetab_{-i}) - \log U_{i}(\theta_{i} ; \thetab_{-i}) > 0.
\end{align*}

In the second case, it holds that $s_i(V_i(\theta_i);\Vb_{-i}(\thetab_{-i})) = 0$, therefore, the only way for provider $i$ to strictly increase their utility is by offering a higher value to the users than the previous winning provider and becoming the winning provider themselves.
Therefore, under $\thetab'$, the winning provider is $i$ and the provider who offers the second highest value to the users is the previous winning provider, which implies that $V'_{\pi'(2)} = V_{\pi(1)}$, where $\pi'$ orders providers according to the values they offer under $\thetab'$.
%
%
As a result, the difference in the function $\Phi$ is
\begin{equation}\label{eq:potential_diff_helper}
    \Phi\left(\theta'_i ; \thetab_{-i}\right) - \Phi\left(\theta_i ; \thetab_{-i}\right) = C\cdot (V_{\pi(1)} - V_{\pi(2)}) + \log U_i (\theta'_i ; \thetab_{-i}) - \log U_{\pi(1)} (\theta_{\pi(1)} ; \thetab_{-\pi(1)}),
\end{equation}
where $V_{\pi(1)} - V_{\pi(2)} > 0$ because we have assumed no ties in the values providers can offer to the users.
Moreover, since the action space $\Theta$ is finite, there has to exist a value $U_{max} > 0$ such that, for all $j$ and $\thetab$, it holds that $U_j(\theta_j ; \thetab_j) \leq U_{max}$. Similarly, there exists a value $U_{min}>0$ such that for all $j$ and $\thetab$, it holds that if $U_j(\theta_j ; \thetab_j) > 0$, then $U_j(\theta_j ; \thetab_j) > U_{min}$.
Combining this with the fact that provider utilities are non-negative, Eq.~\ref{eq:potential_diff_helper} implies that
\begin{equation*}
    \Phi\left(\theta'_i ; \thetab_{-i}\right) - \Phi\left(\theta_i ; \thetab_{-i}\right) \geq C\cdot (V_{\pi(1)} - V_{\pi(2)}) + \log U_{min} - \log U_{max},
\end{equation*}
and it suffices to find a value for the constant $C$ such that $C>\frac{\log U_{max} - \log U_{min}}{V_{\pi(1)} - V_{\pi(2)}}$.
Since we have assumed no ties in the values that providers can offer to the users, there has to exist a $\delta_{min}$ such that $V_i(\theta_i) - V_{j}(\theta_j) > \delta_{min} $, uniformly across providers $i,j$ and compute levels $\theta_i,\theta_j$.
Setting $C=\frac{\log U_{max} - \log U_{min}}{\delta_{min}}$ concludes the proof.

\subsubsection{Potential game (bounded rationality)}

We will show that, in the case of boundedly rational users, the game $\Gcal$ with $\beta<\infty$ is also a generalized ordinal potential game.
To this end, we extend the potential in Eq.~\ref{eq:app-potential-argmax} under perfect rationality to a potential $\Phi$ under bounded rationality by observing that: (i) for $\beta<\infty$, all providers can have a non-zero market share (and hence non-zero utilities), (ii) each provider partially aims to improve their profit $p_i(\theta_i) - c_i(\theta_i)$ while still offering a sufficiently high value $V_i(\theta_i)$ relative to others, and (iii) in the log domain, utilities decompose as:
\begin{align}
    \log U_i(\theta_i;\thetab_{-i}) &= \log\left( s_i(V_i(\theta_i); \Vb_{-i}(\thetab_{-i})) \cdot (p_i(\theta_i)-c_i(\theta_i)) \right)\\
    &=\log\left(p_i(\theta_i)-c_i(\theta_i)\right) + \beta V_i(\theta_i) - \log \left( \sum_{j=0}^N \exp(\beta\cdot V_j(\theta_j)) \right),
\end{align}
for $\thetab\in\Theta^N$ and any provider $i$, and we denote $V_0(\theta_0) \coloneqq V_0$ for simplicity. Based on the three observations above, let us define:
\begin{equation}
    \Phi(\thetab) = \sum_{j=1}^N \log(p_j(\theta_j)-c_j(\theta_j)) +  \beta\cdot\sum_{j=1}^N  V_j(\theta_j) - \log \left( \sum_{j=0}^N \exp(\beta\cdot V_j(\theta_j)) \right),\quad\forall \thetab\in\Theta^N,
\end{equation}
where note that the last term is not summed over all providers, since it is a shared normalization constant related to the softmax allocation $s$ under bounded rationality.
To verify that $\Phi$ is indeed an ordinal potential function, consider a deviation $\theta_i'$ for provider $i$.
Then:
\begin{align*}
    &U_i(\theta'_i; \boldsymbol{\theta}_{-i}) > U_i(\theta_i; \boldsymbol{\theta}_{-i})\\[10pt]
    \iff & \log U_i(\theta'_i; \boldsymbol{\theta}_{-i}) > \log U_i(\theta_i; \boldsymbol{\theta}_{-i}) \\[10pt]
    \iff & \log(p_i(\theta'_i)-c_i(\theta'_i)) + \beta V_i(\theta'_i) - \log \left( \sum_{j \neq i} e^{\beta V_j(\theta_j)} + e^{\beta V_i(\theta'_i)} \right) \\
    &\quad > \log(p_i(\theta_i)-c_i(\theta_i)) + \beta V_i(\theta_i) - \log \left( \sum_{j=0}^N e^{\beta V_j(\theta_j)} \right) \\[10pt]
    \iff & \sum_{j \neq i,0} \left[ \log(p_j(\theta_j)-c_j(\theta_j)) + \beta V_j(\theta_j) \right] + \log(p_i(\theta'_i)-c_i(\theta'_i)) + \beta V_i(\theta'_i) \\
    &\quad {}- \log \left( \sum_{j \neq i} e^{\beta V_j(\theta_j)} + e^{\beta V_i(\theta'_i)} \right) \\
    &\quad > \sum_{j \neq i,0} \left[ \log(p_j(\theta_j)-c_j(\theta_j)) + \beta V_j(\theta_j) \right] + \log(p_i(\theta_i)-c_i(\theta_i)) + \beta V_i(\theta_i) \\
    &\quad {}- \log \left( \sum_{j=0}^N e^{\beta V_j(\theta_j)} \right).
\end{align*}
and we can conclude that:
\begin{align*}
    &U_i(\theta'_i; {\thetab}_{-i}) > U_i(\theta_i; {\thetab}_{-i})\\
    \iff & \textcolor{blue}{\sum_{j \neq i} \left[ \log(p_j(\theta_j)-c_j(\theta_j)) + \beta V_j(\theta_j) \right] + \log(p_i(\theta'_i)-c_i(\theta'_i)) + \beta V_i(\theta'_i)} \\
    &\quad \textcolor{blue}{{}- \log \left( \sum_{j \neq i} e^{\beta V_j(\theta_j)} + e^{\beta V_i(\theta'_i)} \right)} \\
    &\quad > \textcolor{orange}{\sum_{j=1}^N \left[\log(p_j(\theta_j)-c_j(\theta_j)) + \beta V_j(\theta_j)\right] - \log \left( \sum_{j=0}^N e^{\beta V_j(\theta_j)} \right)} \\
    \iff & \textcolor{blue}{\Phi(\theta'_i; {\thetab}_{-i})} > \textcolor{orange}{\Phi({\thetab})}.
\end{align*}

\clearpage
\newpage

\subsection{Proof of Theorem~\ref{thm:rational-equilibrium}}

We fix a test-time compute game $\Gcal$ under perfect rationality and any pure Nash equilibrium $\thetab^\dagger\in\Theta^N$. Firstly, observe that it must be
\begin{equation}
    s_{1}\left(V_{1}(\theta^\dagger_{1});\Vb_{-1}(\thetab^\dagger_{-1})\right) = 1.
\end{equation}
Indeed, if the above were false, by definition of $s$ under perfect rationality, there would exist a provider $i\neq 1$ such that $V_{i}(\theta^\dagger_{i})>V_{1}(\theta^\dagger_{1})$. Then, given that providers are indexed in decreasing order of their maximum possible offered values, we have that $V^*_{1} > V^*_{i}$ and hence provider $1$ could best respond and strictly increase their utility, meaning that $\thetab^\dagger$ would not be a Nash equilibrium. Hence, in $\thetab^\dagger$, provider $1$ offers a strictly higher value than any other provider.
Furthermore, if $V_{1}(\thetab^\dagger_{1})< V_{2}^*$, then, at least provider $2$ could best respond by increasing the value they offer to $V_2^*$, serving all queries and obtaining strictly positive utility. Thus, in $\thetab^\dagger$ it holds that $V_{1}(\theta_{1}) > V^*_{2}$, and since provider $1$ is playing a best response conditional on serving all queries, it must be that:

\begin{equation}\label{eq:app-dominant-NE-rational}
     \theta^\dagger_{1}=\argmax_{\thetab_{1} \in \Theta} 
        \left\{ p_{1}(\theta_{1}) - c_{1}(\theta_{1}) \;\middle|\; V_{1}(\theta_{1}) > V^*_{2} \right\}=\max\{\theta_1 | V_{1}(\theta_{1}) > V^*_{2}  \},
\end{equation}
where the second equality holds because providers have profits that are increasing with the compute level. Next, we derive conditions on the compute of providers $i\neq 1$ such that $\thetab^\dagger$ is a Nash equilibrium. To this end, observe that since $V_1(\theta^\dagger_1)> V_2^*$, any provider $i\neq 1$ is unable to serve any queries, and hence has utility $0$ independently of their deviations in $\thetab^\dagger$.  We then distinguish two cases:
\begin{itemize}
    \item If for any $\theta_{1} \in \Theta$ we have that $V_{1}(\theta_{1})>V_{2}^*$, then any provider $i\neq 1$ cannot increase their utilities, and provider 1 selecting their compute as in Eq.~\ref{eq:app-dominant-NE-rational} is not able to improve their utility. We conclude that in this case, providers $i\neq 1$ selecting their compute arbitrarily leads to a pure Nash equilibrium (although the equilibrium is not strict).
    \item Suppose there exists $\theta_{1} \in \Theta$ such that $V_{1}(\theta_{1}) < V_{2}^*$ and that $p_{1}(\theta_{1}) - c_{1}(\theta_{1}) \geq p_{1}(\theta^\dagger_{1}) - c_{1}(\theta^\dagger_{1})$ (note that we recover the previous case if $p_{1}(\theta_{1}) - c_{1}(\theta_{1}) < p_{1}(\theta^\dagger_{1}) - c_{1}(\theta^\dagger_{1})$), with $\theta^\dagger_{1}$ as in Eq.~\ref{eq:app-dominant-NE-rational}. Then, provider $1$ can strictly increase their utility by choosing $\theta_{1}$ unless there exists another provider $i$ with compute $\theta_i^\dagger$ such that $V_i(\theta_i^\dagger) > V_{1}(\theta_{1})$.
\end{itemize}
We can summarize the above two cases by stating that at least one provider must offer a value higher than:
\begin{equation}
    \max_{\theta_{1}\in\Theta}\left\{ V_{1}(\theta_{1}) \middle|  V_{1}(\theta_{1}) < V_{2}^* \quad\mathrm{and}\quad p_{1}(\theta_{1}) - c_{1}(\theta_{1}) > p_{1}(\theta^\dagger_{1}) - c_{1}(\theta^\dagger_{1}) \right\}
\end{equation}
with the constraint being lifted if the above set is empty. This concludes the proof.

\newpage
\clearpage

\subsection{Proof of Theorem~\ref{thm:high-rationality}}

The intuition for the proof is that the softmax converges to an indicator function for the maximum offered value as $\beta \to \infty$, and hence, the equilibrium computes should also converge. We now make this intuition precise, and assume $V_0=0$ in what follows for simplicity. We will leverage the following well-known fact about the softmax function, which we prove here for completeness:

\begin{lemma}\label{lemma:softmax-indicator}
    Let $\Xb=(x_1,\dots,x_N) \in \mathbb{R}^N$ such that $x_1>x_2>\dots>x_N$ and $\sigma^\beta(\Xb)$ denote the softmax function with inverse temperature $\beta$. Then:
    \begin{equation}
        |\sigma^\beta(\Xb)_i - \mathds{1}\{i=1\}| \leq (N-1) \cdot \exp(-\beta\cdot (x_1-x_2)),\quad \forall i=1,\dots,N.
    \end{equation}
\end{lemma}
\begin{proof}
    From the softmax definition, we have:
    \begin{equation}
        \begin{dcases}
            \sigma^\beta(\Xb)_1 = \frac{1}{1+\sum_{j=2}^N \exp(-\beta(x_1 - x_j))}\\
            \sigma^\beta(\Xb)_i = \frac{\exp(-\beta(x_1 - x_i))}{1+\sum_{j=2}^N \exp(-\beta(x_1 - x_j))},\quad i\geq 2.
        \end{dcases}
    \end{equation}
    Then, for $i\neq 1$ we have:
    \begin{equation}
        |\sigma^\beta(\Xb)_i - \mathds{1}\{i=1\}| = \sigma^\beta(\Xb)_i \leq \exp(-\beta(x_1 - x_i)) \leq (N-1)\cdot \exp(-\beta\cdot (x_1-x_2)),
    \end{equation}
    while for $i = 1$ we have:
    \begin{equation}
        |\sigma^\beta(\Xb)_i - \mathds{1}\{i=1\}| = 1 - \sigma^\beta(\Xb)_1 \leq \sum_{j =  2}^N \exp(-\beta\cdot(x_1-x_j)) \leq (N-1)\cdot \exp(-\beta\cdot (x_1-x_2)).
    \end{equation}
\end{proof}

To prove Theorem~\ref{thm:high-rationality}, we will use Lemma~\ref{lemma:softmax-indicator} with $\Xb$ corresponding to the vector of offered values $(\Vb(\thetab),0)$ (which has $N+1$ components, including $V_0 = 0$). More precisely, let us denote by $s^\beta$ and $s^\infty$ the market share functions under bounded rationality and perfect rationality, respectively (see Table~\ref{tab:market-models}). Then, for any $\thetab\in\Theta^N$, and any provider $i$, we have:
\begin{equation}
    \left| s^\beta_i(V_i(\theta_i);\Vb_{-i}(\thetab_{-i})) - s_i^\infty(V_i(\theta_i);\Vb_{-i}(\thetab_{-i})) \right| \leq N\cdot \exp\left(-\beta \cdot(V_{\pi(1)}(\theta_{\pi(1)}) - V_{\pi(2)}(\theta_{\pi(2)}))\right),
\end{equation}
where $\pi$ is the permutation ordering providers by their offered value at $\thetab$.
Taking the maximum among providers $i$ and the (finitely many) compute levels they can select in the above expression, and letting $\delta_{\min} > 0 $ be the minimum gap in value offered by any two different providers (recall that we are assuming that providers always offer different values):
\begin{equation}
    \max_{\thetab} \max_i \left| s^\beta_i(V_i(\theta_i);\Vb_{-i}(\thetab_{-i})) - s_i^\infty(V_i(\theta_i);\Vb_{-i}(\thetab_{-i})) \right| \leq N\cdot \exp\left(-\beta \cdot \delta_{\min}\right).
\end{equation}
Thus, we have that $s^\beta \to s^\infty$ as $\beta \to \infty$, where the convergence is uniform over $\Theta^N$ and the provider index $i$. Then, since for each provider $i$, the profit $p_i(\theta_i) - c_i(\theta_i)$ is bounded on the finite domain $\Theta$, we have that $\Ub^\beta \to \Ub^\infty$ as $\beta \to \infty$ uniformly over $\Theta^N$ and $i$, where $\Ub^\beta\colon \Theta^N \to \mathbb{R}^N$ denotes the function that maps a joint compute choice $\thetab$ to the provider's utilities as defined in Eq.~\ref{eq:provider-utility} using the market allocation function $s^\beta$, and $\Ub^\infty$ is defined analogously.
To use the above convergence result, we let 
\begin{equation}\label{app:epsilon-con}
    \varepsilon<\frac{1}{2} \min_{\thetab, \thetab'} \min_i \left\{ |U^\infty_i(\theta_i;\thetab_{-i}) - U^\infty_i(\theta'_i;\thetab'_{-i})| \,\colon\, |U^\infty_i(\theta_i;\thetab_{-i}) - U^\infty_i(\theta'_i;\thetab'_{-i})| > 0 \right\}
\end{equation}
be at most half the minimum positive change in utility a provider can achieve in the game $\Gcal$ under perfect rationality. Then, let $\beta_0$ be such that for any $\beta > \beta_0$, it holds that
\begin{equation}\label{app:uniform-con}
    \max_{\thetab} \max_i \left| U^\beta_i(\theta_i;\thetab_{-i}) - U_i^\infty(\theta_i;\thetab_{-i}) \right| \leq \varepsilon
\end{equation}
and fix any pure Nash equilibrium  $\thetab^\dagger$ of $\Gcal$ for $\beta>\beta_0$. To show that $\thetab^\dagger$ is also an equilibrium of $\Gcal$ under perfect rationality, we argue by contradiction, and suppose this premise is false. Then, there exists a provider $i$ that can change their compute in the game $\Gcal$ under perfect rationality and strictly increase their utility, that is, there exists $\theta_i'\in\Theta$ such that $U_i^\infty(\theta'_i;\thetab^\dagger_{-i}) > U_i^\infty(\theta^\dagger_i;\thetab^\dagger_{-i})$. Then:
\begin{align}
    2\varepsilon &\overset{(*)}{<}  U_i^\infty(\theta'_i;\thetab^\dagger_{-i}) - U_i^\infty(\theta^\dagger_i;\thetab^\dagger_{-i}) \\
    &= \left[ U_i^\infty(\theta'_i;\thetab^\dagger_{-i}) - U_i^\beta(\theta'_i;\thetab^\dagger_{-i}) \right] + \left[ U_i^\beta(\theta'_i;\thetab^\dagger_{-i}) - U_i^\beta(\theta^\dagger_i;\thetab^\dagger_{-i}) \right] \nonumber \\
    & \quad + \left[ U_i^\beta(\theta^\dagger_i;\thetab^\dagger_{-i}) - U_i^\infty(\theta^\dagger_i;\thetab^\dagger_{-i}) \right]\\
    & \overset{(**)}{\leq}  \varepsilon + 0 + \varepsilon,
\end{align}
where $(*)$ follows from Eq.~\ref{app:epsilon-con} together with the contradiction hypothesis $U_i^\infty(\theta'_i;\thetab^\dagger_{-i}) - U_i^\infty(\theta^\dagger_i;\thetab^\dagger_{-i}) > 0$, and $(**)$ follows from bounding the first and third terms by $\varepsilon$ via Eq.~\ref{app:uniform-con}, and the second term by $0$ since $\thetab^\dagger$ is a Nash equilibrium of $\Gcal$ under bounded rationality, hence $U_i^\beta(\theta^\dagger_i;\thetab^\dagger_{-i}) \geq U_i^\beta(\theta'_i;\thetab^\dagger_{-i})$. This is a contradiction, which proves that $\thetab^\dagger$ is also an equilibrium of $\Gcal$ under perfect rationality.

\clearpage
\newpage

\subsection{Proof of Theorem~\ref{thm:poa}}\label{app:proof-poa}
The following is an auxiliary result that will be used to prove Theorem~\ref{thm:poa}:


\begin{lemma}\label{lemma:mean-softmax}
    Let $x_1 > x_2 > \dots > x_N$ and $y_1, \dots, y_N$ be real positive 
    numbers, and denote both vectors by $\Xb$ and $\Yb$, respectively. Denote by 
    $\sigma^\beta(\Xb)$ the softmax with inverse temperature $\beta$. Then,
    \begin{equation}\label{eq:app-softmax-average-bound}
        \begin{cases}
            \sum_{i=1}^N y_i \cdot \sigma^\beta(\Xb)_i \leq y_1 + 
                e^{-\beta(x_1-x_2)} \cdot \sum_{i=2}^N (y_i + y_1) \\[4pt]
            \sum_{i=1}^N y_i \cdot \sigma^\beta(\Xb)_i \geq y_1 - 
                e^{-\beta(x_1-x_2)} \cdot \sum_{i=2}^N (y_i + y_1)
        \end{cases}
    \end{equation}
\end{lemma}

\begin{proof}
    From the softmax definition, we have:
    \begin{equation}
        \begin{dcases}
            \sigma^\beta(\Xb)_1 = \frac{1}{1+\sum_{j=2}^N e^{-\beta(x_1 - x_j)}}\\
            \sigma^\beta(\Xb)_i = \frac{e^{-\beta(x_1 - x_i)}}{1+\sum_{j=2}^N e^{-\beta(x_1 - x_j)}} \leq e^{-\beta(x_1 - x_i)},\quad i\geq 2,
        \end{dcases}
    \end{equation}

    and,
    \begin{equation}\label{eq:app-softmax-average-aux}
        \sum_{i=1}^N y_i\cdot \sigma^\beta(\Xb)_i = y_1 +  \sum_{i=2}^N (y_i - y_1) \cdot\sigma^\beta(\Xb)_i. 
    \end{equation}

    We begin by showing that the first inequality in the lemma holds. Since 
    $y_i, y_1 > 0$, we have $y_i - y_1 \leq |y_i - y_1| \leq y_i + y_1$, and 
    since $\sigma^\beta(\Xb)_i \geq 0$:
    \begin{equation}
        \sum_{i=2}^N (y_i - y_1)\,\sigma^\beta(\Xb)_i 
        \;\leq\; \sum_{i=2}^N (y_i + y_1)\,\sigma^\beta(\Xb)_i.
    \end{equation}
    Then, using that $\sigma^\beta(\Xb)_i \leq e^{-\beta(x_1 - x_i)} \leq 
    e^{-\beta(x_1 - x_2)}$ for $i \geq 2$, and that $y_i + y_1 > 0$:
    \begin{equation}
        \sum_{i=2}^N (y_i + y_1)\,\sigma^\beta(\Xb)_i 
        \;\leq\; e^{-\beta(x_1 - x_2)} \sum_{i=2}^N (y_i + y_1),
    \end{equation}
    which combined with Eq.~\eqref{eq:app-softmax-average-aux} yields the first 
    inequality.
    
    Similarly, we can show that the second inequality in~\ref{eq:app-softmax-average-bound} also holds. Indeed, from Eq.~\ref{eq:app-softmax-average-aux}, we obtain:
    \begin{align}
        \sum_{i=1}^N y_i \sigma^\beta(\Xb)_i &= y_1 -  \sum_{i=2}^N (y_1 - y_i) \cdot\sigma^\beta(\Xb)_i\\
        &\geq y_1 -  \sum_{i=2}^N (y_1 + y_i) \cdot\sigma^\beta(\Xb)_i\\
        &\geq y_1 -  \sum_{i=2}^N (y_1 + y_i) \cdot e^{-\beta(x_1-x_i)}\\
        &\geq y_1 -  \sum_{i=2}^N (y_1 + y_i) \cdot e^{-\beta(x_1-x_2)}
    \end{align}
\end{proof}

We now prove Theorem~\ref{thm:poa}. To this end, fix a test-time compute game $\Gcal$ with $0<\beta \leq \infty$ and a Nash equilibrium $\thetab^\dagger$. Further, for each provider $i$, denote by $\theta^*_i = \argmax_\theta \{ q_i(\theta) - c_i(\theta) \}$ their socially optimal compute, and let $\thetab^* = (\theta_i^*)_i$. Recall the definition of the price of anarchy in Eq.~\ref{eq:PoA}:\footnote{For ease of exposition, we take $V_0=0$, with the case $V_0 > 0$ following using an analogous argument.}
\begin{align}\label{app-eq-poa-expansion}
    \text{PoA}(\Gcal) =
    \frac{ \max_{\theta \in \Theta^{N}} \SW(\thetab) }{\SW(\thetab^\dagger)}
    \geq \frac{  \SW(\thetab^*)}{\SW(\thetab^\dagger)}
    = \frac{ \overbrace{\sum_{i=1}^{N} s_i(V_i(\theta^*_i);\Vb_{-i}(\thetab^*_{-i}))\cdot \SW_i(\theta^*_i) }^{(\diamond)} }{ \underbrace{\sum_{i=1}^{N} s_i(V_i(\theta^\dagger_i);\Vb_{-i}(\thetab^\dagger_{-i}))\cdot \SW_i(\theta_i^\dagger)}_{(\bullet)} }.
\end{align}
Denote by $\pi$ and $\pi^*$ the permutations ordering providers by their value at the equilibrium $\thetab^\dagger$ and at $\thetab^*$, respectively. Then, we can use Lemma~\ref{lemma:mean-softmax} to lower bound the numerator and upper bound the denominator in the above.

Starting with the numerator $(\diamond)$, we take in Lemma~\ref{lemma:mean-softmax} $\Xb$ to be the vector $(\Vb(\thetab^*), 0)$ with $N+1$ components ordered by $\pi^*$ (assuming without loss of generality that all providers can offer a value higher than the abstention value) and $\Yb$ to be the vector $(\SW_i(\theta^*_i),0)$ with $N+1$ components ordered by $\pi^*$. Then, Lemma~\ref{lemma:mean-softmax} implies:

\begin{align}
    (\diamond) \geq \SW_{\pi^*(1)}(\theta^*_{\pi^*(1)}) \cdot \left(1 -e^{-\beta\Delta V^*} \right)- e^{-\beta\Delta V^*} \cdot \sum_{i \geq 2}^{N} \left( \SW_{\pi^*(i)}(\theta_{\pi^*(i)}^*) +  \SW_{\pi^*(1)}(\theta_{\pi^*(1)}^*)\right),
\end{align}
where $\Delta V^* = V_{\pi^*(1)}(\theta^*_{\pi^*(1)})- V_{\pi^*(2)}(\theta^*_{\pi^*(2)})$ is the difference in values at $\thetab^*$ between the first and second providers.

Similarly, the term $(\bullet)$ can be upper-bounded by taking in Lemma~\ref{lemma:mean-softmax} $\Xb$ to be the vector $(\Vb(\thetab^\dagger), 0)$ ordered by $\pi$ (assuming without loss of generality that all providers can offer a value higher than the abstention value) and $\Yb$ to be the vector $(\SW_i(\theta^\dagger_i),0)$ ordered by $\pi$. Then, we obtain:

\begin{align}
    (\bullet) \leq \SW_{\pi(1)}(\theta_{\pi(1)}^\dagger) \cdot \left(1 + e^{-\beta\Delta V^\dagger} \right) + e^{-\beta\Delta V^\dagger} \cdot \sum_{i\geq 2}^{N} \left( \SW_{\pi(i)}(\theta_{\pi(i)}^\dagger) +  \SW_{\pi(1)}(\theta_{\pi(1)}^\dagger)\right),
\end{align}
where $\Delta V^\dagger = V_{\pi(1)}(\theta^\dagger_{\pi(1)})- V_{\pi(2)}(\theta^\dagger_{\pi(2)})$. 

Using such bounds on $(\diamond)$ and on $(\bullet)$, the price of anarchy satisfies:
\begin{align}\label{eq:app-poa-bound-aux}
    \text{PoA}(\Gcal) &\geq \frac{ \SW_{\pi^*(1)}(\theta^*_{\pi^*(1)}) \cdot \left(1 -e^{-\beta\Delta V^*} \right) - e^{-\beta\Delta V^*} \cdot \sum_{i \geq 2}^{N} \left( \SW_{\pi^*(i)}(\theta_{\pi^*(i)}^*) +  \SW_{\pi^*(1)}(\theta_{\pi^*(1)}^*)\right) }{ \SW_{\pi(1)}(\theta_{\pi(1)}^\dagger) \cdot \left(1 + e^{-\beta\Delta V^\dagger} \right) + e^{-\beta\Delta V^\dagger} \cdot \sum_{i\geq 2}^{N} \left( \SW_{\pi(i)}(\theta_{\pi(i)}^\dagger) +  \SW_{\pi(1)}(\theta_{\pi(1)}^\dagger)\right) }\nonumber\\
    &= \frac{  \left(1 -e^{-\beta\Delta V^*} \right) - e^{-\beta\Delta V^*} \cdot \sum_{i \geq 2}^{N} \left( 1 + \frac{\SW_{\pi^*(i)}(\theta_{\pi^*(i)}^*)}{ \SW^*}  \right) }
    { \frac{\SW_{\pi(1)}(\theta_{\pi(1)}^\dagger)}{\SW^*} \cdot \left(1 + e^{-\beta\Delta V^\dagger} \right) + e^{-\beta\Delta V^\dagger} \cdot \sum_{i\geq 2}^{N} \left( \frac{\SW_{\pi(1)}(\theta_{\pi(1)}^\dagger) +  \SW_{\pi(i)}(\theta_{\pi(i)}^\dagger)}{\SW^*}\right) },
\end{align}
where we have defined $\SW^* \coloneqq \SW_{\pi^*(1)}(\theta^*_{\pi^*(1)})$.
To simplify, we define:
\begin{align*}
    \Delta \SW \coloneqq \SW_{\pi^*(1)}(\theta^*_{\pi^*(1)}) - \SW_{\pi(1)}(\theta_{\pi(1)}^\dagger),
\end{align*}
which represents the difference in the contribution to the social welfare by the provider offering the highest value at equilibrium, and at compute $\thetab^*$.
Using this in the bound in Eq.~\ref{eq:app-poa-bound-aux} and Taylor-expanding up to second order the function $t \mapsto 1/(1-t)$ around $t=0$, we obtain:
\begin{align*}
    \text{PoA}(\Gcal) &\geq \frac{  \left(1 -e^{-\beta\Delta V^*} \right) - e^{-\beta\Delta V^*} \cdot \sum_{i \geq 2}^{N} \left( 1 + \frac{\SW_{\pi^*(i)}(\theta_{\pi^*(i)}^*)}{ \SW^*}  \right) }
    { 1-\Delta\SW / \SW^* + e^{-\beta\Delta V^\dagger} \left( 1+ \sum_{i\geq 2}^{N} \frac{\SW_{\pi(1)}(\theta_{\pi(1)}^\dagger) +  \SW_{\pi(i)}(\theta_{\pi(i)}^\dagger)}{\SW^*} \right)  + \Ocal\left( \lVert e^{-\beta\Delta V^\dagger}, \Delta\SW \rVert^2 \right) }\\
    &= \left( 
    \left(1 -e^{-\beta\Delta V^*} \right) - e^{-\beta\Delta V^*} \cdot \sum_{i \geq 2}^{N} \left( 1 + \frac{\SW_{\pi^*(i)}(\theta_{\pi^*(i)}^*)}{\SW^*} \right)
    \right)\\
    &\quad \cdot \left( 1 + \Delta \SW / \SW^* - e^{-\beta \Delta V^\dagger}  - e^{-\beta \Delta V^\dagger} \sum_{i\geq 2}^{N} \frac{\SW_{\pi(1)}(\theta_{\pi(1)}^\dagger) + \SW_{\pi(i)}(\theta_{\pi(i)}^\dagger)}{\SW^*} \right)\\
    &\quad+ \Ocal\left( \lVert e^{-\beta\Delta V^\dagger}, e^{-\beta\Delta V^*}, \Delta\SW \rVert^2 \right) \\
    &= \left( 
    1 - e^{-\beta\Delta V^*} \left( 1 + \sum_{i \geq 2}^{N} \left( 1 + \frac{\SW_{\pi^*(i)}(\theta_{\pi^*(i)}^*)}{\SW^*} \right) \right) \right)\\
    &\quad \cdot \left( 1 + \Delta \SW/\SW^* - e^{-\beta \Delta V^\dagger} \left( 1 + \sum_{i\geq 2}^{N} \frac{\SW_{\pi(1)}(\theta_{\pi(1)}^\dagger) + \SW_{\pi(i)}(\theta_{\pi(i)}^\dagger)}{\SW^*} \right)\right)\\
    &\quad + \Ocal\left( \lVert e^{-\beta\Delta V^\dagger}, e^{-\beta\Delta V^*}, \Delta\SW \rVert^2  \right) \\
    &= 1 + \Delta\SW/\SW^* - e^{-\beta \Delta V^\dagger}\cdot \left( 1+ \sum_{i\geq 2}^{N} \frac{\SW_{\pi(1)}(\theta_{\pi(1)}^\dagger) + \SW_{\pi(i)}(\theta_{\pi(i)}^\dagger)}{\SW^*}  \right) \\
    &\quad - e^{-\beta\Delta V^*} \cdot \left( 1+ \sum_{i \geq 2}^{N} \left( 1 + \frac{\SW_{\pi^*(i)}(\theta_{\pi^*(i)}^*)}{\SW^*} \right) \right)\\
    &\quad +  \Ocal\left( \lVert e^{-\beta\Delta V^\dagger}, e^{-\beta\Delta V^*}, \Delta\SW \rVert^2 \right)
\end{align*}
To obtain the same form as in the statement of Theorem~\ref{thm:poa}, we can let $\Delta V = \min(\Delta V^*, \Delta V^\dagger)$ and define $f(\beta)$ such that:
\begin{multline*}
     f(\beta) = - \frac{e^{-\beta \Delta V^\dagger}}{\SW^*}\cdot \left( \SW^* + \sum_{i\geq 2}^{N} \left( W_{\pi(1)}(\theta_{\pi(1)}^\dagger) + \SW_{\pi(i)}(\theta_{\pi(i)}^\dagger) \right) \right) \\- \frac{e^{-\beta\Delta V^*}}{\SW^*} \cdot \left( \SW^*+ \sum_{i \geq 2}^{N} \left( \SW^* + \SW_{\pi^*(i)}(\theta_{\pi^*(i)}^*) \right) \right) = \Ocal\left( e^{-\beta\Delta V}\right).
\end{multline*}

\clearpage
\newpage
\subsection{Proof of Theorem~\ref{thm:second-price-dominant}}\label{app:proof-auction-dominant}

Fix a game $\widetilde{\Gcal}$, a provider $i$, and any compute levels $\thetab_{-i}$ and bid prices ${\mathbf{p}}_{-i}$.
We will prove that, for provider $i$, selecting $\theta_i=\theta_i^*$ and $ p_i = c_i(\theta_i^*)$ is a dominant strategy. To this end, we first show that given a fixed $\theta_i$, the provider maximizes their utility by bidding ${p}_i = c_i(\theta_i)$, and then show that $\theta_i=\theta_i^*$ is their best choice of compute. Recall that the utility of provider $i$ is (Eq.~\ref{eq:utility-auction}):
\begin{align*}
        {U}_i(\theta_i, {p}_i; \thetab_{-i}, {\mathbf{p}}_{-i}) &= \left( {P}(\thetab,{\mathbf{p}}) - c_i(\theta_i) \right) \cdot \mathds{1}\left\{ {V}_i (\theta_i, p_i) > \max_{j\neq i} {V}_j(\theta_j, p_j) \right\}\\
        & = \left( q_{\pi(1)}(\theta_{\pi(1)}) - c_i(\theta_i) - (q_{\pi(2)}(\theta_{\pi(2)})-{p}_{\pi(2)})  \right) \cdot \mathds{1}\left\{ i = \pi(1) \right\}.
\end{align*}

\xhdr{Bidding generation cost is optimal} Let us first prove that, for a fixed $\theta_i$, bidding ${p}_i = c_i(\theta_i)$ is always optimal. To this end, we assume that the provider bids ${p}_i = c_i(\theta_i)$ and show that no deviation to a different price bid can increase their utility. Consider first that provider $i$ wins the bid when bidding $(q_i(\theta_i), c_i(\theta_i))$. Thus, it must be that $i = \pi(1)$ and $q_{i}(\theta_i) - c_i(\theta_i) > q_{\pi(2)}(\theta_{\pi(2)})-{p}_{\pi(2)}$. Then, consider a deviation where provider $i$ bids a different price ${p}_{i} \neq c_i(\theta_i)$. With their new bid $(q_{i}(\theta_i), {p}_{i})$, they either lose the auction, in which case they strictly decrease their utility to $0$, or they are still winning the auction. However, in the latter case, since conditional on winning, their utility is independent of ${p}_{i}$, we conclude that this deviation maintains their utility, and hence the provider cannot profit from it.

Now, suppose that provider $i$ initially loses the auction by bidding $(q_i(\theta_i), c_i(\theta_i))$, which implies that:
\begin{equation}
    q_i(\theta_i) - c_i(\theta_i) < q_{\pi(1)}(\theta_{\pi(1)})-{p}_{\pi(1)}.
\end{equation}
Then, any deviation where they bid a higher ${p}_{i} > c_i(\theta_i)$ will further decrease their offered value, and hence they will still lose the auction and maintain their null utility. In case the provider deviates to a lower price ${p}_{i} < c_i(\theta_i)$ and wins the auction by doing so, denote by $\pi'$ the permutation ordering providers after the deviation. Then, $\pi'(1) = i$ and $\pi'(2) = \pi(1)$. Consequently, the new utility of provider $i$ after the deviation will be:
\begin{align}
    {U}_i(\theta_i, {p}_i; \thetab_{-i}, {\mathbf{p}}_{-i}) = q_i(\theta_i) - c_i(\theta_i) - \left( q_{\pi(1)}(\theta_{\pi(1)})-{p}_{\pi(1)} \right) < 0,
\end{align}
meaning that the deviation is not profitable. Hence, we conclude that bidding ${p}_i = c_i(\theta_i)$ is always optimal.

\xhdr{Bidding the socially optimal compute is dominant} 
We now prove that bidding with $\theta_i=\theta_i^*$ is dominant for provider $i$. To this end, we assume that for any compute $\theta_i$, provider $i$ selects their optimal price ${p}_i = c_i(\theta_i)$. Then, the utility of the provider can be written as:
\begin{align*}
        {U}_i(\theta_i, c_i(\theta_i); \thetab_{-i}, {\mathbf{p}}_{-i}) =
         &\left( q_{\pi(1)}(\theta_{\pi(1)}) -c_i(\theta_i) \right) \cdot \mathds{1}\left\{ i = \pi(1) \right\}\\
         &-(q_{\pi(2)}(\theta_{\pi(2)})-{p}_{\pi(2)}) \cdot \mathds{1}\left\{ i = \pi(1) \right\}
\end{align*}
Since, conditional on provider $i$ winning the bid, the second term in the above expression does not depend on the action of provider $i$, their optimal compute level is to select:
\begin{align*}
    &\argmax_{\theta_i \in \Theta} \left\{ q_{i}(\theta_{i}) -c_i(\theta_i) \middle| i = \pi(1)\right\}\\
     =&  \argmax_{\theta_i \in \Theta}  \left\{ q_{i}(\theta_{i}) -c_i(\theta_i) \middle| q_{i}(\theta_{i}) -c_i(\theta_i)  > q_{\pi(2)}(\theta_{\pi(2)})-{p}_{\pi(2)}\right\}\\
    =& \argmax_{\theta_i \in \Theta}  \left\{ q_{i}(\theta_{i}) -c_i(\theta_i)\right\}\\
    =& \theta_i^*,
\end{align*}
which proves the claim.

\clearpage
\newpage

\section{Additional Experimental Details}\label{app:experimental-details}

Here, we provide additional experimental details for our empirical evaluation in Section~\ref{sec:experiments}. The complete code and implementation is available at \url{https://github.com/Human-Centric-Machine-Learning/strategic-ttc}.

\xhdr{Hardware setup} Our experiments are executed on a compute server equipped with 2 $\times$ Intel Xeon Gold 5317 CPU, $1{,}024$ GB main memory, and $2$ $\times$ A100 Nvidia Tesla GPU ($80$ GB, Ampere Architecture). In each experiment, a single Nvidia A100 GPU is used.

\xhdr{Models} 
We use the following LLMs in our experiments.
From the \texttt{Llama} family, we use \texttt{\seqsplit{Llama-3-8B-Instruct}}, \texttt{\seqsplit{Llama-3.1-8B-Instruct}}, \texttt{\seqsplit{Llama-3.2-1B-Instruct}}, and \texttt{Llama-3.2-3B-Instruct}.
From the \texttt{Qwen} family, we use \texttt{Qwen-2-0.5B-Instruct}, \texttt{Qwen-2-1.5B-Instruct}, \texttt{Qwen-2-7B-Instruct}, \texttt{Qwen-2.5-3B-Instruct}, and \texttt{Qwen-2.5-7B-Instruct}. Finally, we include three reasoning models distilled from \texttt{DeepSeek-R1}: \texttt{DeepSeek-R1-Distill-Llama-8B}, \texttt{DeepSeek-R1-Distill-Qwen-1.5B}, and \texttt{DeepSeek-R1-Distill-Qwen-7B}.
%
Additionally, when using best-of-n sampling, we use the \texttt{ArmoRM-Llama3-8B-v0.1} reward model to score the generated outputs by the above models.
We obtain all the models from the publicly available Hugging Face repository.\footnote{\url{https://huggingface.co}}

\xhdr{Datasets}
We use four datasets to evaluate the performance of all LLMs and simulate test-time compute games: \texttt{GPQA}~\cite{rein2023gpqagraduatelevelgoogleproofqa}, a multiple-choice STEM question-answering benchmark; \texttt{GSM8K}~\cite{cobbe2021training}, a mathematical benchmark containing grade-school level problems; \texttt{AIME}~\cite{aime25}, a mathematical reasoning benchmark with problems from the American Invitational Mathematics Examination; and \texttt{TruthfulQA}~\cite{lin2022truthfulqameasuringmodelsmimic}, which evaluates models on their ability to generate truthful answers to questions reflecting
common misconceptions. The first three datasets have, for each query, a
verifiable ground-truth answer that we use to evaluate accuracy. To measure quality for \texttt{TruthfulQA}, we consider the set of true and false reference answers that are provided for each question, and we use a \texttt{BLEURT}-based score ~\cite{sellam2020bleurtlearningrobustmetrics} defined as the difference between the maximum similarity to the true reference
answers and the maximum similarity to the false reference answers. Furthermore, given that
\texttt{TruthfulQA} consists of open-ended questions where majority voting is
not directly applicable, and involves short responses that do not benefit from
extended chain-of-thought reasoning, we restrict experiments on this dataset to
non-reasoning models using best-of-n. All datasets are obtained from the publicly available Hugging Face repository.\footnote{\url{https://huggingface.co/datasets/Idavidrein/gpqa}}\footnote{\url{https://huggingface.co/datasets/openai/gsm8k}}\footnote{\url{https://huggingface.co/datasets/Maxwell-Jia/AIME_2024}}\footnote{\url{https://huggingface.co/datasets/domenicrosati/TruthfulQA}}

\xhdr{Generation details}
When generating model outputs and evaluating them on the above datasets, we adhere to the temperature settings recommended in the official Hugging Face model cards. That is, temperature $0.6$ for the \texttt{Llama} family and the corresponding distilled reasoning model, and temperature $0.7$ for the \texttt{Qwen} family and the corresponding distilled reasoning models. We disable top-$p$ and top-$k$ sampling when generating model outputs. To ensure robust evaluation and enable bootstrap resampling to estimate uncertainties, we generated $128$ candidate responses per query for non-reasoning models and $32$ for reasoning models. For the simulations, we used subsets of these pools (up to $\theta=64$ for non-reasoning and the full distribution for reasoning quantiles). We report  $95\%$ bootstrapped confidence intervals for the qualities of all models.
Our prompts and answer verification pipelines are adapted from established evaluation frameworks, specifically OpenCompass\footnote{\url{https://github.com/open-compass/opencompass}}, EvalScope\footnote{\url{https://github.com/modelscope/evalscope}}, and the LM Evaluation Harness\footnote{\url{https://github.com/EleutherAI/lm-evaluation-harness}}.

\xhdr{Quality computation} For the datasets with verifiable ground truths (\texttt{GPQA}, \texttt{GSM8K}, and \texttt{AIME}), we linearly map average model accuracy to user quality $q_i(\theta)$ using a proportionality constant $\alpha$, set to $\$0.008$ in our main experiments. We vary this constant in the additional experimental results in Appendix~\ref{app:additional-experimental-results}. Our main results are robust to this linear, yet natural, choice. For instance, under the setting of Section~\ref{sec:experiments}, if users derive zero quality $q_i(\theta) = 0$ whenever the model’s average accuracy is below $70\%$ and quality proportional to average accuracy otherwise, using the same constant $\$0.008$, then the price of anarchy of the test-time compute game is $1.18$.

\xhdr{Test-time compute methods} In our experiments, we consider the following three test-time compute methods, each used to simulate a different test-time compute game $\Gcal$ in Section~\ref{sec:experiments}:
\begin{itemize}
    \item \xhdr{Majority voting} For each query, each model generates $\theta=\{2^0, \ldots, 2^6\}$ independent outputs. The final response of the model is taken as the most frequent response across the $\theta$ generations. The prices $p(\theta)$ and costs $c(\theta)$ are computed based on the total number of generated tokens across all $\theta$ outputs, averaged across queries.
    \item \xhdr{Best-of-n~\citep{chow2024inference}} For each query, the model generates $\theta=\{2^0, \ldots, 2^6\}$ independent outputs. The final response of the model is taken as the highest-scored response according to the scoring model \texttt{ArmoRM-Llama3-8B-v0.1}. Specifically, for \texttt{GSM8K}, \texttt{GPQA} and \texttt{AIME}, we use the default score provided by the reward model, while for \texttt{TruthfulQA}, we follow ~\citet{cui2024ultrafeedbackboostinglanguagemodels} and use the truthfulness attribute of the reward model.
    The prices $p(\theta)$ and costs $c(\theta)$ are computed based on the total number of generated tokens across all $\theta$ outputs, averaged across queries. We do not account for the cost of the reward model scoring in this computation. Empirically, we find that the additional overhead from scoring is small, contributing between approximately 0.5\% and 1.2\% of total inference time across models (\eg, $\sim0.8\%$ for Qwen2-0.5B, $\sim0.5\%$ for Qwen2-7B, $\sim0.2\%$ for Llama-3.2-1B, and $\sim0.5\%$ for Llama-3.1-8B), and thus omit it from the cost model.
    \item  \xhdr{Chain-of-thought~\citep{wei2022chain}} Since the Hugging Face API does not allow explicit control over the level of test-time compute or reasoning used for chain-of-thought for the models we consider, we adopt the following heuristic approach. For each query, we generate $32$ independent model outputs and then group these outputs into five quantile-based bins according to the number of reasoning tokens generated. More concretely, the first bin contains outputs whose reasoning-token counts fall within the lowest $20$-th percentile, and analogously for the remaining bins. In this process, we discard any output in which the model generates fewer than $5$ reasoning tokens or exceeds the maximum generation length of $2048$ tokens (with 58.5\%, 3.1\%, and 11.1\% of outputs retained on \texttt{GSM8K}, \texttt{AIME}, and \texttt{GPQA}, respectively).\footnote{As a result, the reported accuracies (\eg, in Figure~\ref{fig:aime_reasoning_models}) reflect only the subset of outputs that pass this filtering step, and may therefore overestimate true model accuracy.} Then, we compute the accuracy for each bin and average it over queries. While we emphasize that this approach is not equivalent to deterministically controlling for the reasoning level during chain-of-thought, it provides a reasonable first approximation of the relationship between reasoning effort---as measured by the number of reasoning tokens---and accuracy.
    \end{itemize}

\xhdr{Test-time compute prices and costs}
In our experiments in Section~\ref{sec:experiments}, we instantiate the test-time compute games $\Gcal$ with each provider serving a different LLM. To determine the per-output token price for each provider/model, we refer to the Hugging Face list of inference providers\footnote{\url{https://huggingface.co/docs/inference-providers/index}, consulted on December 30, 2025} and, for each LLM, compute the average token price across the listed providers that offer access to that model. The resulting prices per million output tokens are as follows: \$0.1455, \$0.1245, \$0.10, and \$0.08 for $\texttt{Llama-3-8B-Instruct}$, $\texttt{Llama-3.1-8B-Instruct}$, $\texttt{Llama-3.2-1B-Instruct}$, and $\texttt{Llama-3.2-3B-Instruct}$, respectively; \$0.10, \$0.10, \$0.20, \$0.065, and \$0.1465 for $\texttt{Qwen-2-0.5B}$, $\texttt{Qwen-2-1.5B}$, $\texttt{Qwen-2-7B}$, \texttt{\seqsplit{Qwen-2.5-3B}}, and $\texttt{Qwen-2.5-7B}$, respectively; and \$0.125, \$0.10, and \$0.175 for the models $\texttt{R1-D-Llama-8B}$, $\texttt{R1-D-Qwen-1.5B}$, and $\texttt{R1-D-Qwen-7B}$, respectively. Then, we determine the per-output token generation cost by assuming that the per-token prices are $25\%$ higher than the per-token costs, which simply corresponds to the per-token margin of providers.
Lastly, for each player (provider) $i$ in the game $\Gcal$, each test-time method with compute level $\theta$, and each dataset, we separately compute the quantities $p_i(\theta)$ (and correspondingly $c_i(\theta_i)$ using the margin $25\%$) by considering the average number of output tokens across all model outputs in the dataset and multiplying it by the per-token price. We additionally report inefficiency results for margins of $5\%$ and $15\%$ in Appendix~\ref{app:auction-tables} to assess the sensitivity of our findings to this choice.

\xhdr{Dynamics of test-time compute games}
To simulate the dynamics of a test-time compute game, we proceed as follows. First, we consider that all providers start with their lowest test-time compute level, that is, $\theta_i^1 = \min_{\theta \in \Theta} \theta, \forall i \in [N]$, which corresponds to using a single sample for best-of-n and majority voting, and generating responses in the first $20$-th percentile in terms of reasoning tokens for chain-of-thought, see Appendix~\ref{app:experimental-details}. Then, at each iteration $t$, we update the providers' compute levels from $\thetab^t$ to $\thetab^{t+1}$ by randomly selecting a provider $i_t\in[N]$ who does not have maximum utility when the others select $\thetab^t_{-i_t}$. We keep $\thetab^{t+1}_{-i_t}=\thetab^{t}_{-i_t}$ fixed, and take $\theta_{i_t}^{t+1}$ to be the smallest compute that is larger than $\theta_{i_t}^{t}$ if $\theta_{i_t}^{t}\prec\argmax_\theta U_{i_t}(\theta;\thetab_{-i_t}^{t})$, or to be the largest compute that is smaller than $\theta_{i_t}^{t}$ if $\theta_{i_t}^{t}\succ \argmax_\theta U_{i_t}(\theta;\thetab_{-i_t}^{t})$. If all providers are already maximizing their utilities, the better-response dynamics have reached a Nash equilibrium. At each time step, we use the values offered by providers to determine their market share and compute the price of anarchy according to Eq.~\ref{eq:PoA}, see Appendix~\ref{app:dynamics-results}.

\xhdr{Rationality parameter $\beta$} In Section~\ref{sec:experiments} we focus on a rationality parameter $\beta = 1000$, and experiment with different values in Appendix~\ref{app:additional-experimental-results}. We choose the value $1000$ based real behavioral measurements. Indeed, the parameter $\beta$ is, as discussed in Section~\ref{sec:model}, the standard rationality constant of quantal response equilibria, whose empirical value has been estimated in the behavioral literature to lie in the range $4 --- 16$ when utilities are normalized to $[0,1]$ \citep{ONeill1987, MCKELVEY19956, Lieberman2007, PITA20101142}. In our setting, user values are typically centered around $0.005$ USD (see, \eg, Figure~\ref{fig:gsm8k:v_curves_comparison}), so after normalizing to our value scale this empirical range corresponds to $\beta$ in the range $4/0.005 = 800$ to $16/0.005 = 3200$.

\xhdr{Licenses}
The \texttt{Llama-3} models and \texttt{ArmoRM-Llama3-8B-v0.1} are licensed under the LLAMA 3.2 COMMUNITY LICENSE AGREEMENT.\footnote{\url{https://www.llama.com/llama3/license/}}
The \texttt{Llama-3.1} models are licensed under the LLAMA 3.2 COMMUNITY LICENSE AGREEMENT.\footnote{\url{https://www.llama.com/llama3_1/license/}}
The \texttt{Llama-3.2} models are licensed under the LLAMA 3.2 COMMUNITY LICENSE AGREEMENT.\footnote{\url{https://www.llama.com/llama3_2/license/}}
The \texttt{Qwen} models are licensed under the Tongyi Qianwen LICENSE AGREEMENT.\footnote{\url{https://github.com/QwenLM/Qwen/blob/main/Tongyi\%20Qianwen\%20LICENSE\%20AGREEMENT/}}
The \texttt{DeepSeek-R1-Distill-Llama-8B}, \texttt{DeepSeek-R1-Distill-Qwen-1.5B} and \texttt{DeepSeek-R1-Distill-Qwen-7B} models are licensed under the  MIT License.
The \texttt{GPQA} and \texttt{AIME} datasets are licensed under Creative Commons Attribution 4.0, \texttt{GSM8K} is licensed under MIT License, and \texttt{TruthfulQA} is licensed under Apache License 2.0.

\clearpage
\newpage

\section{Additional Experimental Results}\label{app:additional-experimental-results}
This section contains additional experimental results that complement those discussed in Section~\ref{sec:experiments}. In Appendix~\ref{app:model-evaluation}, we summarize the results (quality and number of generated tokens) of evaluating the LLMs served by each provider in the test-time compute games on \texttt{GSM8K}, \texttt{AIME}, \texttt{GPQA}, and \texttt{TruthfulQA}. Appendix~\ref{app:dynamics-results} shows additional outcomes of test-time compute games across all datasets and test-time compute methods. Further, Appendix~\ref{app:auction-tables} compares the equilibrium outcomes of the
test-time compute game $\Gcal$ with the outcomes of the auction mechanism
$\widetilde{\Gcal}$ introduced in Section~\ref{sec:auction}, across all
datasets, test-time compute methods, profit margins ($5\%$, $15\%$, and
$25\%$), and accuracy-to-value conversion parameter values $\alpha$.

\subsection{Model evaluation}\label{app:model-evaluation}

Here, we report, for each dataset and each test-time compute method, the average quality of the models as a function of the test-time compute $\theta$. We also report the average number of generated tokens to obtain the response to each query. See Appendix~\ref{app:experimental-details} for more details regarding the generation of model outputs.

\subsubsection{Model Evaluation on GSM8K}\label{app:model-eval-gsm8k}

\begin{figure}[h!]
    \centering
    \subfloat{
        \centering
        \includegraphics[width=0.9\linewidth]{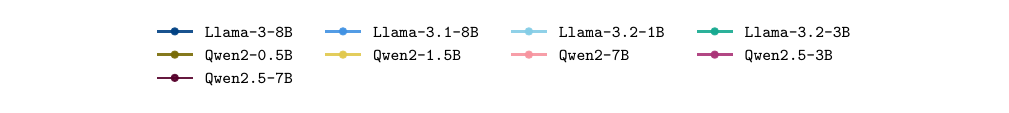}}\\
    \vspace{-0.3cm}
    \subfloat{
        \includegraphics[width=0.31\linewidth, keepaspectratio, trim={0cm 1cm 0cm 0cm}]{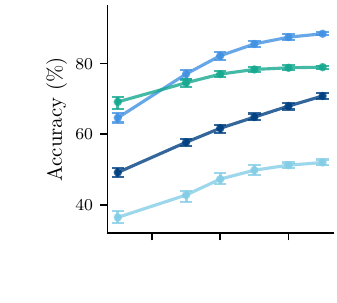}
    }
    \subfloat{
        \includegraphics[width=0.31\linewidth,keepaspectratio, trim={0cm 1cm 0cm 0cm}, clip]{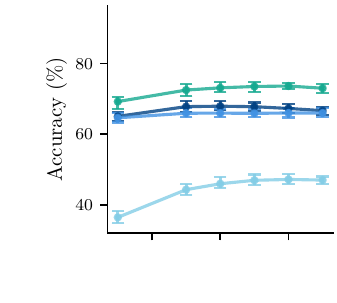}
    }
    \subfloat{
        \includegraphics[width=0.31\linewidth,keepaspectratio, trim={0cm 1cm 0cm 0cm}, clip]{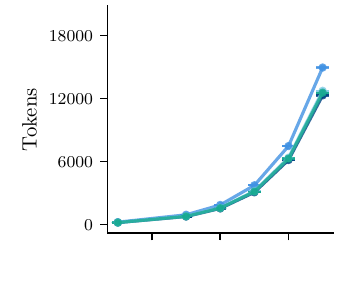}
    }\\
    \vspace{0.5cm}
    \setcounter{subfigure}{0}
    \subfloat[Majority Vote]{
        \includegraphics[width=0.31\linewidth,keepaspectratio, trim={0cm 0cm 0cm 0cm}, clip]{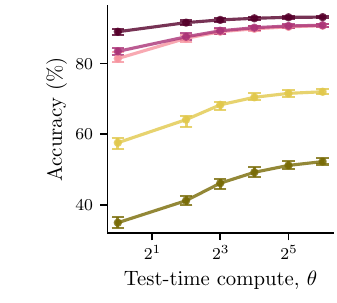}  }
    \subfloat[Best-of-n]{
        \includegraphics[width=0.31\linewidth,keepaspectratio, trim={0cm 0cm 0cm 0cm}, clip]{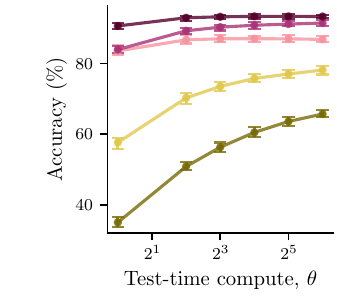}     }
    \subfloat[Number of Tokens]{
        \includegraphics[width=0.31\linewidth,keepaspectratio, trim={0cm 0cm 0cm 0cm}, clip]{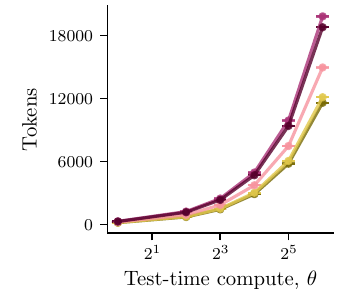}    }
    \caption{\xhdr{Accuracy of \texttt{Llama} and \texttt{Qwen} models on GSM8K using majority voting and best-of-n}
        Panels (a) and (b) show the average accuracy of various LLMs from the \texttt{Llama} and \texttt{Qwen} families over questions from the GSM8K dataset, where the responses of the models are obtained using majority voting or best-of-n, respectively. Panel (c) shows, as a function of the number of samples used to generate the response, the total number of tokens that the models generate to obtain the response to each question, averaged across questions.
        We show $95 \%$ confidence intervals obtained by bootstrapping $50$ times.
    }
    \label{fig:gsm8k_all_unreasoning}
\end{figure}

\begin{figure}[h!]
    \centering
    \subfloat{
        \centering
        \includegraphics[width=0.9\linewidth]{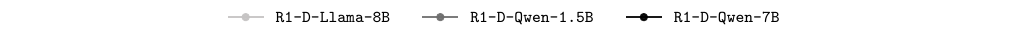}
    }\\
    \setcounter{subfigure}{0}
    \subfloat[Accuracy vs Reasoning Effort]{
        \includegraphics[width=0.48\linewidth, keepaspectratio, trim={0cm 0.5cm 0cm 0cm}, clip]{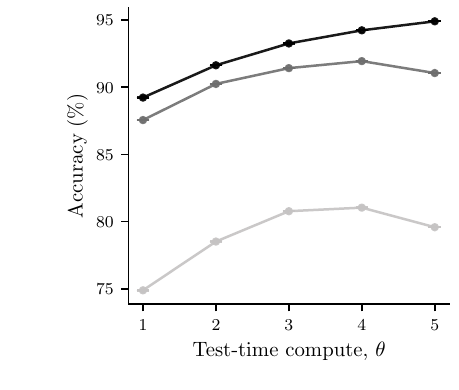}
    }
    \hfill 
    \subfloat[Tokens vs Reasoning Effort]{
        \includegraphics[width=0.48\linewidth, keepaspectratio, trim={0cm 0.5cm 0cm 0cm}, clip]{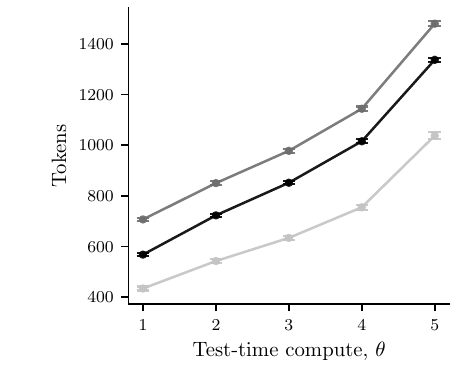}
    }
    \caption{\xhdr{Accuracy of reasoning models distilled from \texttt{DeepSeek-R1} on GSM8K using chain-of-thought}
        Panel (a) shows the average accuracy of various reasoning models from the \texttt{Llama} and \texttt{Qwen} families distilled from \texttt{DeepSeek-R1} over questions from the GSM8K dataset. Here, the reasoning effort is defined by binning the model outputs into quantiles based on the number of reasoning tokens (see Appendix~\ref{app:experimental-details}).
        Panel (c) shows, as a function of the reasoning effort, the total number of tokens (including reasoning and non-reasoning tokens) that the models generate as a response to each question, averaged across questions.
        We show $95 \%$ confidence intervals obtained by bootstrapping $50$ times. Refer to Appendix~\ref{app:experimental-details} for further details regarding the evaluation of the models.
        The reported accuracies are computed only over outputs in which the model generates at least $5$ reasoning tokens and does not exceed the maximum generation length of $2048$ tokens (see Appendix~\ref{app:experimental-details} for more details).
    }
    \label{fig:gsm8k_reasoning_models}
\end{figure}

\clearpage
\newpage

\subsubsection{Model Evaluation on AIME}

\begin{figure}[h!]
    \centering
    \subfloat{
        \centering
        \includegraphics[width=0.9\linewidth]{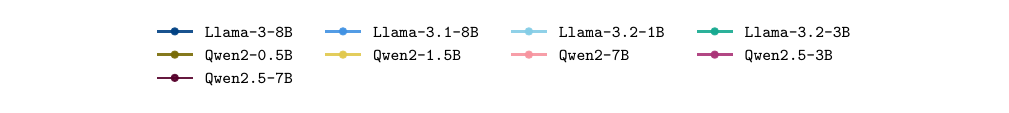}}\\
    \vspace{-0.3cm}
    \subfloat{
        \includegraphics[width=0.31\linewidth, keepaspectratio, trim={0cm 1cm 0cm 0cm}]{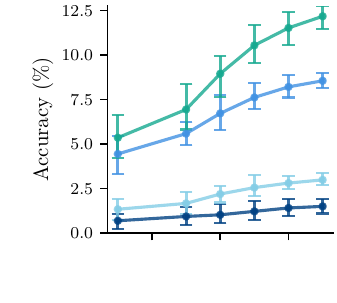}
    }
    \subfloat{
        \includegraphics[width=0.31\linewidth,keepaspectratio, trim={0cm 1cm 0cm 0cm}, clip]{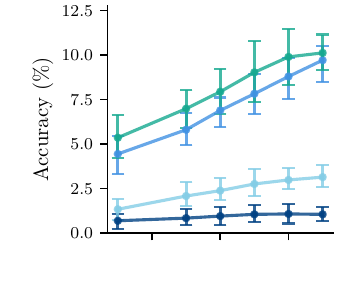}
    }
    \subfloat{
        \includegraphics[width=0.31\linewidth,keepaspectratio, trim={0cm 1cm 0cm 0cm}, clip]{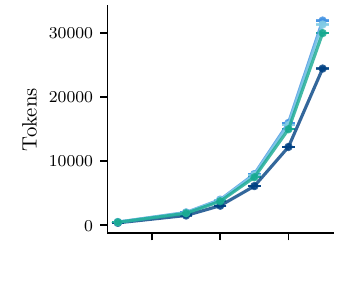}
    }\\
    \vspace{0.5cm}
    \setcounter{subfigure}{0}
    \subfloat[Majority Vote]{
        \includegraphics[width=0.31\linewidth,keepaspectratio, trim={0cm 0cm 0cm 0cm}, clip]{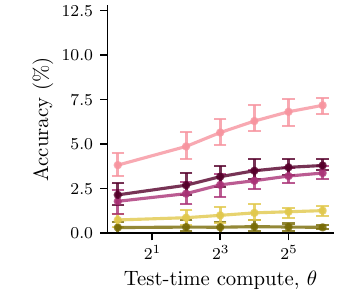}  }
    \subfloat[Best-of-n]{
        \includegraphics[width=0.31\linewidth,keepaspectratio, trim={0cm 0cm 0cm 0cm}, clip]{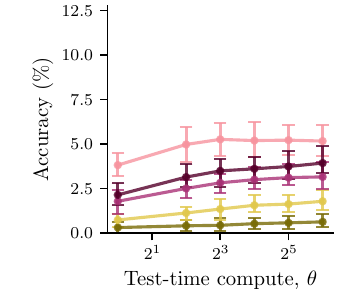}     }
    \subfloat[Number of Tokens]{
        \includegraphics[width=0.31\linewidth,keepaspectratio, trim={0cm 0cm 0cm 0cm}, clip]{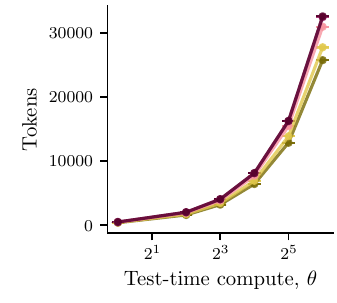}    }
    \caption{\xhdr{Accuracy of \texttt{Llama} and \texttt{Qwen} models on AIME using majority voting and best-of-n}
        Panels (a) and (b) show the average accuracy of various LLMs from the \texttt{Llama} and \texttt{Qwen} families over questions from the AIME dataset, where the responses of the models are obtained using majority voting or best-of-n, respectively. Panel (c) shows, as a function of the number of samples used to generate the response, the total number of tokens that the models generate to obtain the response to each question, averaged across questions. Here, we compute the accuracies for majority voting, and best-of-n are computed across the same outputs, and hence both majority voting and best-of-n generate the exact same number of average tokens.
        We show $95 \%$ confidence intervals obtained by bootstrapping $50$ times. Refer to Appendix~\ref{app:experimental-details} for further details regarding the evaluation of the models.
    }
    \label{fig:aime_all_unreasoning}
\end{figure}

\begin{figure}[h!]
    \centering
    \subfloat{
        \centering
        \includegraphics[width=0.9\linewidth]{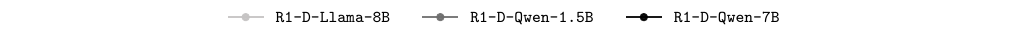}
    }\\
    \setcounter{subfigure}{0}
    \subfloat[Accuracy vs Reasoning Effort]{
        \includegraphics[width=0.48\linewidth, keepaspectratio, trim={0cm 0.5cm 0cm 0cm}, clip]{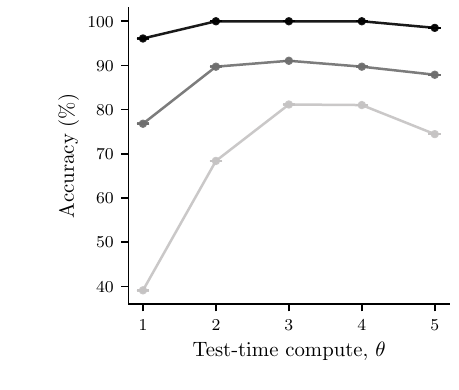}
    }
    \hfill 
    \subfloat[Tokens vs Reasoning Effort]{
        \includegraphics[width=0.48\linewidth, keepaspectratio, trim={0cm 0.5cm 0cm 0cm}, clip]{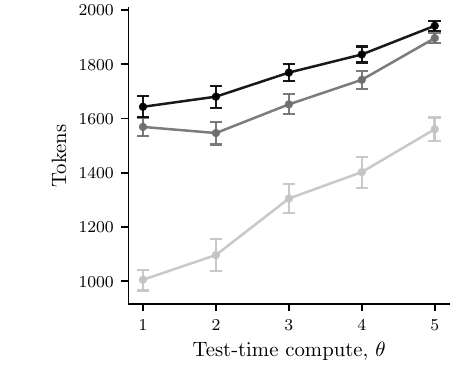}
    }
    
    \caption{\xhdr{Accuracy of reasoning models distilled from \texttt{DeepSeek-R1} on AIME using chain-of-thought}
        Panel (a) shows the average accuracy of various reasoning models from the \texttt{Llama} and \texttt{Qwen} families distilled from \texttt{DeepSeek-R1} over questions from the AIME dataset. Here, the reasoning effort is defined by binning the model outputs into quantiles based on the number of reasoning tokens (see Appendix~\ref{app:experimental-details}).
        Panel (c) shows, as a function of the reasoning effort, the total number of tokens (including reasoning and non-reasoning tokens) that the models generate as a response to each question, averaged across questions.
        We show $95 \%$ confidence intervals obtained by bootstrapping $50$ times. Refer to Appendix~\ref{app:experimental-details} for further details regarding the evaluation of the models.
        The reported accuracies are computed only over outputs in which the model generates at least $5$ reasoning tokens and does not exceed the maximum generation length of $2048$ tokens (see Appendix~\ref{app:experimental-details} for more details).
    }
    \label{fig:aime_reasoning_models}
\end{figure}

\clearpage
\newpage

\subsubsection{Model Evaluation on GPQA}

\begin{figure}[h!]
    \centering
    \subfloat{
        \centering
        \includegraphics[width=0.9\linewidth]{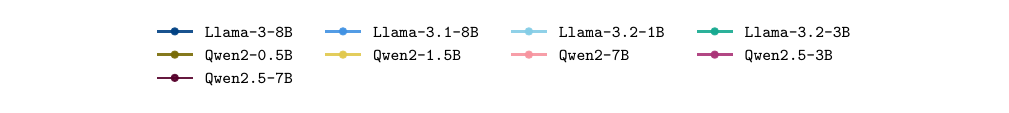}}\\
    \vspace{-0.3cm}
    \subfloat{
        \includegraphics[width=0.31\linewidth, keepaspectratio, trim={0cm 1cm 0cm 0cm}]{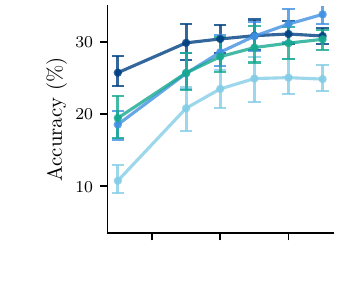}
    }
    \subfloat{
        \includegraphics[width=0.31\linewidth,keepaspectratio, trim={0cm 1cm 0cm 0cm}, clip]{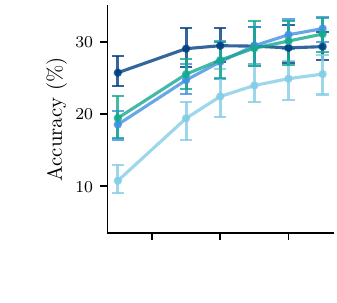}
    }
    \subfloat{
        \includegraphics[width=0.31\linewidth,keepaspectratio, trim={0cm 1cm 0cm 0cm}, clip]{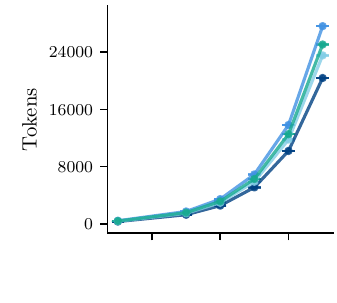}
    }\\
    \vspace{0.5cm}
    \setcounter{subfigure}{0}
    \subfloat[Majority Vote]{
        \includegraphics[width=0.31\linewidth,keepaspectratio, trim={0cm 0cm 0cm 0cm}, clip]{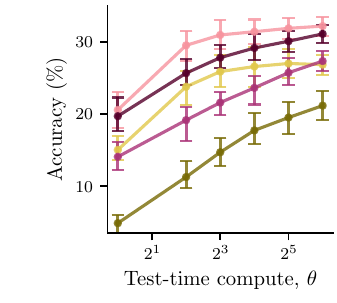}  }
    \subfloat[Best-of-n]{
        \includegraphics[width=0.31\linewidth,keepaspectratio, trim={0cm 0cm 0cm 0cm}, clip]{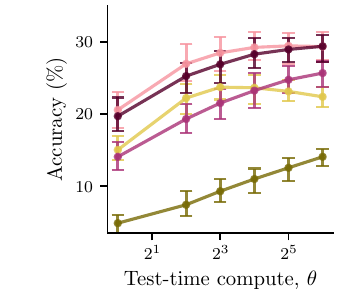}     }
    \subfloat[Number of Tokens]{
        \includegraphics[width=0.31\linewidth,keepaspectratio, trim={0cm 0cm 0cm 0cm}, clip]{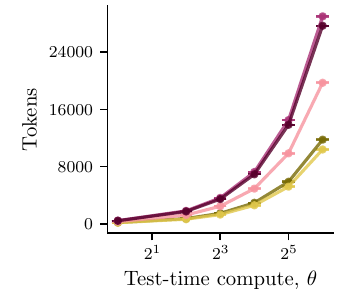}    }
    \caption{\xhdr{Accuracy of \texttt{Llama} and \texttt{Qwen} models on GPQA using majority voting and best-of-n}
        Panels (a) and (b) show the average accuracy of various LLMs from the \texttt{Llama} and \texttt{Qwen} families over questions from the GPQA dataset, where the responses of the models are obtained using majority voting or best-of-n, respectively. Panel (c) shows, as a function of the number of samples used to generate the response, the total number of tokens that the models generate to obtain the response to each question, averaged across questions. Here, we compute the accuracies for majority voting, and best-of-n are computed across the same outputs, and hence both majority voting and best-of-n generate the exact same number of average tokens.
        We show $95 \%$ confidence intervals obtained by bootstrapping $50$ times. Refer to Appendix~\ref{app:experimental-details} for further details regarding the evaluation of the models.
    }
    \label{fig:gpqa_all_unreasoning}
\end{figure}

\begin{figure}[h!]
    \centering
    \subfloat{
        \centering
        \includegraphics[width=0.9\linewidth]{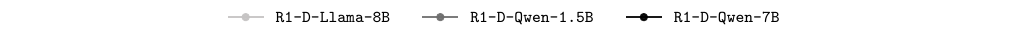}
    }\\
    \setcounter{subfigure}{0}
    \subfloat[Accuracy vs Reasoning Effort]{
        \includegraphics[width=0.48\linewidth, keepaspectratio, trim={0cm 0.5cm 0cm 0cm}, clip]{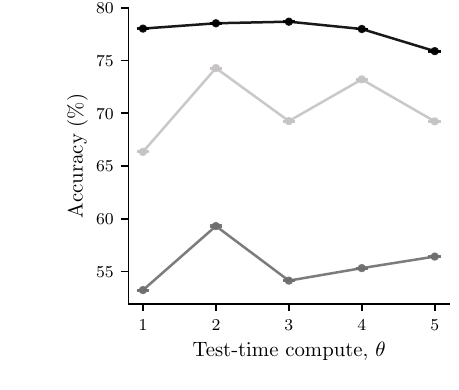}
    }
    \hfill 
    \subfloat[Tokens vs Reasoning Effort]{
        \includegraphics[width=0.48\linewidth, keepaspectratio, trim={0cm 0.5cm 0cm 0cm}, clip]{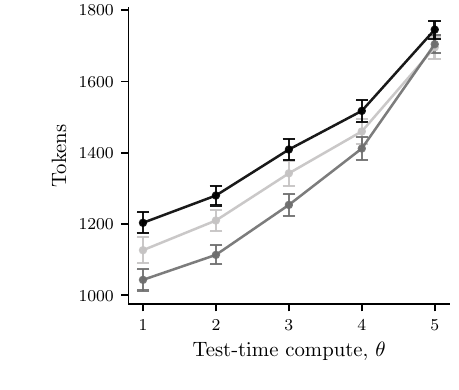}
    }
    
    \caption{\xhdr{Accuracy of reasoning models distilled from \texttt{DeepSeek-R1} on GPQA using chain-of-thought}
        Panel (a) shows the average accuracy of various reasoning models from the \texttt{Llama} and \texttt{Qwen} families distilled from \texttt{DeepSeek-R1} over questions from the GPQA dataset. Here, the reasoning effort is defined by binning the model outputs into quantiles based on the number of reasoning tokens (see Appendix~\ref{app:experimental-details}).
        Panel (c) shows, as a function of the reasoning effort, the total number of tokens (including reasoning and non-reasoning tokens) that the models generate as a response to each question, averaged across questions.
        We show $95 \%$ confidence intervals obtained by bootstrapping $50$ times. Refer to Appendix~\ref{app:experimental-details} for further details regarding the evaluation of the models.
        The reported accuracies are computed only over outputs in which the model generates at least $5$ reasoning tokens and does not exceed the maximum generation length of $2048$ tokens (see Appendix~\ref{app:experimental-details} for more details).
    }
    \label{fig:gpqa_reasoning_models}
\end{figure}

\clearpage
\newpage

\subsubsection{Model Evaluation on TruthfulQA}

\begin{figure}[h!]
    \centering
    \subfloat{
        \centering
        \includegraphics[width=0.9\linewidth]{figures/GPQA/dynamics/unreasoning/best-of-n/values/legend.pdf}}\\
    \vspace{-0.3cm}
    \subfloat{
        \includegraphics[width=0.31\linewidth,keepaspectratio, trim={0cm 0cm 0cm 0cm}, clip]{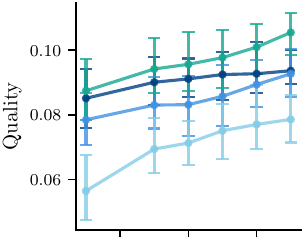}
    }
    \subfloat{
        \includegraphics[width=0.31\linewidth,keepaspectratio, trim={0cm 0cm 0cm 0cm}, clip]{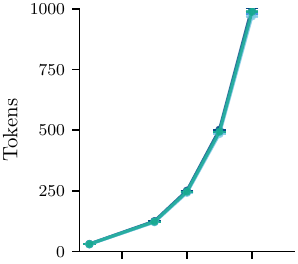}
    }\\
    \vspace{0.5cm}
    \setcounter{subfigure}{0}
    \subfloat[Best-of-n]{
        \includegraphics[width=0.31\linewidth,keepaspectratio, trim={0cm 0cm 0cm 0cm}, clip]{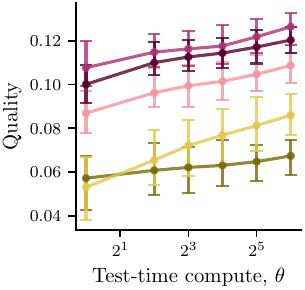}     }
    \subfloat[Number of Tokens]{
        \includegraphics[width=0.31\linewidth,keepaspectratio, trim={0cm 0cm 0cm 0cm}, clip]{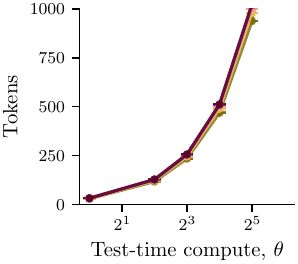}    }
    \caption{\xhdr{Quality of \texttt{Llama} and \texttt{Qwen} models on TruthfulQA using best-of-n}
        Panel (a) shows the average quality of various LLMs from the \texttt{Llama} and \texttt{Qwen} families over questions from the TruthfulQA dataset, where the responses of the models are obtained using best-of-n. To measure quality, we use a \texttt{BLEURT}-based score defined as the difference between the maximum similarity to a true reference answer and the maximum similarity to a false reference answer. Panel (b) shows, as a function of the number of samples used to generate the response, the total number of tokens that the models generate to obtain the response to each question, averaged across questions.
        We show $95 \%$ confidence intervals obtained by bootstrapping $50$ times. Refer to Appendix~\ref{app:experimental-details} for further details regarding the evaluation of the models.
    }
    \label{fig:truhthfullqa_all_unreasoning}
\end{figure}

\clearpage
\newpage

\subsection{Dynamics and Equilibria of Test-Time Compute Games}\label{app:dynamics-results}

Here, for each dataset (\texttt{GSM8K}, \texttt{AIME}, \texttt{GPQA}, \texttt{TruthfulQA}) and each test-time compute method (majority voting and best-of-n for non-reasoning models, and chain-of-thought for reasoning models), we report: (i) the value offered by providers as a function of their test time compute, (ii) the compute levels and market shares when providers better-respond to each other, as a function of the iteration $t$, (iii) the evolution of the potential $\Phi$ (Eq.~\ref{eq:potential}) when providers better-respond, and (iv), the inefficiency (PoA($\Gcal$)-1) at equilibrium for each game as a function of user's rationality $\beta$. We also report the inefficiency at each time step, which corresponds to $\max_{\thetab} \SW(\thetab)/\SW(\thetab^t)$ whenever $\thetab^t$ is not an equilibrium.

\subsubsection{Test-Time Compute Equilibria on GSM8K}\label{app:dynamics-results-gsm8k}

\begin{figure}[h!]
    \centering
    \subfloat{
        \centering
        \includegraphics[width=0.6\linewidth]{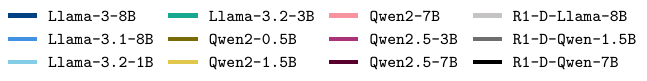}
    }\\
    \setcounter{subfigure}{0}
    \centering
    \begin{subfigure}{0.32\textwidth}
        \centering
        \includegraphics[width=\linewidth]{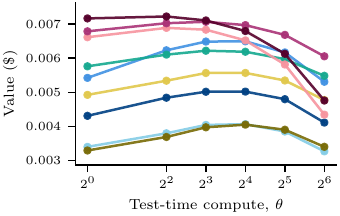}
        \caption{Majority voting}
        \label{fig:gsm8k:v_curves:majority}
    \end{subfigure}
    \hfill
    \begin{subfigure}{0.32\textwidth}
        \centering
        \includegraphics[width=\linewidth]{figures/GSM8K/dynamics/unreasoning/best-of-n/values/v_curves.pdf}
        \caption{Best-of-n}
        \label{fig:gsm8k:v_curves:bon}
    \end{subfigure}
    \hfill
    \begin{subfigure}{0.32\textwidth}
        \centering
        \includegraphics[width=\linewidth]{figures/GSM8K/dynamics/reasoning/relative_efforts/values/v_curves.pdf}
        \caption{Chain-of-thought}
        \label{fig:gsm8k:v_curves:reasoning}
    \end{subfigure}
    
    \caption{{\bf User values offered by providers in a test-time compute game on GSM8K.}
     The figure shows the user values $V_i(\theta)$ offered by providers in a test-time compute game $\Gcal$, as a function of their test-time compute $\theta$.
     Panel~(a) and (b) correspond to games with $N=9$ providers serving non-reasoning models from the \texttt{Llama} and \texttt{Qwen} families, where providers use, respectively, majority voting and best-of-n across $\theta$ samples. Panel (c) corresponds to a game with $N=3$ providers serving reasoning models distilled from \texttt{DeepSeek-R1}, where $\theta$ represents reasoning effort, defined by binning the model outputs into quantiles based on the number of reasoning tokens (see Appendix~\ref{app:experimental-details}).
     In both games, providers serve queries $Q$ from the \texttt{GSM8K} dataset, we set $\beta = 1000$ and consider that each (average) point of accuracy offers a value of $\$0.008$ to the users.
    }
    \label{fig:gsm8k:v_curves_comparison}
\end{figure}

\begin{figure}[h!]
    \centering
    \subfloat{
        \centering
        \includegraphics[width=0.7\linewidth]{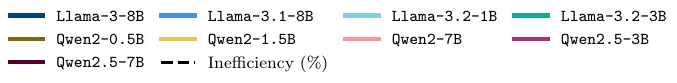}
    }

    \subfloat{
        \includegraphics[width=0.3\linewidth]{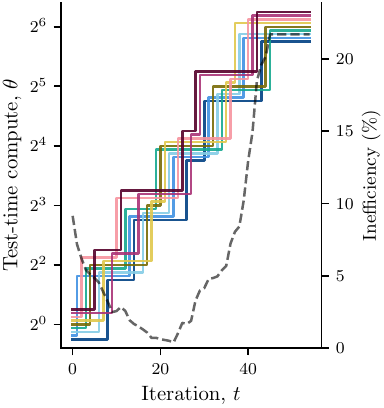}
    }
    \subfloat{
        \includegraphics[width=0.3\linewidth]{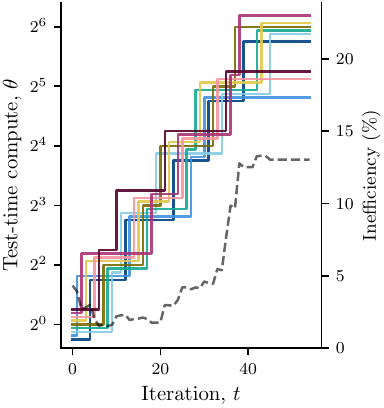}
    }
    \subfloat{
        \includegraphics[width=0.3\linewidth]{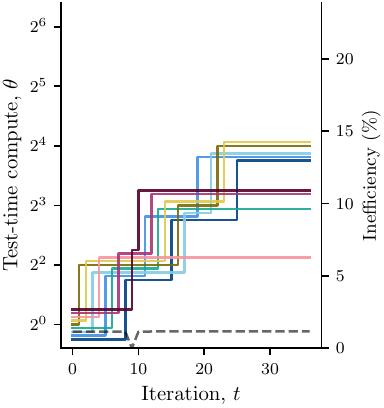}
    }\\
    \setcounter{subfigure}{0}
    \hspace{-5mm}
    \subfloat[$\beta=2\cdot10^2$]{
        \includegraphics[width=0.3\linewidth]{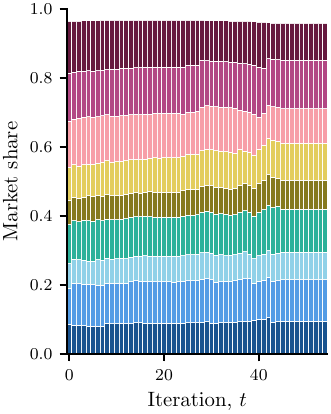}
        }
    \subfloat[$\beta=10^3$]{
        \includegraphics[width=0.3\linewidth]{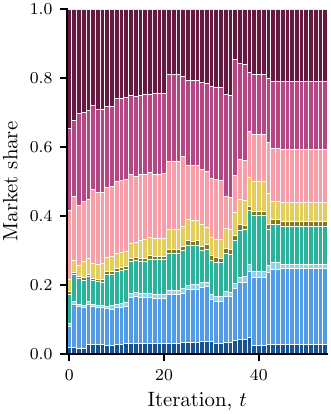}
    }
    \subfloat[$\beta=10^5$]{
        \includegraphics[width=0.3\linewidth]{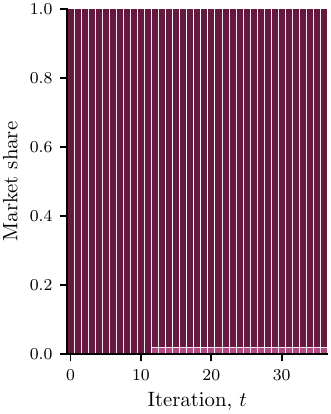}
    }
    \caption{{\bf Dynamics of a test-time compute game using majority-voting.}
     The figure shows, for different levels of user rationality $\beta$, the better-response dynamics of a test-time compute game $\Gcal$ where $N=9$ providers sequentially select a test-time compute level that increases their utility. The upper panels show the compute levels $\theta$ selected by each provider and the resulting market inefficiency $(\text{PoA}(\Gcal) - 1)$, and the lower panels show the market share of each provider.
     Here, all providers use majority-voting across $\theta$ samples as their test-time compute method to serve queries $Q$ from the \texttt{GSM8K} dataset. We consider that providers operate with a margin of $25\%$ between per-token price and per-token cost, and that each (average) point of accuracy offers a value of $\$0.008$ to the users.
}
    \label{fig:market-dynamics-gsm8k-unreasoning-maj}
\end{figure}

\begin{figure}[h!]
    \centering
    \subfloat[Potential ($\beta=2\cdot10^2$)]{
        \includegraphics[width=0.3\linewidth]{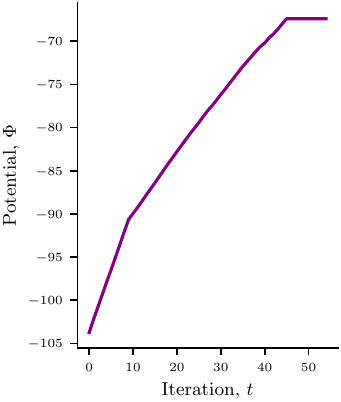}
    }
    \subfloat[Potential ($\beta=10^3$)]{
        \includegraphics[width=0.3\linewidth]{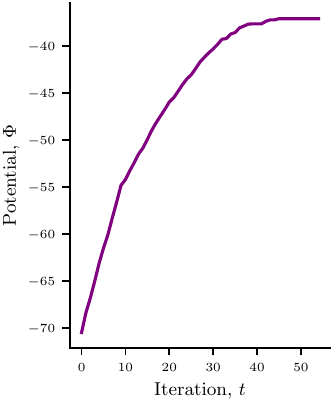}
    }
    \subfloat[Potential ($\beta=10^5$)]{
        \includegraphics[width=0.3\linewidth]{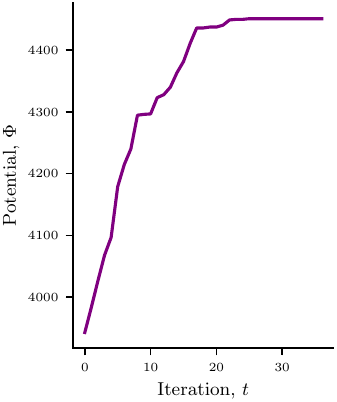}
    }
    \caption{{\bf Potential of a test-time compute game using majority-voting.}
     The figure shows, for different levels of user rationality $\beta$, the evolution of the potential $\Phi$ (see Eq.~\ref{eq:potential}) in the test-time compute games in Figure~\ref{fig:market-dynamics-gsm8k-unreasoning-maj} where $N=9$ providers sequentially select a test-time compute level that increases their utility.
     }
    \label{fig:potential-gsm8k-unreasoning-maj}
\end{figure}

\begin{figure}
    \centering
    \includegraphics[width=0.7\linewidth]{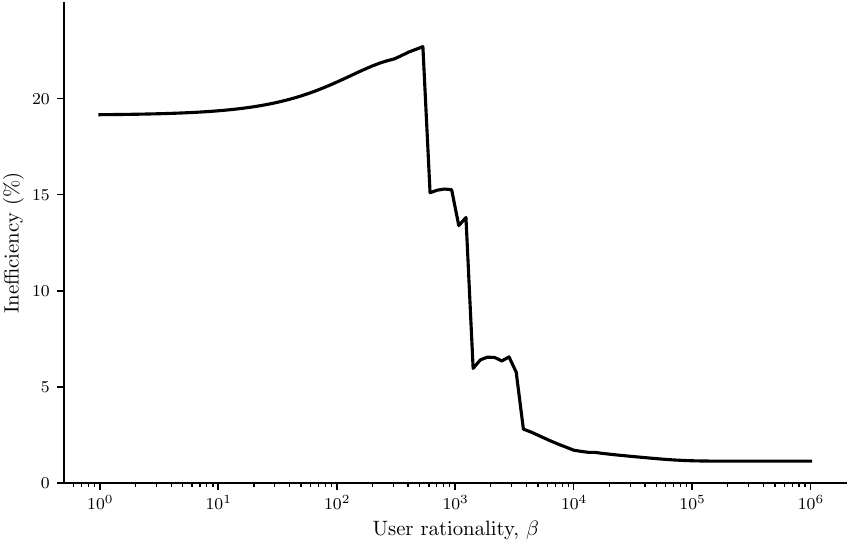}
    \caption{{\bf Inefficiency of a test-time compute game using majority-voting.}
     The figure shows, as a function of users' rationality $\beta$, the inefficiency ($\text{PoA}(\Gcal)-1$) of the test-time compute game in Figure~\ref{fig:market-dynamics-gsm8k-unreasoning-maj}, where $N=9$ providers sequentially select a test-time compute level that increases their utility.
     }
    \label{fig:poa-beta-gsm8k-unreasoning-maj}
\end{figure}

\begin{figure}[h!]
    \centering
    \subfloat{
        \centering
        \includegraphics[width=0.7\linewidth]{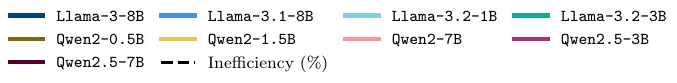}
    }

    \subfloat{
        \includegraphics[width=0.3\linewidth]{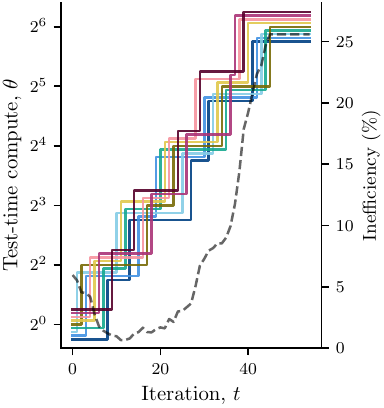}
    }
    \subfloat{
        \includegraphics[width=0.3\linewidth]{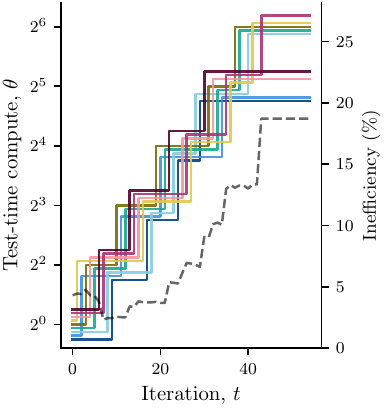}
    }
    \subfloat{
        \includegraphics[width=0.3\linewidth]{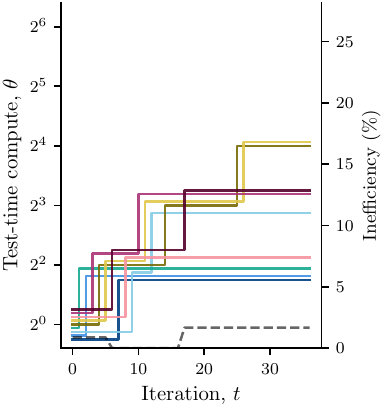}
    }\\
    \setcounter{subfigure}{0}
    \hspace{-5mm}
    \subfloat[$\beta=2\cdot10^2$]{
        \includegraphics[width=0.3\linewidth]{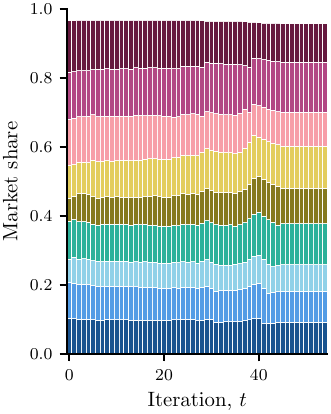}
        }
    \subfloat[$\beta=10^3$]{
        \includegraphics[width=0.3\linewidth]{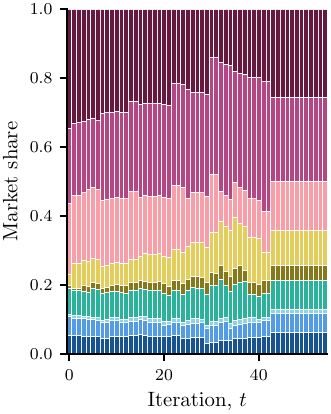}
    }
    \subfloat[$\beta=10^5$]{
        \includegraphics[width=0.3\linewidth]{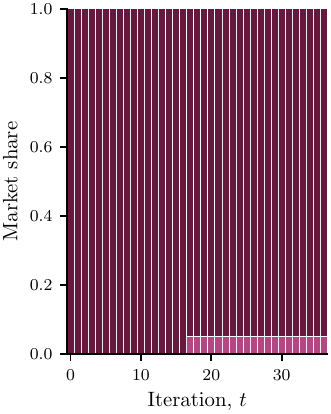}
    }
    \caption{\xhdr{Better-response dynamics of a test-time compute game using best-of-n} 
    The upper panels show, for varying levels of user rationality $\beta$, better-response dynamics when providers serving models from the \texttt{Llama} and \texttt{Qwen} families use best-of-n to serve queries from the GSM8K dataset.
    The solid colored lines (left $y$-axis) represent the test-time compute $\theta$ selected by each provider at each iteration, corresponding to the number of samples used for best-of-n. 
    The dashed black line (right $y$-axis) tracks the Market Inefficiency, defined as $(\text{PoA} - 1) \times 100$ (see Eq.~\ref{eq:PoA}).
    The lower panels show the evolution of the market share at each time step of the better-response dynamics.
    The initial compute level $\thetab^1$ is taken as the lowest possible compute.
    We apply a small vertical jitter to the strategy lines to distinguish overlapping providers and take a fixed profit margin of $25\%$.
}
    \label{fig:market-dynamics-gsm8k-unreasoning-bon}
\end{figure}

\begin{figure}[h!]
    \centering
    \subfloat[Potential ($\beta=2\cdot10^2$)]{
        \includegraphics[width=0.3\linewidth]{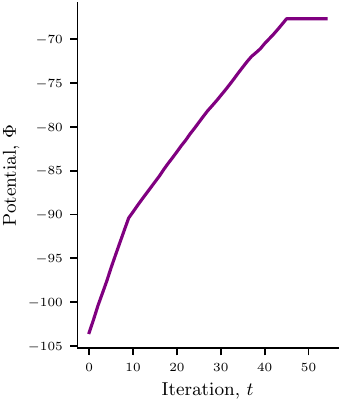}
    }
    \subfloat[Potential ($\beta=10^3$)]{
        \includegraphics[width=0.3\linewidth]{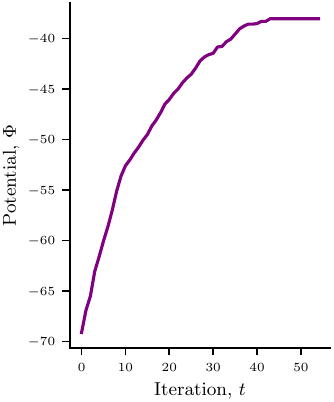}
    }
    \subfloat[Potential ($\beta=10^5$)]{
        \includegraphics[width=0.3\linewidth]{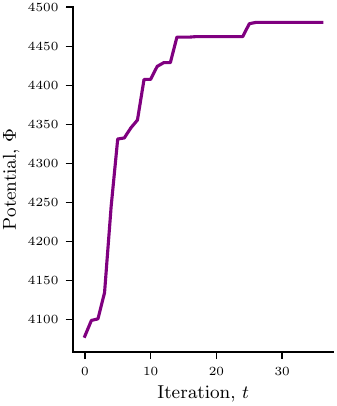}
    }
    \caption{{\bf Potential of a test-time compute game using best-of-n.}
     The figure shows, for different levels of user rationality $\beta$, the evolution of the potential $\Phi$ (see Eq.~\ref{eq:potential}) in the test-time compute games in Figure~\ref{fig:market-dynamics-gsm8k-unreasoning-bon} where $N=9$ providers sequentially select a test-time compute level that increases their utility.
     }
    \label{fig:potential-gsm8k-unreasoning-bon}
\end{figure}

\begin{figure}
    \centering
    \includegraphics[width=0.8\linewidth]{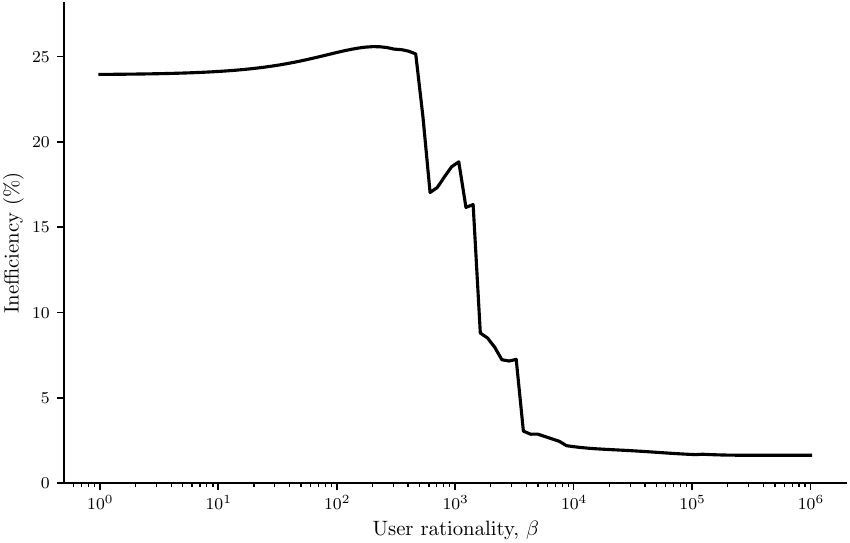}
    \caption{{\bf Inefficiency of a test-time compute game using best-of-n.}
     The figure shows, as a function of users' rationality $\beta$, the inefficiency ($\text{PoA}(\Gcal)-1$) of the test-time compute game in Figure~\ref{fig:market-dynamics-gsm8k-unreasoning-bon}, where $N=9$ providers sequentially select a test-time compute level that increases their utility.
     }
    \label{fig:poa-beta-gsm8k-unreasoning-bon}
\end{figure}

\begin{figure}[h!]
    \centering
    \subfloat{
        \centering
        \includegraphics[width=0.9\linewidth]{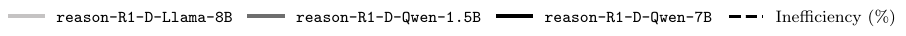}
    }

    \subfloat{
        \includegraphics[width=0.3\linewidth]{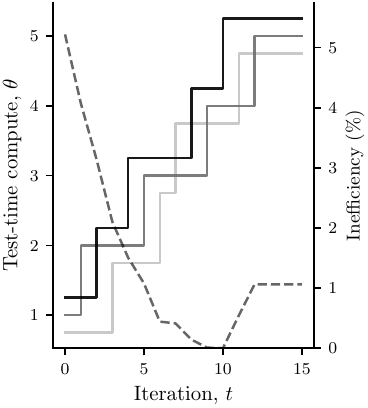}
    }
    \subfloat{
        \includegraphics[width=0.3\linewidth]{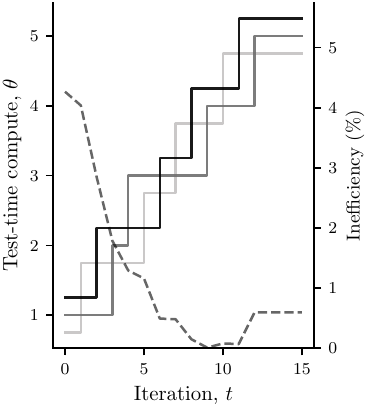}
    }
    \subfloat{
        \includegraphics[width=0.3\linewidth]{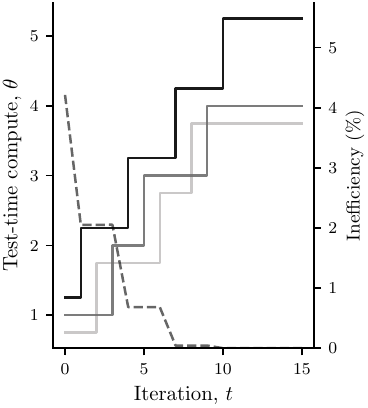}
    }\\
    \setcounter{subfigure}{0}
    \hspace{-5mm}
    \subfloat[$\beta=2\cdot10^2$]{
        \includegraphics[width=0.3\linewidth]{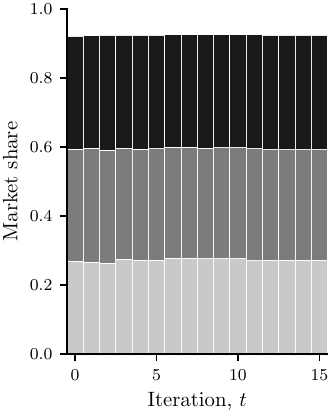}
        }
    \subfloat[$\beta=10^3$]{
        \includegraphics[width=0.3\linewidth]{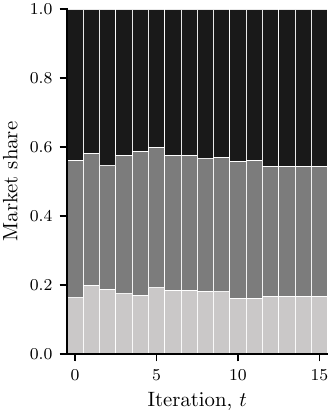}
    }
    \subfloat[$\beta=10^5$]{
        \includegraphics[width=0.3\linewidth]{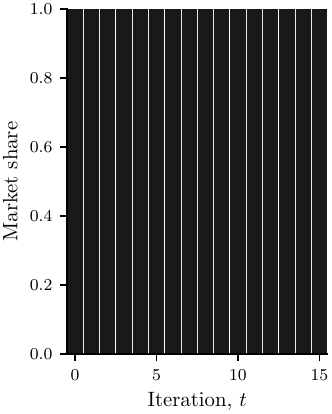}
    }
    \caption{\textbf{Better-response dynamics of a test-time compute game using CoT.} 
    The upper panels show, for varying levels of user rationality $\beta$, better-response dynamics when providers serving models from the \texttt{Llama} and \texttt{Qwen} families distilled from \texttt{DeepSeek-R1} use chain-of-thought to serve queries from the GSM8K dataset.
    The solid colored lines (left $y$-axis) represent the test-time compute $\theta$ selected by each provider at each iteration, corresponding to the reasoning effort used. 
    The dashed black line (right $y$-axis) tracks the Market Inefficiency, defined as $(\text{PoA} - 1) \times 100$ (see Eq.~\ref{eq:PoA}).
    The lower panels show the evolution of the market share at each time step of the better-response dynamics.
    The initial compute level $\thetab^1$ is taken as the lowest possible compute.
    We apply a small vertical jitter to the strategy lines to distinguish overlapping providers and take a fixed profit margin of $25\%$.
}
    \label{fig:market-dynamics-gsm8k-reasoning}
\end{figure}

\begin{figure}[h!]
    \centering
    \subfloat[Potential ($\beta=2\cdot10^2$)]{
        \includegraphics[width=0.3\linewidth]{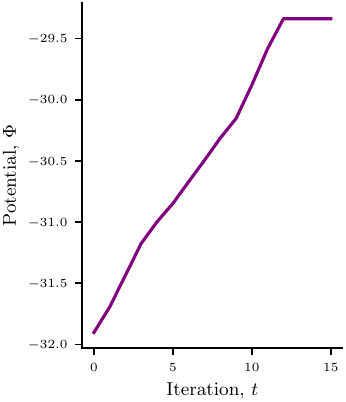}
    }
    \subfloat[Potential ($\beta=10^3$)]{
        \includegraphics[width=0.3\linewidth]{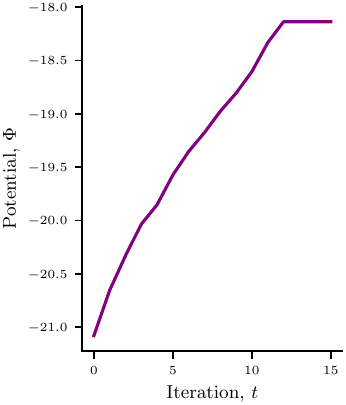}
    }
    \subfloat[Potential ($\beta=10^5$)]{
        \includegraphics[width=0.3\linewidth]{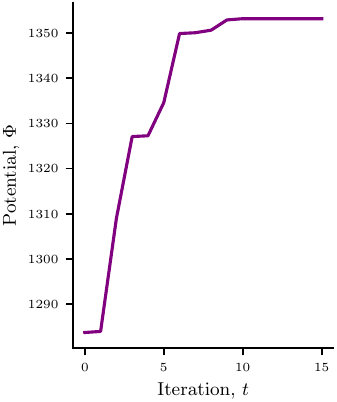}
    }
    
    \caption{{\bf Potential of a test-time compute game using CoT.}
     The figure shows, for different levels of user rationality $\beta$, the evolution of the potential $\Phi$ (see Eq.~\ref{eq:potential}) in the test-time compute games in Figure~\ref{fig:market-dynamics-gsm8k-reasoning} where $N=3$ providers sequentially select a test-time compute level that increases their utility.
     }
    \label{fig:potential-gsm8k-reasoning}
\end{figure}

\begin{figure}
    \centering
    \includegraphics[width=0.8\linewidth]{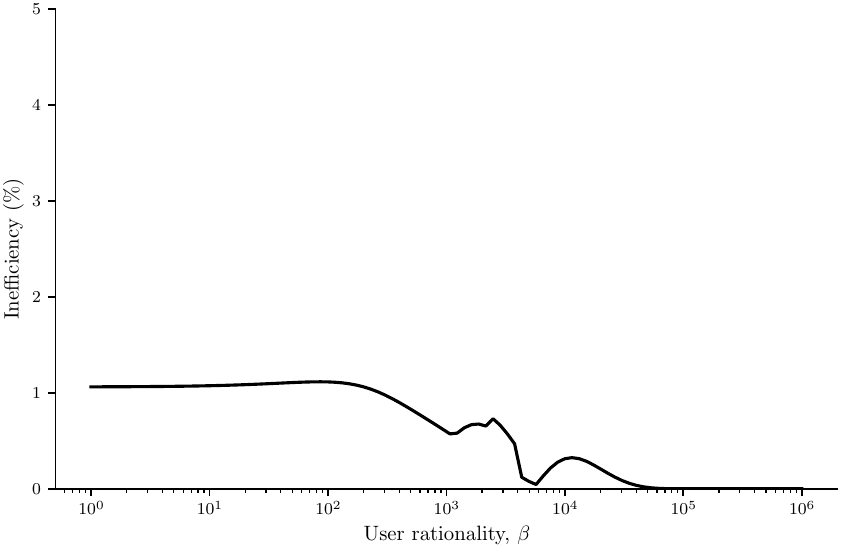}
    \caption{{\bf Inefficiency of a test-time compute game using CoT.}
     The figure shows, as a function of users' rationality $\beta$, the inefficiency ($\text{PoA}(\Gcal)-1$) of the test-time compute game in Figure~\ref{fig:market-dynamics-gsm8k-reasoning}, where $N=3$ providers sequentially select a test-time compute level that increases their utility.
     }
    \label{fig:poa-beta-gsm8k-reasoning}
\end{figure}

\clearpage
\newpage


\subsubsection{Test-Time Compute Equilibria on AIME}\label{app:dynamics-results-AIME}

\begin{figure}[h!]
    \centering
    \subfloat{
        \centering
        \includegraphics[width=0.6\linewidth]{figures/var/legend_all.pdf}
    }\\
    \setcounter{subfigure}{0}
    \centering
    \begin{subfigure}{0.32\textwidth}
        \centering
        \includegraphics[width=\linewidth]{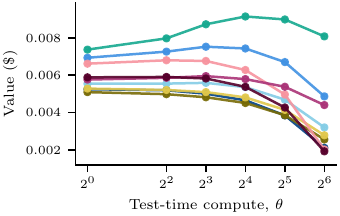}
        \caption{Majority voting}
        \label{fig:AIME:v_curves:majority}
    \end{subfigure}
    \hfill
    \begin{subfigure}{0.32\textwidth}
        \centering
        \includegraphics[width=\linewidth]{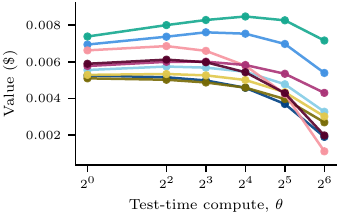}
        \caption{Best-of-n}
        \label{fig:AIME:v_curves:bon}
    \end{subfigure}
    \hfill
    \begin{subfigure}{0.32\textwidth}
        \centering
        \includegraphics[width=\linewidth]{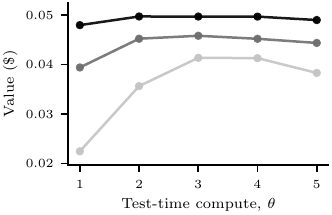}
        \caption{Chain-of-thought}
        \label{fig:AIME:v_curves:reasoning}
    \end{subfigure}
    
    \caption{{\bf User values offered by providers in a test-time compute game on AIME.}
     The figure shows the user values $V_i(\theta)$ offered by providers in a test-time compute game $\Gcal$, as a function of their test-time compute $\theta$.
     Panel~(a) and (b) correspond to games with $N=9$ providers serving non-reasoning models from the \texttt{Llama} and \texttt{Qwen} families, where providers use, respectively, majority voting and best-of-n across $\theta$ samples. Panel (c) corresponds to a game with $N=3$ providers serving reasoning models distilled from \texttt{DeepSeek-R1}, where $\theta$ represents reasoning effort, defined by binning the model outputs into quantiles based on the number of reasoning tokens (see Appendix~\ref{app:experimental-details}).
     In both games, providers serve queries $Q$ from the \texttt{AIME} dataset, we set $\beta = 1000$ and consider that each (average) point of accuracy offers a value of $\$0.05$ to the users.
    }
    \label{fig:AIME:v_curves_comparison}
\end{figure}

\begin{figure}[h!]
    \centering
    \subfloat{
        \centering
        \includegraphics[width=0.7\linewidth]{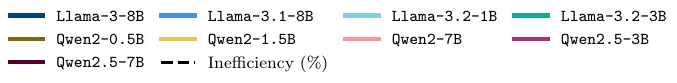}
    }

    \subfloat{
        \includegraphics[width=0.3\linewidth]{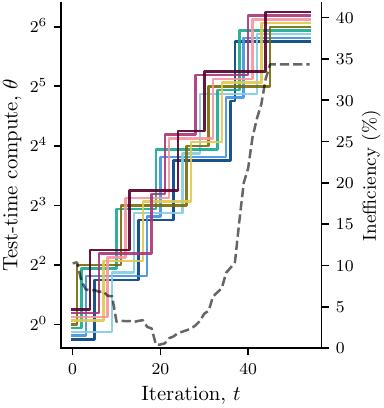}
    }
    \subfloat{
        \includegraphics[width=0.3\linewidth]{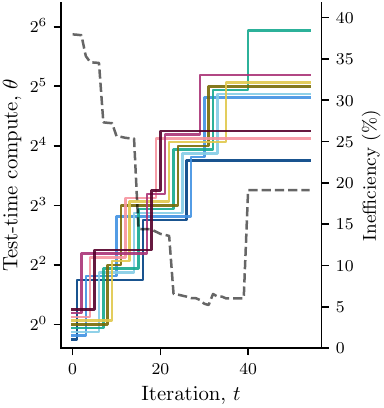}
    }
    \subfloat{
        \includegraphics[width=0.3\linewidth]{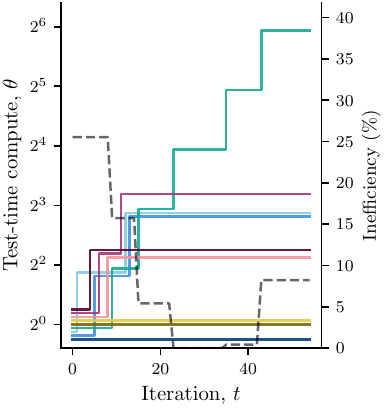}
    }\\
    \setcounter{subfigure}{0}
    \hspace{-5mm}
    \subfloat[$\beta=2\cdot10^2$]{
        \includegraphics[width=0.3\linewidth]{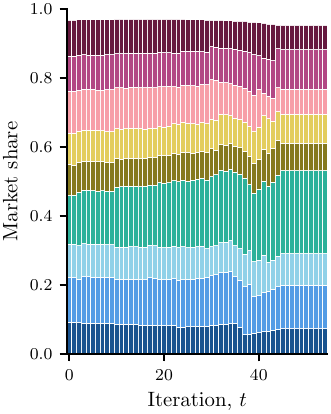}
        }
    \subfloat[$\beta=10^3$]{
        \includegraphics[width=0.3\linewidth]{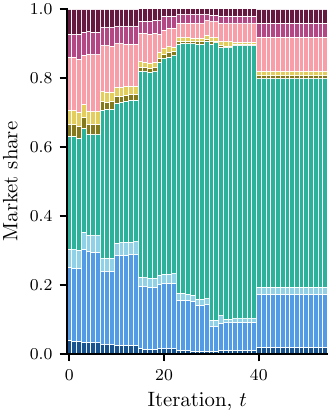}
    }
    \subfloat[$\beta=10^5$]{
        \includegraphics[width=0.3\linewidth]{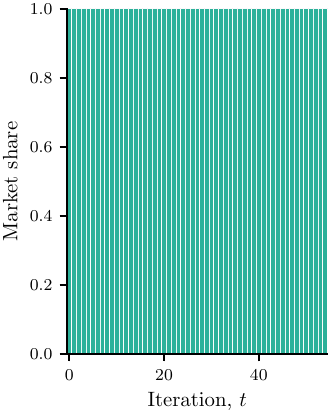}
    }
    \caption{{\bf Dynamics of a test-time compute game using majority-voting.}
     The figure shows, for different levels of user rationality $\beta$, the better-response dynamics of a test-time compute game $\Gcal$ where $N=9$ providers sequentially select a test-time compute level that increases their utility. The upper panels show the compute levels $\theta$ selected by each provider and the resulting market inefficiency $(\text{PoA}(\Gcal) - 1)$, and the lower panels show the market share of each provider.
     Here, all providers use majority-voting across $\theta$ samples as their test-time compute method to serve queries $Q$ from the \texttt{AIME} dataset. We consider that providers operate with a margin of $25\%$ between per-token price and per-token cost, and that each (average) point of accuracy offers a value of $\$0.008$ to the users.
}
    \label{fig:market-dynamics-AIME-unreasoning-maj}
\end{figure}

\begin{figure}[h!]
    \centering
    \subfloat[Potential ($\beta=2\cdot10^2$)]{
        \includegraphics[width=0.3\linewidth]{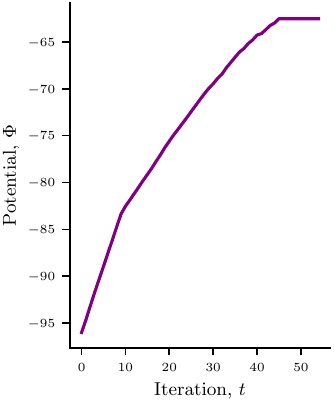}
    }
    \subfloat[Potential ($\beta=10^3$)]{
        \includegraphics[width=0.3\linewidth]{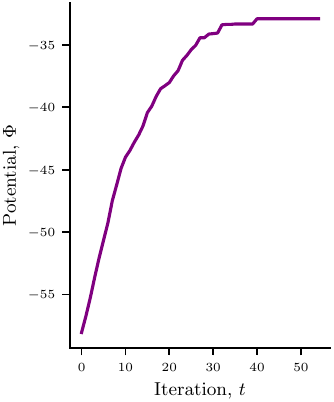}
    }
    \subfloat[Potential ($\beta=10^5$)]{
        \includegraphics[width=0.3\linewidth]{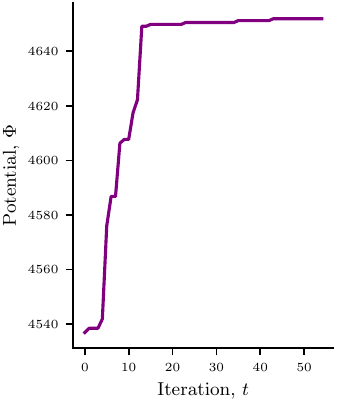}
    }
    \caption{{\bf Potential of a test-time compute game using majority-voting.}
     The figure shows, for different levels of user rationality $\beta$, the evolution of the potential $\Phi$ (see Eq.~\ref{eq:potential}) in the test-time compute games in Figure~\ref{fig:market-dynamics-AIME-unreasoning-maj} where $N=9$ providers sequentially select a test-time compute level that increases their utility.
     }
    \label{fig:potential-AIME-unreasoning-maj}
\end{figure}

\begin{figure}
    \centering
    \includegraphics[width=0.7\linewidth]{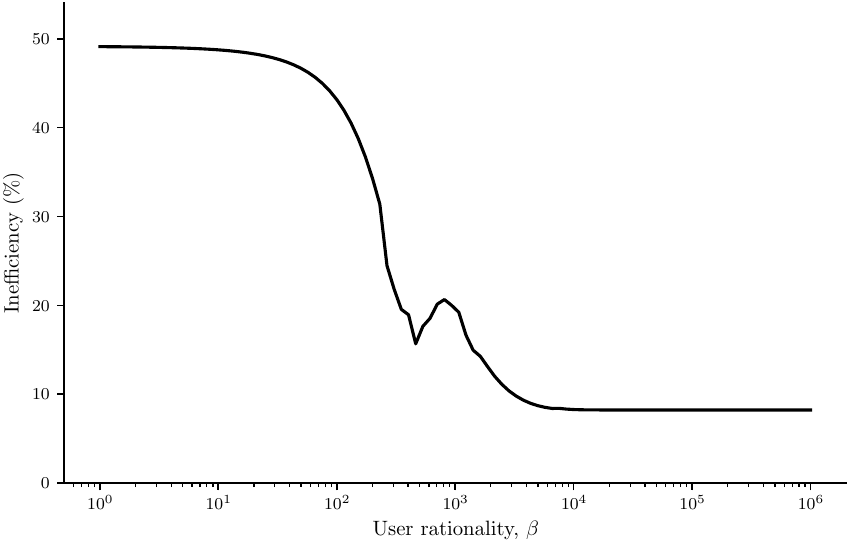}
    \caption{{\bf Inefficiency of a test-time compute game using majority-voting.}
     The figure shows, as a function of users' rationality $\beta$, the inefficiency ($\text{PoA}(\Gcal)-1$) of the test-time compute game in Figure~\ref{fig:market-dynamics-AIME-unreasoning-maj}, where $N=9$ providers sequentially select a test-time compute level that increases their utility.
     }
    \label{fig:poa-beta-AIME-unreasoning-maj}
\end{figure}

\begin{figure}[h!]
    \centering
    \subfloat{
        \centering
        \includegraphics[width=0.7\linewidth]{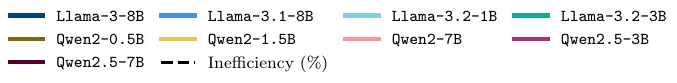}
    }

    \subfloat{
        \includegraphics[width=0.3\linewidth]{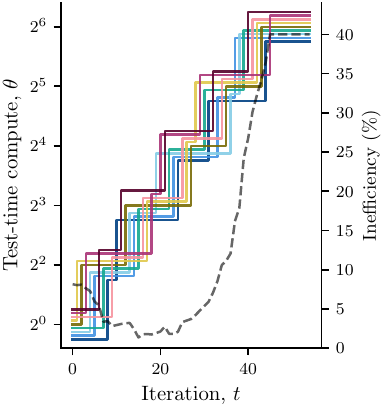}
    }
    \subfloat{
        \includegraphics[width=0.3\linewidth]{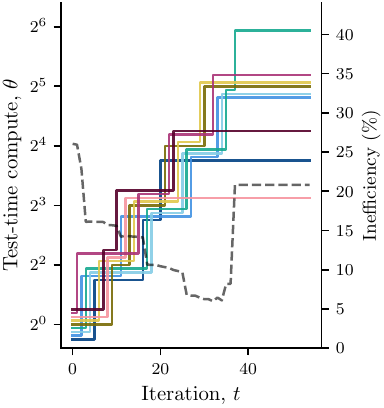}
    }
    \subfloat{
        \includegraphics[width=0.3\linewidth]{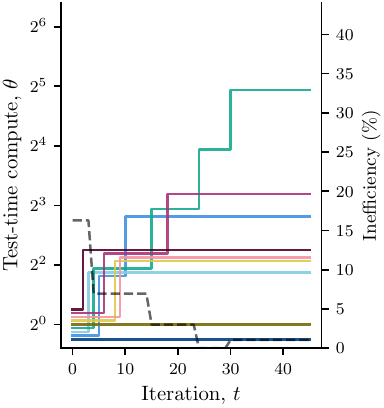}
    }\\
    \setcounter{subfigure}{0}
    \hspace{-5mm}
    \subfloat[$\beta=2\cdot10^2$]{
        \includegraphics[width=0.3\linewidth]{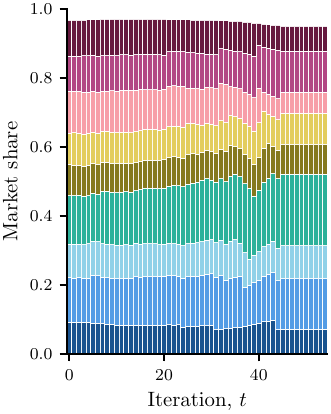}
        }
    \subfloat[$\beta=10^3$]{
        \includegraphics[width=0.3\linewidth]{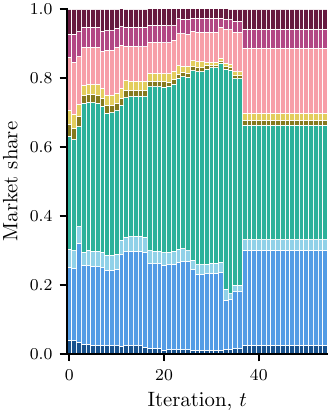}
    }
    \subfloat[$\beta=10^5$]{
        \includegraphics[width=0.3\linewidth]{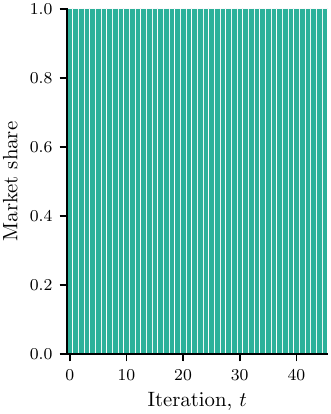}
    }
    \caption{\xhdr{Better-response dynamics of a test-time compute game using best-of-n} 
    The upper panels show, for varying levels of user rationality $\beta$, better-response dynamics when providers serving models from the \texttt{Llama} and \texttt{Qwen} families use best-of-n to serve queries from the AIME dataset.
    The solid colored lines (left $y$-axis) represent the test-time compute $\theta$ selected by each provider at each iteration, corresponding to the number of samples used for best-of-n. 
    The dashed black line (right $y$-axis) tracks the Market Inefficiency, defined as $(\text{PoA} - 1) \times 100$ (see Eq.~\ref{eq:PoA}).
    The lower panels show the evolution of the market share at each time step of the better-response dynamics.
    The initial compute level $\thetab^1$ is taken as the lowest possible compute.
    We apply a small vertical jitter to the strategy lines to distinguish overlapping providers and take a fixed profit margin of $25\%$.
}
    \label{fig:market-dynamics-AIME-unreasoning-bon}
\end{figure}

\begin{figure}[h!]
    \centering
    \subfloat[Potential ($\beta=2\cdot10^2$)]{
        \includegraphics[width=0.3\linewidth]{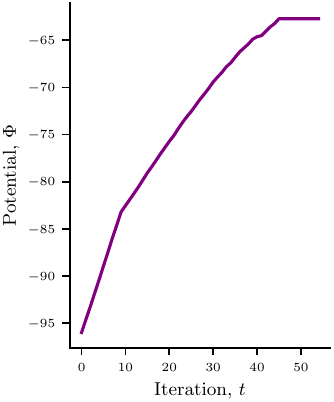}
    }
    \subfloat[Potential ($\beta=10^3$)]{
        \includegraphics[width=0.3\linewidth]{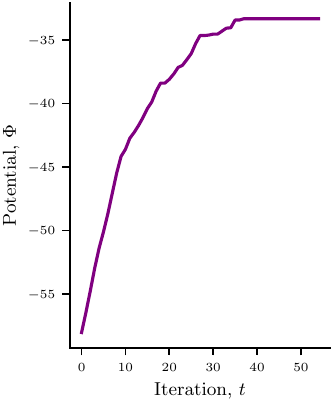}
    }
    \subfloat[Potential ($\beta=10^5$)]{
        \includegraphics[width=0.3\linewidth]{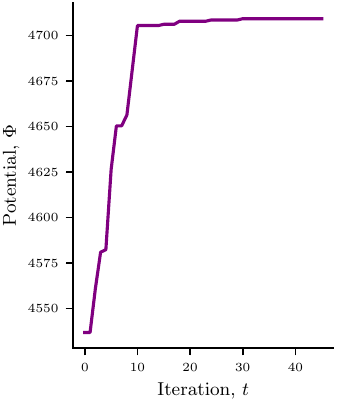}
    }
    \caption{{\bf Potential of a test-time compute game using best-of-n.}
     The figure shows, for different levels of user rationality $\beta$, the evolution of the potential $\Phi$ (see Eq.~\ref{eq:potential}) in the test-time compute games in Figure~\ref{fig:market-dynamics-AIME-unreasoning-bon} where $N=9$ providers sequentially select a test-time compute level that increases their utility.
     }
    \label{fig:potential-AIME-unreasoning-bon}
\end{figure}

\begin{figure}
    \centering
    \includegraphics[width=0.8\linewidth]{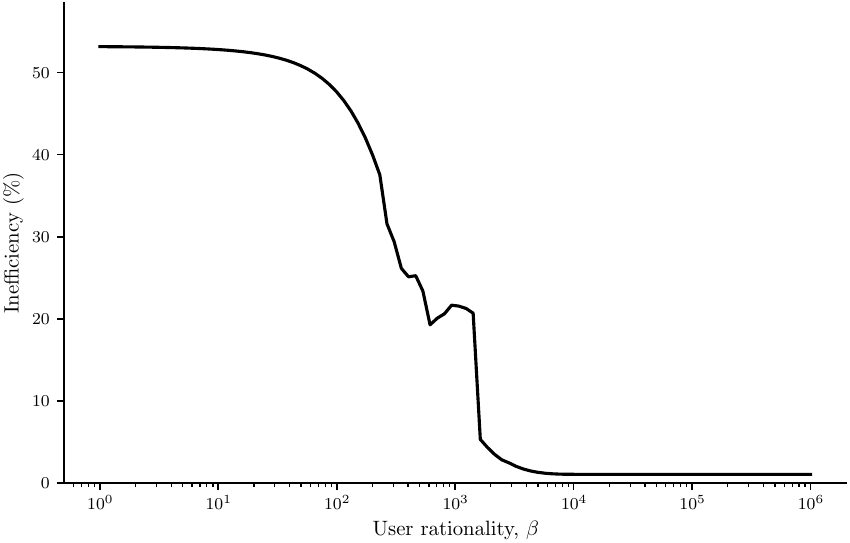}
    \caption{{\bf Inefficiency of a test-time compute game using best-of-n.}
     The figure shows, as a function of users' rationality $\beta$, the inefficiency ($\text{PoA}(\Gcal)-1$) of the test-time compute game in Figure~\ref{fig:market-dynamics-AIME-unreasoning-bon}, where $N=9$ providers sequentially select a test-time compute level that increases their utility.
     }
    \label{fig:poa-beta-AIME-unreasoning-bon}
\end{figure}

\begin{figure}[h!]
    \centering
    \subfloat{
        \centering
        \includegraphics[width=0.9\linewidth]{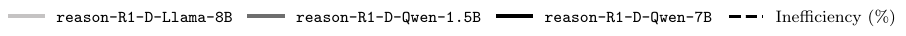}
    }

    \subfloat{
        \includegraphics[width=0.3\linewidth]{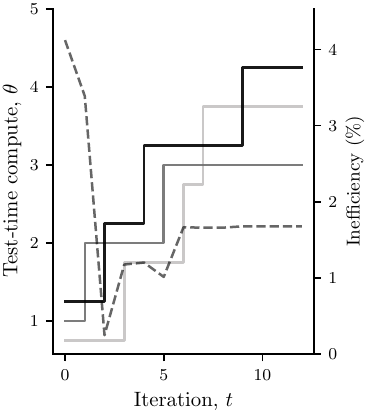}
    }
    \subfloat{
        \includegraphics[width=0.3\linewidth]{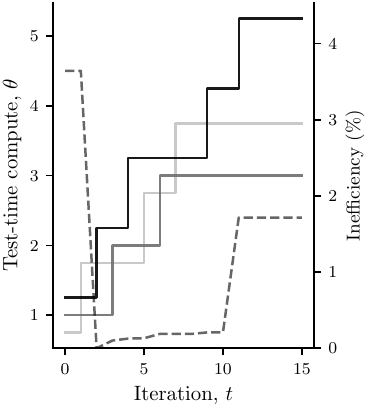}
    }
    \subfloat{
        \includegraphics[width=0.3\linewidth]{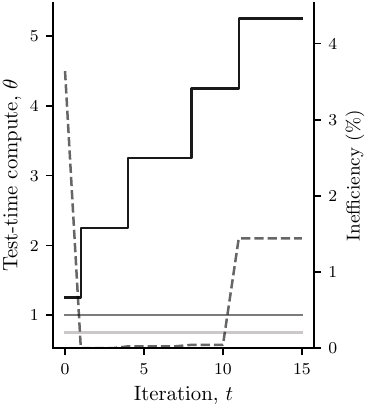}
    }\\
    \setcounter{subfigure}{0}
    \hspace{-5mm}
    \subfloat[$\beta=2\cdot10^2$]{
        \includegraphics[width=0.3\linewidth]{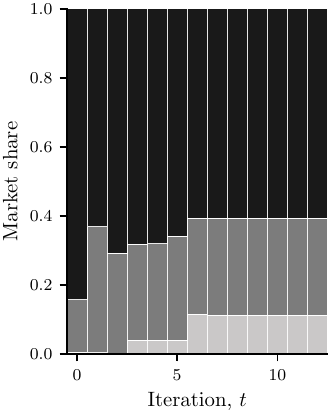}
        }
    \subfloat[$\beta=10^3$]{
        \includegraphics[width=0.3\linewidth]{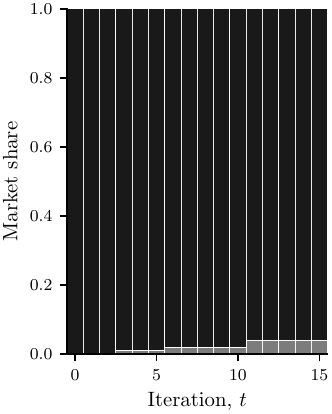}
    }
    \subfloat[$\beta=10^5$]{
        \includegraphics[width=0.3\linewidth]{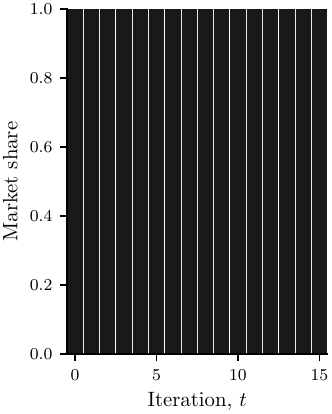}
    }
    \caption{\xhdr{Better-response dynamics of a test-time compute game using CoT} 
    The upper panels show, for varying levels of user rationality $\beta$, better-response dynamics when providers serving models from the \texttt{Llama} and \texttt{Qwen} families distilled from \texttt{DeepSeek-R1} use chain-of-thought to serve queries from the AIME dataset.
    The solid colored lines (left $y$-axis) represent the test-time compute $\theta$ selected by each provider at each iteration, corresponding to the reasoning effort used. 
    The dashed black line (right $y$-axis) tracks the Market Inefficiency, defined as $(\text{PoA} - 1) \times 100$ (see Eq.~\ref{eq:PoA}).
    The lower panels show the evolution of the market share at each time step of the better-response dynamics.
    The initial compute level $\thetab^1$ is taken as the lowest possible compute.
    We apply a small vertical jitter to the strategy lines to distinguish overlapping providers and take a fixed profit margin of $25\%$.
}
    \label{fig:market-dynamics-AIME-reasoning}
\end{figure}

\begin{figure}[h!]
    \centering
    \subfloat[Potential ($\beta=2\cdot10^2$)]{
        \includegraphics[width=0.3\linewidth]{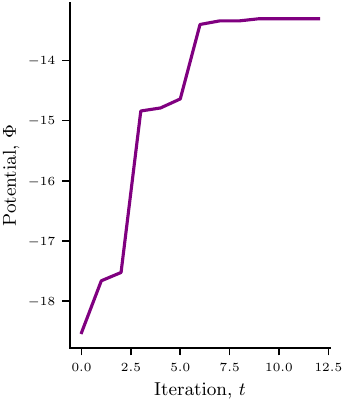}
    }
    \subfloat[Potential ($\beta=10^3$)]{
        \includegraphics[width=0.3\linewidth]{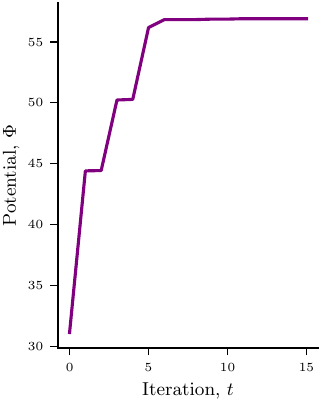}
    }
    \subfloat[Potential ($\beta=10^5$)]{
        \includegraphics[width=0.3\linewidth]{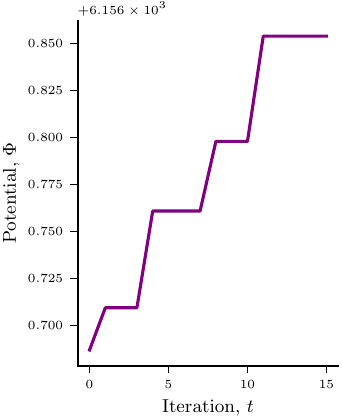}
    }
    
    \caption{{\bf Potential of a test-time compute game using CoT.}
     The figure shows, for different levels of user rationality $\beta$, the evolution of the potential $\Phi$ (see Eq.~\ref{eq:potential}) in the test-time compute games in Figure~\ref{fig:market-dynamics-AIME-reasoning} where $N=3$ providers sequentially select a test-time compute level that increases their utility.
     }
    \label{fig:potential-AIME-reasoning}
\end{figure}

\begin{figure}
    \centering
    \includegraphics[width=0.8\linewidth]{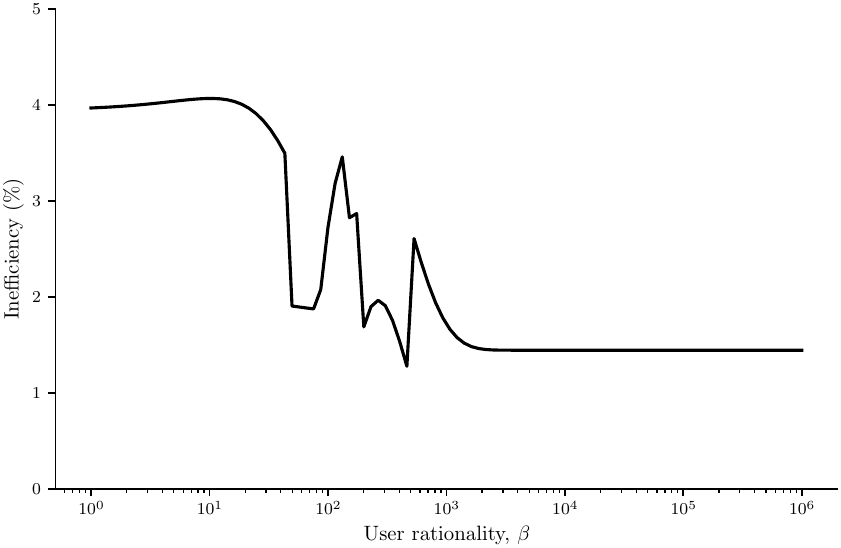}
    \caption{{\bf Inefficiency of a test-time compute game using CoT.}
     The figure shows, as a function of users' rationality $\beta$, the inefficiency ($\text{PoA}(\Gcal)-1$) of the test-time compute game in Figure~\ref{fig:market-dynamics-AIME-reasoning}, where $N=3$ providers sequentially select a test-time compute level that increases their utility.
     }
    \label{fig:poa-beta-AIME-reasoning}
\end{figure}

\clearpage
\newpage


\subsubsection{Test-Time Compute Equilibria on GPQA}\label{app:dynamics-results-GPQA}

\begin{figure}[h!]
    \centering
    \subfloat{
        \centering
        \includegraphics[width=0.6\linewidth]{figures/var/legend_all.pdf}
    }\\
    \setcounter{subfigure}{0}
    \centering
    \begin{subfigure}{0.32\textwidth}
        \centering
        \includegraphics[width=\linewidth]{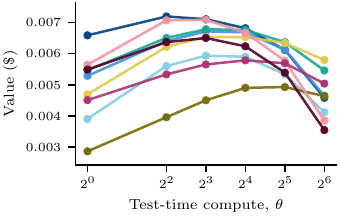}
        \caption{Majority voting}
        \label{fig:GPQA:v_curves:majority}
    \end{subfigure}
    \hfill
    \begin{subfigure}{0.32\textwidth}
        \centering
        \includegraphics[width=\linewidth]{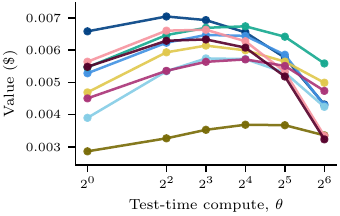}
        \caption{Best-of-n}
        \label{fig:GPQA:v_curves:bon}
    \end{subfigure}
    \hfill
    \begin{subfigure}{0.32\textwidth}
        \centering
        \includegraphics[width=\linewidth]{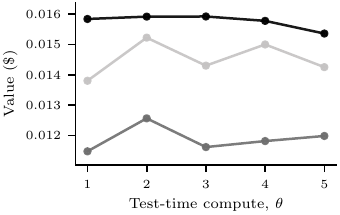}
        \caption{Chain-of-thought}
        \label{fig:GPQA:v_curves:reasoning}
    \end{subfigure}
    
    \caption{{\bf User values offered by providers in a test-time compute game on GPQA.}
     The figure shows the user values $V_i(\theta)$ offered by providers in a test-time compute game $\Gcal$, as a function of their test-time compute $\theta$.
     Panel~(a) and (b) correspond to games with $N=9$ providers serving non-reasoning models from the \texttt{Llama} and \texttt{Qwen} families, where providers use, respectively, majority voting and best-of-n across $\theta$ samples. Panel (c) corresponds to a game with $N=3$ providers serving reasoning models distilled from \texttt{DeepSeek-R1}, where $\theta$ represents reasoning effort, defined by binning the model outputs into quantiles based on the number of reasoning tokens (see Appendix~\ref{app:experimental-details}).
     In both games, providers serve queries $Q$ from the \texttt{GPQA} dataset, we set $\beta = 1000$ and consider that each (average) point of accuracy offers a value of $\$0.02$ to the users.
    }
    \label{fig:GPQA:v_curves_comparison}
\end{figure}

\begin{figure}[h!]
    \centering
    \subfloat{
        \centering
        \includegraphics[width=0.7\linewidth]{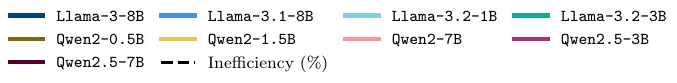}
    }

    \subfloat{
        \includegraphics[width=0.3\linewidth]{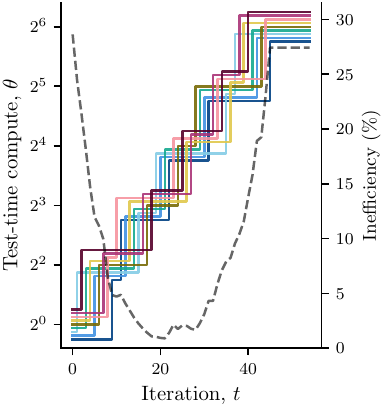}
    }
    \subfloat{
        \includegraphics[width=0.3\linewidth]{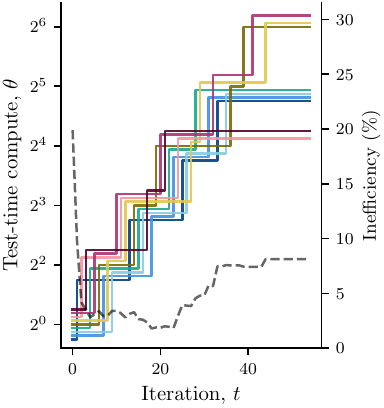}
    }
    \subfloat{
        \includegraphics[width=0.3\linewidth]{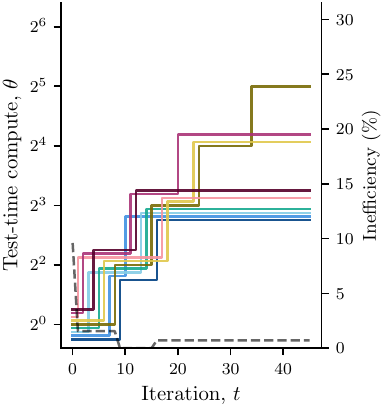}
    }\\
    \setcounter{subfigure}{0}
    \hspace{-5mm}
    \subfloat[$\beta=2\cdot10^2$]{
        \includegraphics[width=0.3\linewidth]{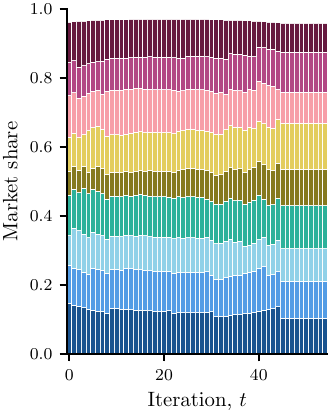}
        }
    \subfloat[$\beta=10^3$]{
        \includegraphics[width=0.3\linewidth]{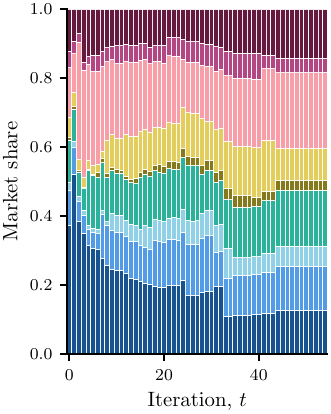}
    }
    \subfloat[$\beta=10^5$]{
        \includegraphics[width=0.3\linewidth]{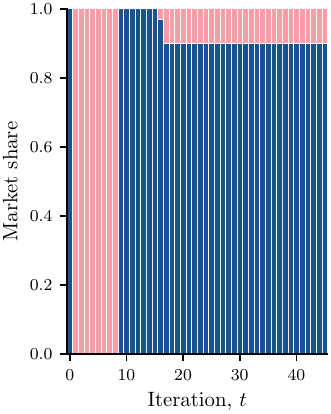}
    }
    \caption{{\bf Dynamics of a test-time compute game using majority-voting.}
     The figure shows, for different levels of user rationality $\beta$, the better-response dynamics of a test-time compute game $\Gcal$ where $N=9$ providers sequentially select a test-time compute level that increases their utility. The upper panels show the compute levels $\theta$ selected by each provider and the resulting market inefficiency $(\text{PoA}(\Gcal) - 1)$, and the lower panels show the market share of each provider.
     Here, all providers use majority-voting across $\theta$ samples as their test-time compute method to serve queries $Q$ from the \texttt{GPQA} dataset. We consider that providers operate with a margin of $25\%$ between per-token price and per-token cost, and that each (average) point of accuracy offers a value of $\$0.008$ to the users.
}
    \label{fig:market-dynamics-GPQA-unreasoning-maj}
\end{figure}

\begin{figure}[h!]
    \centering
    \subfloat[Potential ($\beta=2\cdot10^2$)]{
        \includegraphics[width=0.3\linewidth]{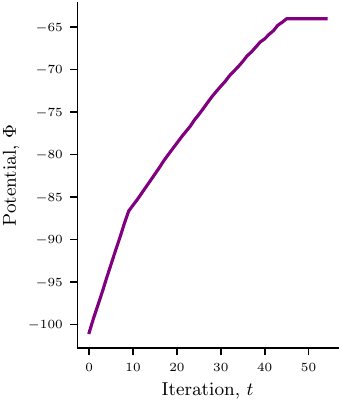}
    }
    \subfloat[Potential ($\beta=10^3$)]{
        \includegraphics[width=0.3\linewidth]{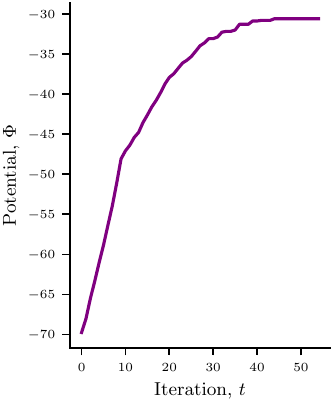}
    }
    \subfloat[Potential ($\beta=10^5$)]{
        \includegraphics[width=0.3\linewidth]{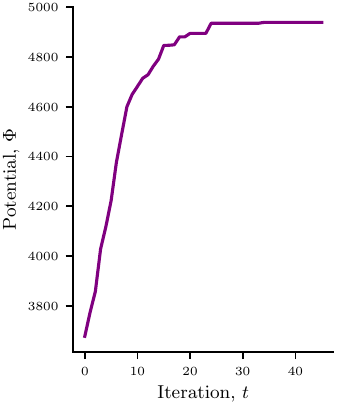}
    }
    \caption{{\bf Potential of a test-time compute game using majority-voting.}
     The figure shows, for different levels of user rationality $\beta$, the evolution of the potential $\Phi$ (see Eq.~\ref{eq:potential}) in the test-time compute games in Figure~\ref{fig:market-dynamics-GPQA-unreasoning-maj} where $N=9$ providers sequentially select a test-time compute level that increases their utility.
     }
    \label{fig:potential-GPQA-unreasoning-maj}
\end{figure}

\begin{figure}
    \centering
    \includegraphics[width=0.7\linewidth]{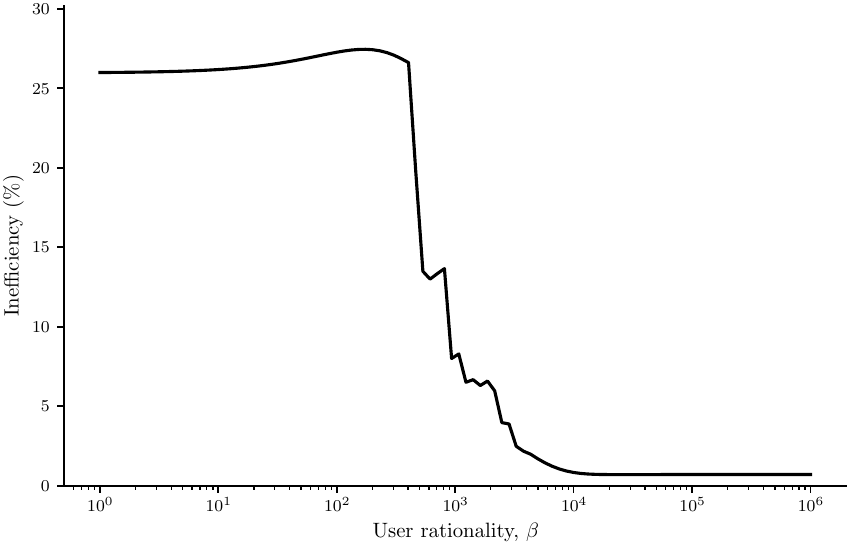}
    \caption{{\bf Inefficiency of a test-time compute game using majority-voting.}
     The figure shows, as a function of users' rationality $\beta$, the inefficiency ($\text{PoA}(\Gcal)-1$) of the test-time compute game in Figure~\ref{fig:market-dynamics-GPQA-unreasoning-maj}, where $N=9$ providers sequentially select a test-time compute level that increases their utility.
     }
    \label{fig:poa-beta-GPQA-unreasoning-maj}
\end{figure}

\begin{figure}[h!]
    \centering
    \subfloat{
        \centering
        \includegraphics[width=0.7\linewidth]{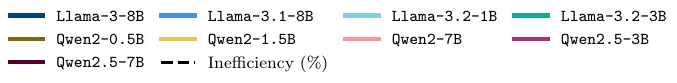}
    }

    \subfloat{
        \includegraphics[width=0.3\linewidth]{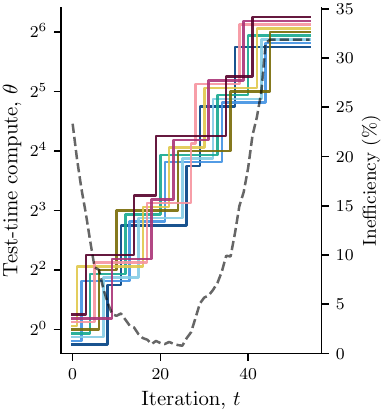}
    }
    \subfloat{
        \includegraphics[width=0.3\linewidth]{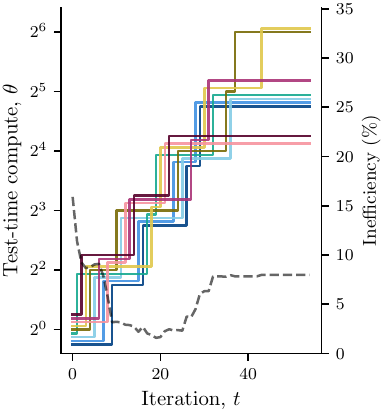}
    }
    \subfloat{
        \includegraphics[width=0.3\linewidth]{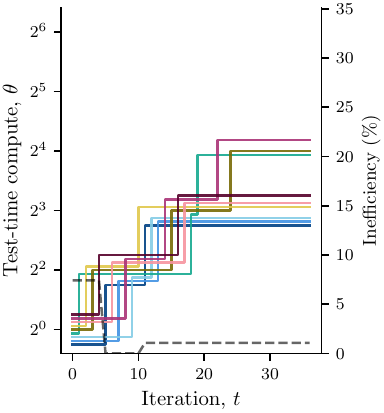}
    }\\
    \setcounter{subfigure}{0}
    \hspace{-5mm}
    \subfloat[$\beta=2\cdot10^2$]{
        \includegraphics[width=0.3\linewidth]{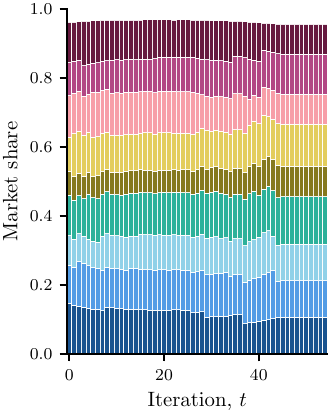}
        }
    \subfloat[$\beta=10^3$]{
        \includegraphics[width=0.3\linewidth]{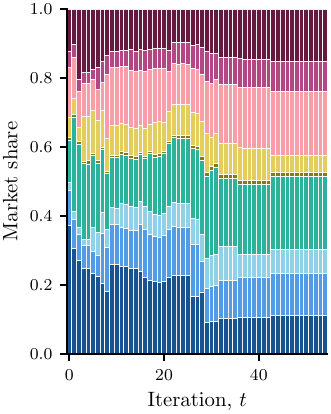}
    }
    \subfloat[$\beta=10^5$]{
        \includegraphics[width=0.3\linewidth]{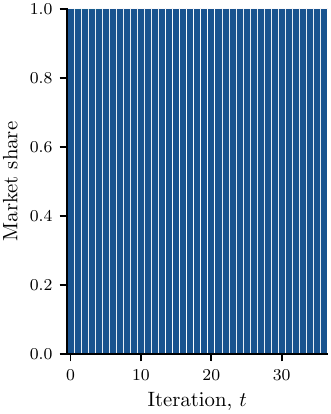}
    }
    \caption{\xhdr{Better-response dynamics of a test-time compute game using best-of-n} 
    The upper panels show, for varying levels of user rationality $\beta$, better-response dynamics when providers serving models from the \texttt{Llama} and \texttt{Qwen} families use best-of-n to serve queries from the GPQA dataset.
    The solid colored lines (left $y$-axis) represent the test-time compute $\theta$ selected by each provider at each iteration, corresponding to the number of samples used for best-of-n. 
    The dashed black line (right $y$-axis) tracks the Market Inefficiency, defined as $(\text{PoA} - 1) \times 100$ (see Eq.~\ref{eq:PoA}).
    The lower panels show the evolution of the market share at each time step of the better-response dynamics.
    The initial compute level $\thetab^1$ is taken as the lowest possible compute.
    We apply a small vertical jitter to the strategy lines to distinguish overlapping providers and take a fixed profit margin of $25\%$.
}
    \label{fig:market-dynamics-GPQA-unreasoning-bon}
\end{figure}

\begin{figure}[h!]
    \centering
    \subfloat[Potential ($\beta=2\cdot10^2$)]{
        \includegraphics[width=0.3\linewidth]{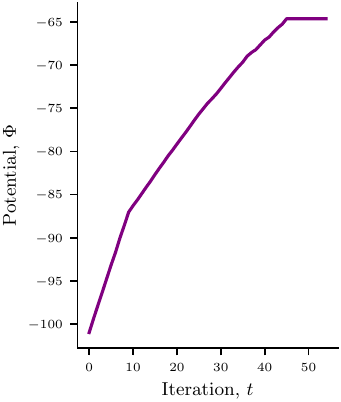}
    }
    \subfloat[Potential ($\beta=10^3$)]{
        \includegraphics[width=0.3\linewidth]{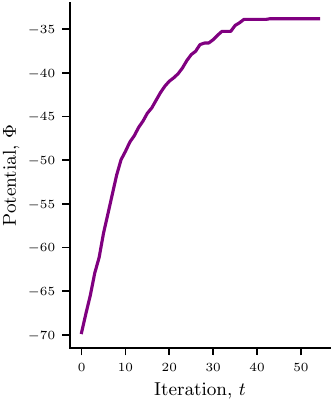}
    }
    \subfloat[Potential ($\beta=10^5$)]{
        \includegraphics[width=0.3\linewidth]{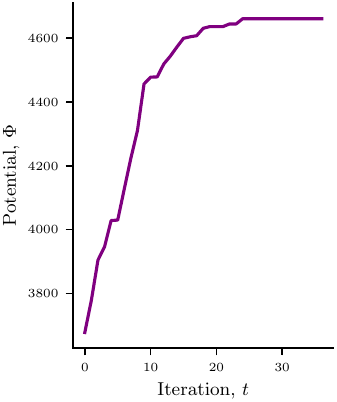}
    }
    \caption{{\bf Potential of a test-time compute game using best-of-n.}
     The figure shows, for different levels of user rationality $\beta$, the evolution of the potential $\Phi$ (see Eq.~\ref{eq:potential}) in the test-time compute games in Figure~\ref{fig:market-dynamics-GPQA-unreasoning-bon} where $N=9$ providers sequentially select a test-time compute level that increases their utility.
     }
    \label{fig:potential-GPQA-unreasoning-bon}
\end{figure}

\begin{figure}
    \centering
    \includegraphics[width=0.8\linewidth]{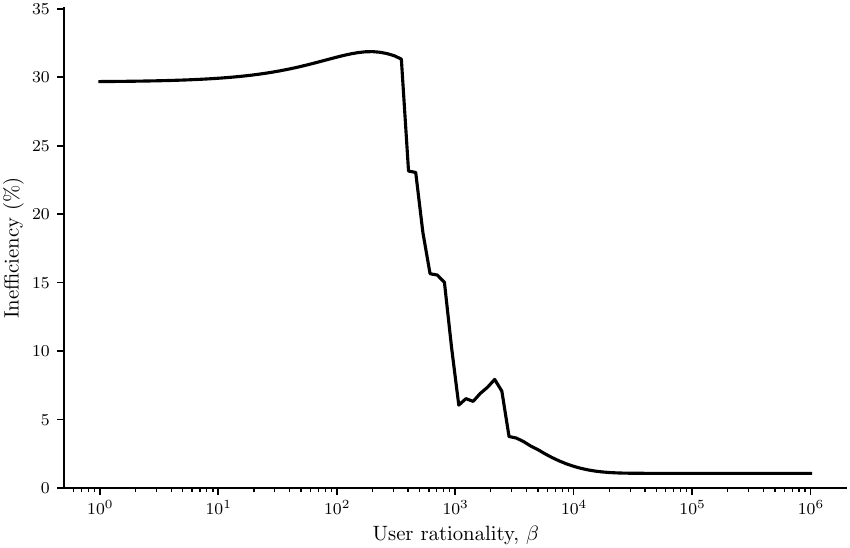}
    \caption{{\bf Inefficiency of a test-time compute game using best-of-n.}
     The figure shows, as a function of users' rationality $\beta$, the inefficiency ($\text{PoA}(\Gcal)-1$) of the test-time compute game in Figure~\ref{fig:market-dynamics-GPQA-unreasoning-bon}, where $N=9$ providers sequentially select a test-time compute level that increases their utility.
     }
    \label{fig:poa-beta-GPQA-unreasoning-bon}
\end{figure}

\begin{figure}[h!]
    \centering
    \subfloat{
        \centering
        \includegraphics[width=0.9\linewidth]{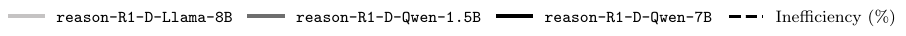}
    }

    \subfloat{
        \includegraphics[width=0.3\linewidth]{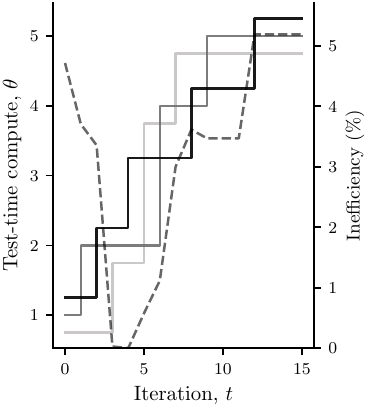}
    }
    \subfloat{
        \includegraphics[width=0.3\linewidth]{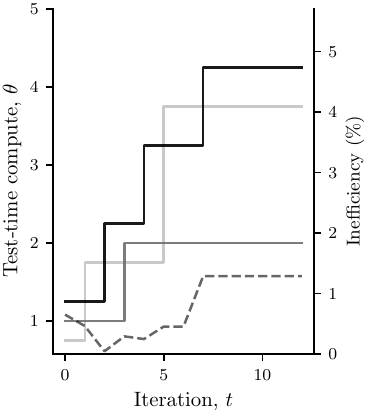}
    }
    \subfloat{
        \includegraphics[width=0.3\linewidth]{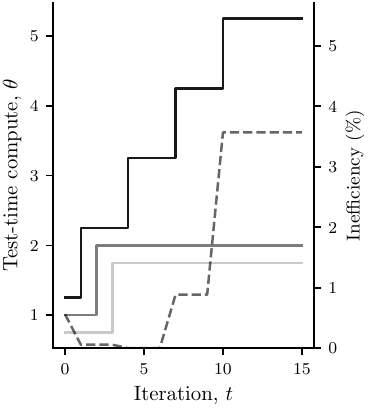}
    }\\
    \setcounter{subfigure}{0}
    \hspace{-5mm}
    \subfloat[$\beta=2\cdot10^2$]{
        \includegraphics[width=0.3\linewidth]{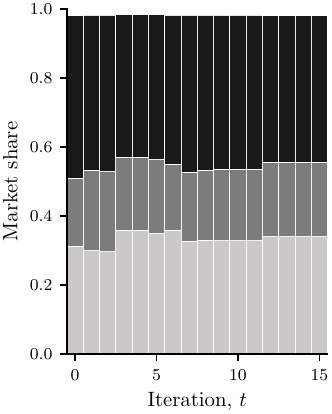}
        }
    \subfloat[$\beta=10^3$]{
        \includegraphics[width=0.3\linewidth]{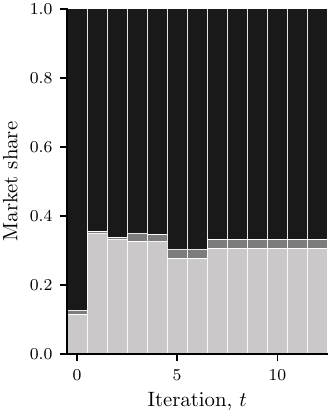}
    }
    \subfloat[$\beta=10^5$]{
        \includegraphics[width=0.3\linewidth]{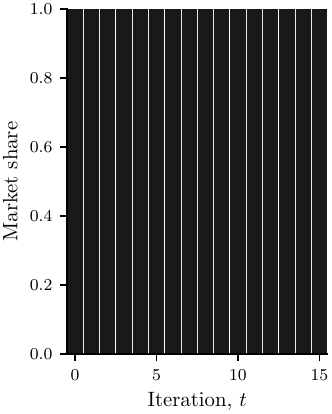}
    }
    \caption{\xhdr{Better-response dynamics of a test-time compute game using CoT} 
    The upper panels show, for varying levels of user rationality $\beta$, better-response dynamics when providers serving models from the \texttt{Llama} and \texttt{Qwen} families distilled from \texttt{DeepSeek-R1} use chain-of-thought to serve queries from the GPQA dataset.
    The solid colored lines (left $y$-axis) represent the test-time compute $\theta$ selected by each provider at each iteration, corresponding to the reasoning effort used. 
    The dashed black line (right $y$-axis) tracks the Market Inefficiency, defined as $(\text{PoA} - 1) \times 100$ (see Eq.~\ref{eq:PoA}).
    The lower panels show the evolution of the market share at each time step of the better-response dynamics.
    The initial compute level $\thetab^1$ is taken as the lowest possible compute.
    We apply a small vertical jitter to the strategy lines to distinguish overlapping providers and take a fixed profit margin of $25\%$.
}
    \label{fig:market-dynamics-GPQA-reasoning}
\end{figure}

\begin{figure}[h!]
    \centering
    \subfloat[Potential ($\beta=2\cdot10^2$)]{
        \includegraphics[width=0.3\linewidth]{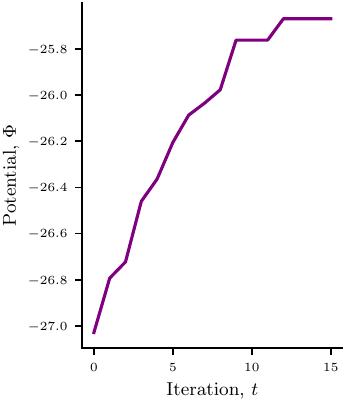}
    }
    \subfloat[Potential ($\beta=10^3$)]{
        \includegraphics[width=0.3\linewidth]{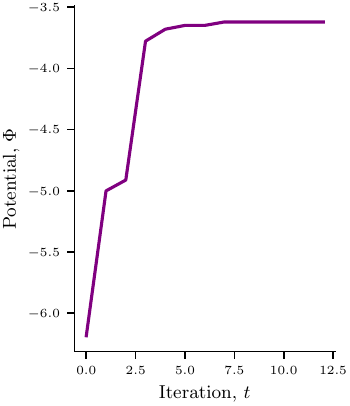}
    }
    \subfloat[Potential ($\beta=10^5$)]{
        \includegraphics[width=0.3\linewidth]{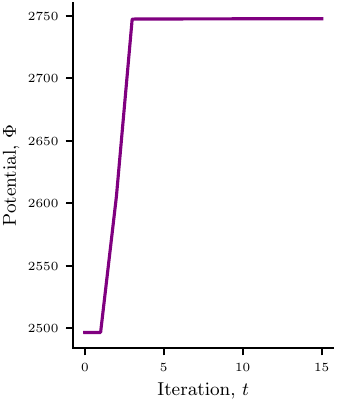}
    }
    
    \caption{{\bf Potential of a test-time compute game using CoT.}
     The figure shows, for different levels of user rationality $\beta$, the evolution of the potential $\Phi$ (see Eq.~\ref{eq:potential}) in the test-time compute games in Figure~\ref{fig:market-dynamics-GPQA-reasoning} where $N=3$ providers sequentially select a test-time compute level that increases their utility.
     }
    \label{fig:potential-GPQA-reasoning}
\end{figure}

\begin{figure}
    \centering
    \includegraphics[width=0.8\linewidth]{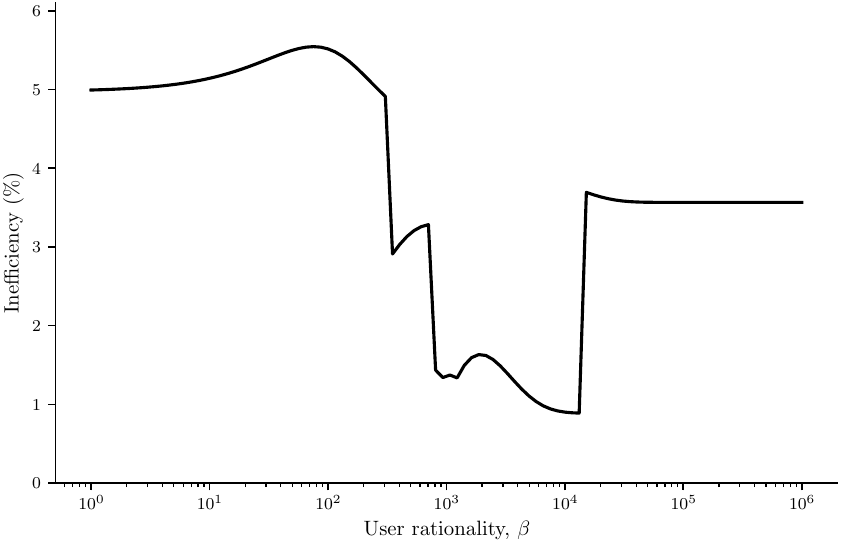}
    \caption{{\bf Inefficiency of a test-time compute game using CoT.}
     The figure shows, as a function of users' rationality $\beta$, the inefficiency ($\text{PoA}(\Gcal)-1$) of the test-time compute game in Figure~\ref{fig:market-dynamics-GPQA-reasoning}, where $N=3$ providers sequentially select a test-time compute level that increases their utility.
     }
    \label{fig:poa-beta-GPQA-reasoning}
\end{figure}

\clearpage
\newpage


\subsubsection{Test-Time Compute Equilibria on TruthfulQA}\label{app:dynamics-results-TruthfulQA}

\begin{figure}[h!]
    \centering
    \subfloat{
        \centering
        \includegraphics[width=0.6\linewidth]{figures/var/legend_all.pdf}
    }\\
    \setcounter{subfigure}{0}
    \centering
    \begin{subfigure}{0.32\textwidth}
        \centering
        \includegraphics[width=\linewidth]{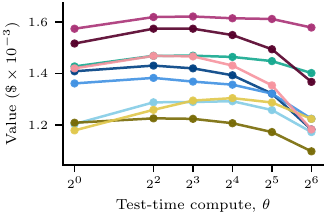}
        \caption{Best-of-n}
        \label{fig:TruthfulQA:v_curves:bon}
    \end{subfigure}
    
    \caption{{\bf User values offered by providers in a test-time compute game on TruthfulQA.}
     The figure shows the user values $V_i(\theta)$ offered by providers in a test-time compute game $\Gcal$, as a function of their test-time compute $\theta$.
     The panel corresponds to games with $N=9$ providers serving non-reasoning models from the \texttt{Llama} and \texttt{Qwen} families, where providers use best-of-n across $\theta$ samples.
     Providers serve queries $Q$ from the \texttt{TruthfulQA} dataset, we set $\beta = 1000$ and consider that each (average) point of quality (as reported in Figure~\ref{fig:truhthfullqa_all_unreasoning}) offers a value of $\$0.008$ to the users.
    }
    \label{fig:TruthfulQA:v_curves_comparison}
\end{figure}

\begin{figure}[h!]
    \centering
    \subfloat{
        \centering
        \includegraphics[width=0.7\linewidth]{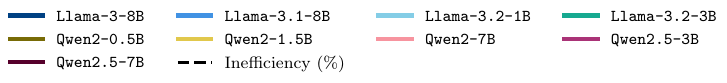}
    }

    \subfloat{
        \includegraphics[width=0.3\linewidth]{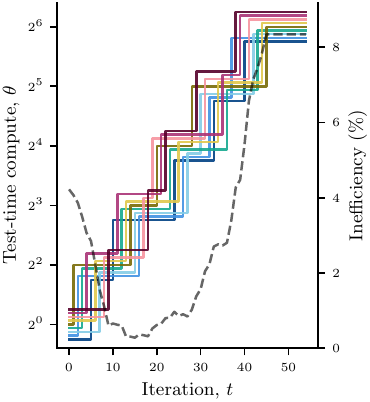}
    }
    \subfloat{
        \includegraphics[width=0.3\linewidth]{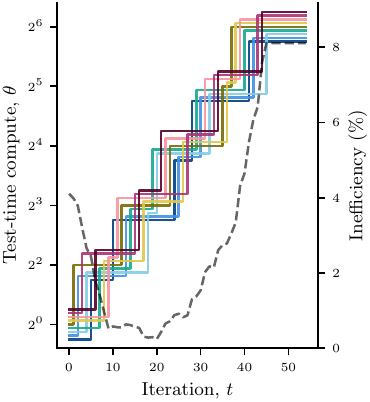}
    }
    \subfloat{
        \includegraphics[width=0.3\linewidth]{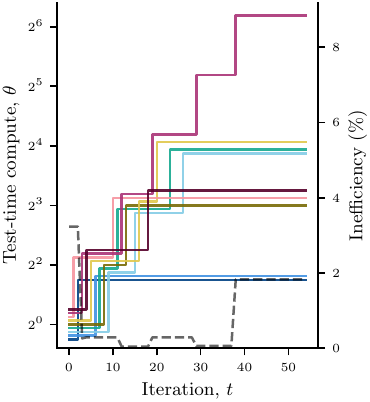}
    }\\
    \setcounter{subfigure}{0}
    \hspace{-5mm}
    \subfloat[$\beta=2\cdot10^2$]{
        \includegraphics[width=0.3\linewidth]{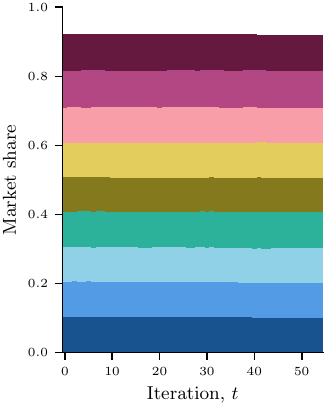}
        }
    \subfloat[$\beta=10^3$]{
        \includegraphics[width=0.3\linewidth]{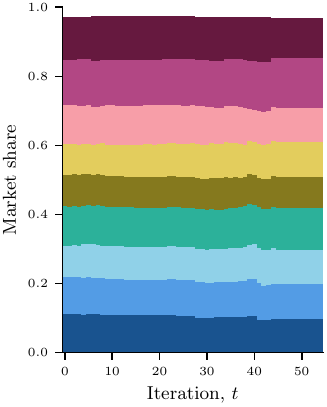}
    }
    \subfloat[$\beta=10^5$]{
        \includegraphics[width=0.3\linewidth]{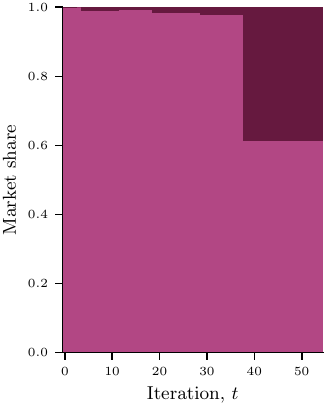}
    }
    \caption{\xhdr{Better-response dynamics of a test-time compute game using best-of-n} 
    The upper panels show, for varying levels of user rationality $\beta$, better-response dynamics when providers serving models from the \texttt{Llama} and \texttt{Qwen} families use best-of-n to serve queries from the TruthfulQA dataset.
    The solid colored lines (left $y$-axis) represent the test-time compute $\theta$ selected by each provider at each iteration, corresponding to the number of samples used for best-of-n. 
    The dashed black line (right $y$-axis) tracks the Market Inefficiency, defined as $(\text{PoA} - 1) \times 100$ (see Eq.~\ref{eq:PoA}).
    The lower panels show the evolution of the market share at each time step of the better-response dynamics.
    The initial compute level $\thetab^1$ is taken as the lowest possible compute.
    We apply a small vertical jitter to the strategy lines to distinguish overlapping providers and take a fixed profit margin of $25\%$.
}
    \label{fig:market-dynamics-TRUTHFULQA-unreasoning-bon}
\end{figure}

\begin{figure}[h!]
    \centering
    \subfloat[Potential ($\beta=2\cdot10^2$)]{
        \includegraphics[width=0.3\linewidth]{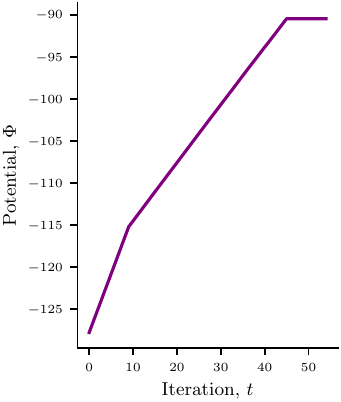}
    }
    \subfloat[Potential ($\beta=10^3$)]{
        \includegraphics[width=0.3\linewidth]{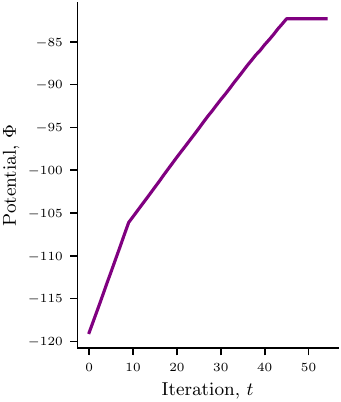}
    }
    \subfloat[Potential ($\beta=10^5$)]{
        \includegraphics[width=0.3\linewidth]{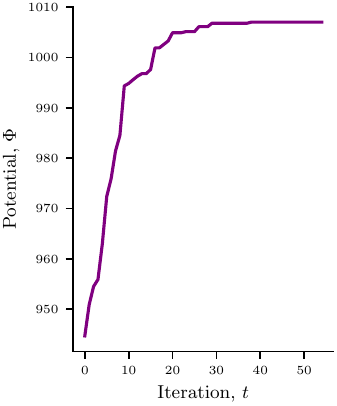}
    }
    \caption{{\bf Potential of a test-time compute game using best-of-n.}
     The figure shows, for different levels of user rationality $\beta$, the evolution of the potential $\Phi$ (see Eq.~\ref{eq:potential}) in the test-time compute games in Figure~\ref{fig:market-dynamics-TRUTHFULQA-unreasoning-bon} where $N=9$ providers sequentially select a test-time compute level that increases their utility.
     }
    \label{fig:potential-TRUTHFULQA-unreasoning-bon}
\end{figure}

\begin{figure}
    \centering
    \includegraphics[width=0.8\linewidth]{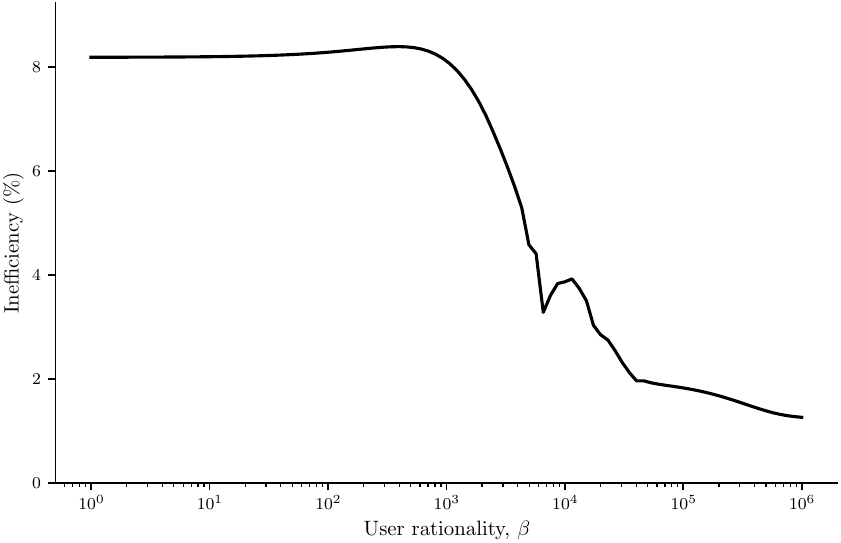}
    \caption{{\bf Inefficiency of a test-time compute game using best-of-n.}
     The figure shows, as a function of users' rationality $\beta$, the inefficiency ($\text{PoA}(\Gcal)-1$) of the test-time compute game in Figure~\ref{fig:market-dynamics-TRUTHFULQA-unreasoning-bon}, where $N=9$ providers sequentially select a test-time compute level that increases their utility.
     }
    \label{fig:poa-beta-TRUTHFULQA-unreasoning-bon}
\end{figure}

\clearpage

\subsection{Results of the auction mechanism}\label{app:auction-tables}

\subsubsection{Results of the auction mechanism on GSM8K}
\begin{table}[H]
\centering
\captionsetup{font=small}
\setlength{\tabcolsep}{4pt}
\renewcommand{\arraystretch}{0.5}
\caption{\textbf{Comparison of the equilibrium between $\Gcal$ and the auction $\widetilde{\Gcal}$ on GSM8K using majority voting.}
The table shows, for $\widetilde{\Gcal}$, the value received by the user (once the auction is conducted), the price they pay (according to Eq.~\ref{eq:second-price-payment}), the provider's utility, the social welfare, and the provider (LLM) that wins the auction. For $\Gcal$, the table shows the providers' average utility, the users' average value and price, all weighted by their equilibrium market shares, together with the social welfare. We vary the provider margin---used to derive the per-token cost from the per-token price---and the proportionality parameter $\alpha$, which linearly converts a percentage point of accuracy into a monetary value (see Appendix~\ref{app:experimental-details}). All quantities (except the compute) have units $(\$ \times 10^{-3})$.
}
\vspace{0.5cm}
\label{tab:auction-gsm8k-maj-all}

\begin{subtable}{\textwidth}
\centering
\begin{tabular}{lcccc}
\toprule
\multicolumn{1}{l}{
\smash{
\begin{tabular}[t]{@{}l@{}}
\small margin $=0.25$ \\
\small $\alpha=\$0.008$
\end{tabular}
}
}
 & \multicolumn{3}{c}{Game $\Gcal$} & Auction $\widetilde{\Gcal}$ \\
\cmidrule(lr){2-4} \cmidrule(lr){5-5}
{} & $\beta=2\cdot10^2$ & $\beta=10^3$ & $\beta=10^5$ & \texttt{Qwen2.5-7B} \\
\midrule
Compute ($\theta$)   & 32   & 16   & 4    & 2     \\
User value           & 4.5  & 5.8  & 7.1  & 7.1   \\
Price                & 1.6  & 1.2  & 0.34 & 0.30  \\
Provider(s') utility & 0.32 & 0.25 & 0.07 & 0.186 \\
Social welfare       & 4.8  & 6.0  & 7.2  & 7.3   \\
\bottomrule
\end{tabular}
\end{subtable}

\vspace{1em}

\begin{subtable}{\textwidth}
\centering
\begin{tabular}{lcccc}
\toprule
\multicolumn{1}{l}{
\smash{
\begin{tabular}[t]{@{}l@{}}
\small margin $=0.15$ \\
\small $\alpha=\$0.008$
\end{tabular}
}
}
 & \multicolumn{3}{c}{Game $\Gcal$} & Auction $\widetilde{\Gcal}$ \\
\cmidrule(lr){2-4} \cmidrule(lr){5-5}
{} & $\beta=2\cdot10^2$ & $\beta=10^3$ & $\beta=10^5$ & \texttt{Qwen2.5-7B} \\
\midrule
Compute ($\theta$)   & 32   & 16   & 4    & 2    \\
User value           & 4.57 & 5.80 & 7.11 & 7.09 \\
Price                & 1.62 & 1.24 & 0.34 & 0.31 \\
Provider(s') utility & 0.21 & 0.16 & 0.04 & 0.16 \\
Social welfare       & 4.78 & 5.96 & 7.15 & 7.25 \\
\bottomrule
\end{tabular}
\end{subtable}

\vspace{1em}

\begin{subtable}{\textwidth}
\centering
\begin{tabular}{lcccc}
\toprule
\multicolumn{1}{l}{
\smash{
\begin{tabular}[t]{@{}l@{}}
\small margin $=0.05$ \\
\small $\alpha=\$0.008$
\end{tabular}
}
}
 & \multicolumn{3}{c}{Game $\Gcal$} & Auction $\widetilde{\Gcal}$ \\
\cmidrule(lr){2-4} \cmidrule(lr){5-5}
{} & $\beta=2\cdot10^2$ & $\beta=10^3$ & $\beta=10^5$ & \texttt{Qwen2.5-7B} \\
\midrule
Compute ($\theta$)   & 32   & 16   & 4    & 2     \\
User value           & 4.57 & 5.80 & 7.11 & 7.08  \\
Price                & 1.62 & 1.24 & 0.34 & 0.32  \\
Provider(s') utility & 0.08 & 0.06 & 0.02 & 0.16  \\
Social welfare       & 4.64 & 5.86 & 7.13 & 7.23  \\
\bottomrule
\end{tabular}
\end{subtable}

\vspace{1em}

\begin{subtable}{\textwidth}
\centering
\begin{tabular}{lcccc}
\toprule
\multicolumn{1}{l}{
\smash{
\begin{tabular}[t]{@{}l@{}}
\small margin $=0.25$ \\
\small $\alpha=\$0.016$
\end{tabular}
}
}
 & \multicolumn{3}{c}{Game $\Gcal$} & Auction $\widetilde{\Gcal}$ \\
\cmidrule(lr){2-4} \cmidrule(lr){5-5}
{} & $\beta=2\cdot10^2$ & $\beta=10^3$ & $\beta=10^5$ & \texttt{Qwen2.5-7B} \\
\midrule
Compute ($\theta$)   & 32   & 16   & 4    & 2     \\
User value           & 11.41 & 13.15 & 14.56 & 14.33 \\
Price                & 1.75  & 1.40  & 0.34  & 0.46  \\
Provider(s') utility & 0.35  & 0.28  & 0.07  & 0.32  \\
Social welfare       & 11.75 & 13.43 & 14.63 & 14.66 \\
\bottomrule
\end{tabular}
\end{subtable}

\vspace{1em}

\begin{subtable}{\textwidth}
\centering
\begin{tabular}{lcccc}
\toprule
\multicolumn{1}{l}{
\smash{
\begin{tabular}[t]{@{}l@{}}
\small margin $=0.25$ \\
\small $\alpha=\$0.004$
\end{tabular}
}
}
 & \multicolumn{3}{c}{Game $\Gcal$} & Auction $\widetilde{\Gcal}$ \\
\cmidrule(lr){2-4} \cmidrule(lr){5-5}
{} & $\beta=2\cdot10^2$ & $\beta=10^3$ & $\beta=10^5$ & \texttt{Qwen2.5-7B} \\
\midrule
Compute ($\theta$)   & 32   & 16   & 2    & 1     \\
User value           & 1.40 & 2.16 & 3.53 & 3.49  \\
Price                & 1.53 & 1.15 & 0.17 & 0.12  \\
Provider(s') utility & 0.31 & 0.23 & 0.03 & 0.08  \\
Social welfare       & 1.71 & 2.39 & 3.56 & 3.57  \\
\bottomrule
\end{tabular}
\end{subtable}

\end{table}

\begin{table}[H]
\centering
\captionsetup{font=small}
\setlength{\tabcolsep}{4pt}
\renewcommand{\arraystretch}{0.5}
\caption{\textbf{Comparison of the equilibrium between $\Gcal$ and the auction $\widetilde{\Gcal}$ on GSM8K using best-of-n.}
The table shows, for $\widetilde{\Gcal}$, the value received by the user (once the auction is conducted), the price they pay (according to Eq.~\ref{eq:second-price-payment}), the provider's utility, the social welfare, and the provider (LLM) that wins the auction. For $\Gcal$, the table shows the providers' average utility, the users' average value and price, all weighted by their equilibrium market shares, together with the social welfare. We vary the provider margin---used to derive the per-token cost from the per-token price---and the proportionality parameter $\alpha$, which linearly converts a percentage point of accuracy into a monetary value (see Appendix~\ref{app:experimental-details}). All quantities (except the compute) have units $(\$ \times 10^{-3})$.
}
\vspace{0.5cm}
\label{tab:auction-gsm8k-bon}

\begin{subtable}{\textwidth}
\centering
\begin{tabular}{lcccc}
\toprule
\multicolumn{1}{l}{
\smash{
\begin{tabular}[t]{@{}l@{}}
\small margin $=0.25$ \\
\small $\alpha=\$0.008$
\end{tabular}
}
}
 & \multicolumn{3}{c}{Game $\Gcal$} & Auction $\widetilde{\Gcal}$ \\
\cmidrule(lr){2-4} \cmidrule(lr){5-5}
{} & $\beta=2\cdot10^2$ & $\beta=10^3$ & $\beta=10^5$ & \texttt{Qwen2.5-7B} \\
\midrule
Compute ($\theta$)   &  32  & 16 & 4 & 2    \\
User value           & 4.4 & 5.5 & 7.1 & 7.2 \\
Price                & 1.6 & 1.2 & 0.33 & 0.32    \\
Provider(s') utility & 0.32 & 0.25 & 0.07 & 0.18    \\
Social welfare       & 4.7 & 5.8 & 7.2 & 7.3    \\
\bottomrule
\end{tabular}
\end{subtable}

\vspace{1em}

\begin{subtable}{\textwidth}
\centering
\begin{tabular}{lcccc}
\toprule
\multicolumn{1}{l}{
\smash{
\begin{tabular}[t]{@{}l@{}}
\small margin $=0.15$ \\
\small $\alpha=\$0.008$
\end{tabular}
}
}
 & \multicolumn{3}{c}{Game $\Gcal$} & Auction $\widetilde{\Gcal}$ \\
\cmidrule(lr){2-4} \cmidrule(lr){5-5}
{} & $\beta=2\cdot10^2$ & $\beta=10^3$ & $\beta=10^5$ & \texttt{Qwen2.5-7B} \\
\midrule
Compute ($\theta$)   & 32   & 16   & 4    & 2     \\
User value           & 4.41 & 5.58 & 7.17 & 7.17  \\
Price                & 1.60 & 1.26 & 0.33 & 0.33  \\
Provider(s') utility & 0.21 & 0.16 & 0.04 & 0.18  \\
Social welfare       & 4.62 & 5.75 & 7.22 & 7.35  \\
\bottomrule
\end{tabular}
\end{subtable}

\vspace{1em}

\begin{subtable}{\textwidth}
\centering
\begin{tabular}{lcccc}
\toprule
\multicolumn{1}{l}{
\smash{
\begin{tabular}[t]{@{}l@{}}
\small margin $=0.05$ \\
\small $\alpha=\$0.008$
\end{tabular}
}
}
 & \multicolumn{3}{c}{Game $\Gcal$} & Auction $\widetilde{\Gcal}$ \\
\cmidrule(lr){2-4} \cmidrule(lr){5-5}
{} & $\beta=2\cdot10^2$ & $\beta=10^3$ & $\beta=10^5$ & \texttt{Qwen2.5-7B} \\
\midrule
Compute ($\theta$)   & 32   & 16   & 4    & 2     \\
User value           & 4.41 & 5.58 & 7.17 & 7.15  \\
Price                & 1.60 & 1.26 & 0.33 & 0.34  \\
Provider(s') utility & 0.08 & 0.06 & 0.02 & 0.18  \\
Social welfare       & 4.49 & 5.64 & 7.19 & 7.33  \\
\bottomrule
\end{tabular}
\end{subtable}

\vspace{1em}

\begin{subtable}{\textwidth}
\centering
\begin{tabular}{lcccc}
\toprule
\multicolumn{1}{l}{
\smash{
\begin{tabular}[t]{@{}l@{}}
\small margin $=0.25$ \\
\small $\alpha=\$0.016$
\end{tabular}
}
}
 & \multicolumn{3}{c}{Game $\Gcal$} & Auction $\widetilde{\Gcal}$ \\
\cmidrule(lr){2-4} \cmidrule(lr){5-5}
{} & $\beta=2\cdot10^2$ & $\beta=10^3$ & $\beta=10^5$ & \texttt{Qwen2.5-7B} \\
\midrule
Compute ($\theta$)   & 32   & 16   & 4    & 2     \\
User value           & 11.05 & 13.15 & 14.69 & 14.48 \\
Price                & 1.71  & 1.33  & 0.34  & 0.51  \\
Provider(s') utility & 0.34  & 0.27  & 0.07  & 0.38  \\
Social welfare       & 11.40 & 13.41 & 14.76 & 14.86 \\
\bottomrule
\end{tabular}
\end{subtable}

\vspace{1em}

\begin{subtable}{\textwidth}
\centering
\begin{tabular}{lcccc}
\toprule
\multicolumn{1}{l}{
\smash{
\begin{tabular}[t]{@{}l@{}}
\small margin $=0.25$ \\
\small $\alpha=\$0.004$
\end{tabular}
}
}
 & \multicolumn{3}{c}{Game $\Gcal$} & Auction $\widetilde{\Gcal}$ \\
\cmidrule(lr){2-4} \cmidrule(lr){5-5}
{} & $\beta=2\cdot10^2$ & $\beta=10^3$ & $\beta=10^5$ & \texttt{Qwen2.5-7B} \\
\midrule
Compute ($\theta$)   & 32   & 32   & 2    & 1     \\
User value           & 1.34 & 2.02 & 3.58 & 3.55  \\
Price                & 1.52 & 1.18 & 0.17 & 0.11  \\
Provider(s') utility & 0.30 & 0.24 & 0.03 & 0.08  \\
Social welfare       & 1.64 & 2.25 & 3.61 & 3.63  \\
\bottomrule
\end{tabular}
\end{subtable}

\end{table}


\begin{table}[H]
\centering
\captionsetup{font=small}
\setlength{\tabcolsep}{4pt}
\renewcommand{\arraystretch}{0.5}
\caption{\textbf{Comparison of the equilibrium between $\Gcal$ and the auction $\widetilde{\Gcal}$ on GSM8K using chain-of-thought.}
The table shows, for $\widetilde{\Gcal}$, the value received by the user (once the auction is conducted), the price they pay (according to Eq.~\ref{eq:second-price-payment}), the provider's utility, the social welfare, and the provider (LLM) that wins the auction. For $\Gcal$, the table shows the providers' average utility, the users' average value and price, all weighted by their equilibrium market shares, together with the social welfare. We vary the provider margin---used to derive the per-token cost from the per-token price---and the proportionality parameter $\alpha$, which linearly converts a percentage point of accuracy into a monetary value (see Appendix~\ref{app:experimental-details}). All quantities (except the compute) have units $(\$ \times 10^{-3})$.
}
\label{tab:auction-gsm8k-cot}
\vspace{0.5cm}
\begin{subtable}{\textwidth}
\centering
\begin{tabular}{lcccc}
\toprule
\multicolumn{1}{l}{
\smash{
\begin{tabular}[t]{@{}l@{}}
\small margin $=0.25$ \\
\small $\alpha=\$0.008$
\end{tabular}
}
}
 & \multicolumn{3}{c}{Game $\Gcal$} & Auction $\widetilde{\Gcal}$ \\
\cmidrule(lr){2-4} \cmidrule(lr){5-5}
{} & $\beta=2\cdot10^2$ & $\beta=10^3$ & $\beta=10^5$ & \texttt{R1-D-Qwen-7B} \\
\midrule
Compute ($\theta$)   &  4  & 4 & 4 & 4    \\
User value           & 6.5 & 7.1 & 7.4 & 7.3 \\
Price                & 0.16 & 0.17 & 0.23 & 0.30    \\
Provider(s') utility & 0.03 & 0.04 & 0.05 & 0.12    \\
Social welfare       & 6.5 & 7.1 & 7.4 & 7.4    \\
\bottomrule
\end{tabular}
\end{subtable}

\vspace{1em}

\begin{subtable}{\textwidth}
\centering
\begin{tabular}{lcccc}
\toprule
\multicolumn{1}{l}{
\smash{
\begin{tabular}[t]{@{}l@{}}
\small margin $=0.15$ \\
\small $\alpha=\$0.008$
\end{tabular}
}
}
 & \multicolumn{3}{c}{Game $\Gcal$} & Auction $\widetilde{\Gcal}$ \\
\cmidrule(lr){2-4} \cmidrule(lr){5-5}
{} & $\beta=2\cdot10^2$ & $\beta=10^3$ & $\beta=10^5$ & \texttt{R1-D-Qwen-7B} \\
\midrule
Compute ($\theta$)   & 4    & 4    & 4    & 4     \\
User value           & 6.51 & 7.16 & 7.40 & 7.33  \\
Price                & 0.16 & 0.18 & 0.23 & 0.30  \\
Provider(s') utility & 0.03 & 0.04 & 0.05 & 0.12  \\
Social welfare       & 6.54 & 7.19 & 7.45 & 7.45  \\
\bottomrule
\end{tabular}
\end{subtable}

\vspace{1em}

\begin{subtable}{\textwidth}
\centering
\begin{tabular}{lcccc}
\toprule
\multicolumn{1}{l}{
\smash{
\begin{tabular}[t]{@{}l@{}}
\small margin $=0.05$ \\
\small $\alpha=\$0.008$
\end{tabular}
}
}
 & \multicolumn{3}{c}{Game $\Gcal$} & Auction $\widetilde{\Gcal}$ \\
\cmidrule(lr){2-4} \cmidrule(lr){5-5}
{} & $\beta=2\cdot10^2$ & $\beta=10^3$ & $\beta=10^5$ & \texttt{R1-D-Qwen-7B} \\
\midrule
Compute ($\theta$)   & 4    & 4    & 4    & 4     \\
User value           & 6.51 & 7.16 & 7.40 & 7.33  \\
Price                & 0.16 & 0.18 & 0.23 & 0.30  \\
Provider(s') utility & 0.03 & 0.04 & 0.05 & 0.12  \\
Social welfare       & 6.54 & 7.19 & 7.45 & 7.45  \\
\bottomrule
\end{tabular}
\end{subtable}

\vspace{1em}

\begin{subtable}{\textwidth}
\centering
\begin{tabular}{lcccc}
\toprule
\multicolumn{1}{l}{
\smash{
\begin{tabular}[t]{@{}l@{}}
\small margin $=0.25$ \\
\small $\alpha=\$0.016$
\end{tabular}
}
}
 & \multicolumn{3}{c}{Game $\Gcal$} & Auction $\widetilde{\Gcal}$ \\
\cmidrule(lr){2-4} \cmidrule(lr){5-5}
{} & $\beta=2\cdot10^2$ & $\beta=10^3$ & $\beta=10^5$ & \texttt{R1-D-Qwen-7B} \\
\midrule
Compute ($\theta$)   & 4    & 4    & 4    & 4     \\
User value           & 14.06 & 14.72 & 15.03 & 14.75 \\
Price                & 0.17  & 0.20  & 0.23  & 0.52  \\
Provider(s') utility & 0.03  & 0.04  & 0.05  & 0.33  \\
Social welfare       & 14.09 & 14.76 & 15.08 & 15.08 \\
\bottomrule
\end{tabular}
\end{subtable}

\vspace{1em}

\begin{subtable}{\textwidth}
\centering
\begin{tabular}{lcccc}
\toprule
\multicolumn{1}{l}{
\smash{
\begin{tabular}[t]{@{}l@{}}
\small margin $=0.25$ \\
\small $\alpha=\$0.004$
\end{tabular}
}
}
 & \multicolumn{3}{c}{Game $\Gcal$} & Auction $\widetilde{\Gcal}$ \\
\cmidrule(lr){2-4} \cmidrule(lr){5-5}
{} & $\beta=2\cdot10^2$ & $\beta=10^3$ & $\beta=10^5$ & \texttt{R1-D-Qwen-7B} \\
\midrule
Compute ($\theta$)   & 4    & 4    & 3    & 3     \\
User value           & 2.93 & 3.42 & 3.61 & 3.62  \\
Price                & 0.15 & 0.17 & 0.17 & 0.17  \\
Provider(s') utility & 0.03 & 0.03 & 0.03 & 0.03  \\
Social welfare       & 2.96 & 3.45 & 3.65 & 3.65  \\
\bottomrule
\end{tabular}
\end{subtable}

\end{table}


\subsubsection{Results of the auction mechanism on AIME}

\begin{table}[H]
\centering
\captionsetup{font=small}
\setlength{\tabcolsep}{4pt}
\renewcommand{\arraystretch}{0.5}
\caption{\textbf{Comparison of the equilibrium between $\Gcal$ and the auction $\widetilde{\Gcal}$ on AIME using majority voting.}
The table shows, for $\widetilde{\Gcal}$, the value received by the user (once the auction is conducted), the price they pay (according to Eq.~\ref{eq:second-price-payment}), the provider's utility, the social welfare, and the provider (LLM) that wins the auction. For $\Gcal$, the table shows the providers' average utility, the users' average value and price, all weighted by their equilibrium market shares, together with the social welfare. We vary the provider margin---used to derive the per-token cost from the per-token price---and the proportionality parameter $\alpha$, which linearly converts a percentage point of accuracy into a monetary value (see Appendix~\ref{app:experimental-details}). All quantities (except the compute) have units $(\$ \times 10^{-3})$.
}
\label{tab:auction-aime-maj-all}
\vspace{0.5cm}
\begin{subtable}{\textwidth}
\centering
\begin{tabular}{lcccc}
\toprule
\multicolumn{1}{l}{
\smash{
\begin{tabular}[t]{@{}l@{}}
\small margin $=0.25$ \\
\small $\alpha=\$0.05$
\end{tabular}
}
}
 & \multicolumn{3}{c}{Game $\Gcal$} & Auction $\widetilde{\Gcal}$ \\
\cmidrule(lr){2-4} \cmidrule(lr){5-5}
{} & $\beta=2\cdot10^2$ & $\beta=10^3$ & $\beta=10^5$ & \texttt{Llama-3.2-3B} \\
\midrule
Compute ($\theta$)   & 32   & 32   & 32   & 8     \\
User value           & 4.24 & 7.25 & 8.08 & 7.63  \\
Price                & 3.09 & 2.08 & 2.40 & 2.11  \\
Provider(s') utility & 0.62 & 0.42 & 0.48 & 1.63  \\
Social welfare       & 4.85 & 7.66 & 8.56 & 9.26  \\
\bottomrule
\end{tabular}
\end{subtable}

\vspace{1em}

\begin{subtable}{\textwidth}
\centering
\begin{tabular}{lcccc}
\toprule
\multicolumn{1}{l}{
\smash{
\begin{tabular}[t]{@{}l@{}}
\small margin $=0.15$ \\
\small $\alpha=\$0.05$
\end{tabular}
}
}
 & \multicolumn{3}{c}{Game $\Gcal$} & Auction $\widetilde{\Gcal}$ \\
\cmidrule(lr){2-4} \cmidrule(lr){5-5}
{} & $\beta=2\cdot10^2$ & $\beta=10^3$ & $\beta=10^5$ & \texttt{Llama-3.2-3B} \\
\midrule
Compute ($\theta$)   & 32   & 32   & 32   & 8     \\
User value           & 4.24 & 7.25 & 8.08 & 7.59  \\
Price                & 3.09 & 2.08 & 2.40 & 2.15  \\
Provider(s') utility & 0.40 & 0.27 & 0.31 & 1.63  \\
Social welfare       & 4.64 & 7.52 & 8.39 & 9.22  \\
\bottomrule
\end{tabular}
\end{subtable}

\vspace{1em}

\begin{subtable}{\textwidth}
\centering
\begin{tabular}{lcccc}
\toprule
\multicolumn{1}{l}{
\smash{
\begin{tabular}[t]{@{}l@{}}
\small margin $=0.05$ \\
\small $\alpha=\$0.05$
\end{tabular}
}
}
 & \multicolumn{3}{c}{Game $\Gcal$} & Auction $\widetilde{\Gcal}$ \\
\cmidrule(lr){2-4} \cmidrule(lr){5-5}
{} & $\beta=2\cdot10^2$ & $\beta=10^3$ & $\beta=10^5$ & \texttt{Llama-3.2-3B} \\
\midrule
Compute ($\theta$)   & 32   & 32   & 32   & 8     \\
User value           & 4.24 & 7.25 & 8.08 & 7.55  \\
Price                & 3.09 & 2.08 & 2.40 & 2.19  \\
Provider(s') utility & 0.15 & 0.10 & 0.11 & 1.62  \\
Social welfare       & 4.38 & 7.35 & 8.19 & 9.17  \\
\bottomrule
\end{tabular}
\end{subtable}

\vspace{1em}

\begin{subtable}{\textwidth}
\centering
\begin{tabular}{lcccc}
\toprule
\multicolumn{1}{l}{
\smash{
\begin{tabular}[t]{@{}l@{}}
\small margin $=0.25$ \\
\small $\alpha=\$0.1$
\end{tabular}
}
}
 & \multicolumn{3}{c}{Game $\Gcal$} & Auction $\widetilde{\Gcal}$ \\
\cmidrule(lr){2-4} \cmidrule(lr){5-5}
{} & $\beta=2\cdot10^2$ & $\beta=10^3$ & $\beta=10^5$ & \texttt{Llama-3.2-3B} \\
\midrule
Compute ($\theta$)   & 32   & 32   & 32   & 16    \\
User value           & 13.63 & 18.36 & 18.56 & 16.06 \\
Price                & 3.15  & 2.37  & 2.40  & 4.31  \\
Provider(s') utility & 0.63  & 0.47  & 0.48  & 3.35  \\
Social welfare       & 14.26 & 18.83 & 19.04 & 19.41 \\
\bottomrule
\end{tabular}
\end{subtable}

\vspace{1em}

\begin{subtable}{\textwidth}
\centering
\begin{tabular}{lcccc}
\toprule
\multicolumn{1}{l}{
\smash{
\begin{tabular}[t]{@{}l@{}}
\small margin $=0.25$ \\
\small $\alpha=\$0.025$
\end{tabular}
}
}
 & \multicolumn{3}{c}{Game $\Gcal$} & Auction $\widetilde{\Gcal}$ \\
\cmidrule(lr){2-4} \cmidrule(lr){5-5}
{} & $\beta=2\cdot10^2$ & $\beta=10^3$ & $\beta=10^5$ & \texttt{Llama-3.2-3B} \\
\midrule
Compute ($\theta$)   & 32   & 16   & 16   & 8     \\
User value           & 0.38 & 2.98 & 3.89 & 3.61  \\
Price                & 2.93 & 1.17 & 1.20 & 1.26  \\
Provider(s') utility & 0.59 & 0.23 & 0.24 & 0.78  \\
Social welfare       & 0.96 & 3.22 & 4.13 & 4.39  \\
\bottomrule
\end{tabular}
\end{subtable}

\end{table}


\begin{table}[H]
\centering
\captionsetup{font=small}
\setlength{\tabcolsep}{4pt}
\renewcommand{\arraystretch}{0.5}
\caption{\textbf{Comparison of the equilibrium between $\Gcal$ and the auction $\widetilde{\Gcal}$ on AIME using best-of-n.}
The table shows, for $\widetilde{\Gcal}$, the value received by the user (once the auction is conducted), the price they pay (according to Eq.~\ref{eq:second-price-payment}), the provider's utility, the social welfare, and the provider (LLM) that wins the auction. For $\Gcal$, the table shows the providers' average utility, the users' average value and price, all weighted by their equilibrium market shares, together with the social welfare. We vary the provider margin---used to derive the per-token cost from the per-token price---and the proportionality parameter $\alpha$, which linearly converts a percentage point of accuracy into a monetary value (see Appendix~\ref{app:experimental-details}). All quantities (except the compute) have units $(\$ \times 10^{-3})$.
}
\label{tab:auction-aime-bon}
\vspace{0.5cm}
\begin{subtable}{\textwidth}
\centering
\begin{tabular}{lcccc}
\toprule
\multicolumn{1}{l}{
\smash{
\begin{tabular}[t]{@{}l@{}}
\small margin $=0.25$ \\
\small $\alpha=\$0.05$
\end{tabular}
}
}
 & \multicolumn{3}{c}{Game $\Gcal$} & Auction $\widetilde{\Gcal}$ \\
\cmidrule(lr){2-4} \cmidrule(lr){5-5}
{} & $\beta=2\cdot10^2$ & $\beta=10^3$ & $\beta=10^5$ & \texttt{Llama-3.2-3B} \\
\midrule
Compute ($\theta$)   & 32   & 32   & 16   & 8     \\
User value           & 3.95 & 6.55 & 8.25 & 7.72  \\
Price                & 3.09 & 1.73 & 1.20 & 1.34  \\
Provider(s') utility & 0.62 & 0.35 & 0.24 & 0.86  \\
Social welfare       & 4.57 & 6.90 & 8.49 & 8.58  \\
\bottomrule
\end{tabular}
\end{subtable}

\vspace{1em}

\begin{subtable}{\textwidth}
\centering
\begin{tabular}{lcccc}
\toprule
\multicolumn{1}{l}{
\smash{
\begin{tabular}[t]{@{}l@{}}
\small margin $=0.15$ \\
\small $\alpha=\$0.05$
\end{tabular}
}
}
 & \multicolumn{3}{c}{Game $\Gcal$} & Auction $\widetilde{\Gcal}$ \\
\cmidrule(lr){2-4} \cmidrule(lr){5-5}
{} & $\beta=2\cdot10^2$ & $\beta=10^3$ & $\beta=10^5$ & \texttt{Llama-3.2-3B} \\
\midrule
Compute ($\theta$)   & 32   & 32   & 16   & 8     \\
User value           & 3.95 & 6.55 & 8.25 & 7.67  \\
Price                & 3.09 & 1.73 & 1.20 & 1.40  \\
Provider(s') utility & 0.40 & 0.23 & 0.16 & 0.87  \\
Social welfare       & 4.36 & 6.78 & 8.41 & 8.54  \\
\bottomrule
\end{tabular}
\end{subtable}

\vspace{1em}

\begin{subtable}{\textwidth}
\centering
\begin{tabular}{lcccc}
\toprule
\multicolumn{1}{l}{
\smash{
\begin{tabular}[t]{@{}l@{}}
\small margin $=0.05$ \\
\small $\alpha=\$0.05$
\end{tabular}
}
}
 & \multicolumn{3}{c}{Game $\Gcal$} & Auction $\widetilde{\Gcal}$ \\
\cmidrule(lr){2-4} \cmidrule(lr){5-5}
{} & $\beta=2\cdot10^2$ & $\beta=10^3$ & $\beta=10^5$ & \texttt{Llama-3.2-3B} \\
\midrule
Compute ($\theta$)   & 32   & 32   & 16   & 8     \\
User value           & 3.95 & 6.55 & 8.25 & 7.62  \\
Price                & 3.09 & 1.73 & 1.20 & 1.44  \\
Provider(s') utility & 0.15 & 0.08 & 0.06 & 0.87  \\
Social welfare       & 4.10 & 6.64 & 8.31 & 8.49  \\
\bottomrule
\end{tabular}
\end{subtable}

\vspace{1em}

\begin{subtable}{\textwidth}
\centering
\begin{tabular}{lcccc}
\toprule
\multicolumn{1}{l}{
\smash{
\begin{tabular}[t]{@{}l@{}}
\small margin $=0.25$ \\
\small $\alpha=\$0.1$
\end{tabular}
}
}
 & \multicolumn{3}{c}{Game $\Gcal$} & Auction $\widetilde{\Gcal}$ \\
\cmidrule(lr){2-4} \cmidrule(lr){5-5}
{} & $\beta=2\cdot10^2$ & $\beta=10^3$ & $\beta=10^5$ & \texttt{Llama-3.2-3B} \\
\midrule
Compute ($\theta$)   & 32   & 32   & 32   & 16    \\
User value           & 12.56 & 16.28 & 16.71 & 16.32 \\
Price                & 3.21  & 2.19  & 2.40  & 2.58  \\
Provider(s') utility & 0.64  & 0.44  & 0.48  & 1.62  \\
Social welfare       & 13.20 & 16.72 & 17.19 & 17.94 \\
\bottomrule
\end{tabular}
\end{subtable}

\vspace{1em}

\begin{subtable}{\textwidth}
\centering
\begin{tabular}{lcccc}
\toprule
\multicolumn{1}{l}{
\smash{
\begin{tabular}[t]{@{}l@{}}
\small margin $=0.25$ \\
\small $\alpha=\$0.025$
\end{tabular}
}
}
 & \multicolumn{3}{c}{Game $\Gcal$} & Auction $\widetilde{\Gcal}$ \\
\cmidrule(lr){2-4} \cmidrule(lr){5-5}
{} & $\beta=2\cdot10^2$ & $\beta=10^3$ & $\beta=10^5$ & \texttt{Llama-3.2-3B} \\
\midrule
Compute ($\theta$)   & 32   & 16   & 8    & 8     \\
User value           & 0.30 & 2.56 & 3.93 & 3.65  \\
Price                & 2.93 & 1.32 & 0.60 & 0.88  \\
Provider(s') utility & 0.59 & 0.26 & 0.12 & 0.40  \\
Social welfare       & 0.88 & 2.82 & 4.05 & 4.05  \\
\bottomrule
\end{tabular}
\end{subtable}

\end{table}
\begin{table}[H]
\centering
\captionsetup{font=small}
\setlength{\tabcolsep}{4pt}
\renewcommand{\arraystretch}{0.5}
\caption{\textbf{Comparison of the equilibrium between $\Gcal$ and the auction $\widetilde{\Gcal}$ on AIME using chain-of-thought.}
The table shows, for $\widetilde{\Gcal}$, the value received by the user (once the auction is conducted), the price they pay (according to Eq.~\ref{eq:second-price-payment}), the provider's utility, the social welfare, and the provider (LLM) that wins the auction. For $\Gcal$, the table shows the providers' average utility, the users' average value and price, all weighted by their equilibrium market shares, together with the social welfare. We vary the provider margin---used to derive the per-token cost from the per-token price---and the proportionality parameter $\alpha$, which linearly converts a percentage point of accuracy into a monetary value (see Appendix~\ref{app:experimental-details}). All quantities (except the compute) have units $(\$ \times 10^{-3})$.
}
\label{tab:auction-aime-cot}
\vspace{0.5cm}
\begin{subtable}{\textwidth}
\centering
\begin{tabular}{lcccc}
\toprule
\multicolumn{1}{l}{
\smash{
\begin{tabular}[t]{@{}l@{}}
\small margin $=0.25$ \\
\small $\alpha=\$0.05$
\end{tabular}
}
}
 & \multicolumn{3}{c}{Game $\Gcal$} & Auction $\widetilde{\Gcal}$ \\
\cmidrule(lr){2-4} \cmidrule(lr){5-5}
{} & $\beta=2\cdot10^2$ & $\beta=10^3$ & $\beta=10^5$ & \texttt{R1-D-Qwen-7B} \\
\midrule
Compute ($\theta$)   & 3    & 4    & 4    & 1     \\
User value           & 47.64 & 48.86 & 48.99 & 45.84 \\
Price                & 0.26  & 0.33  & 0.34  & 4.16  \\
Provider(s') utility & 0.05  & 0.07  & 0.07  & 3.92  \\
Social welfare       & 47.70 & 48.92 & 49.06 & 49.76 \\
\bottomrule
\end{tabular}
\end{subtable}

\vspace{1em}

\begin{subtable}{\textwidth}
\centering
\begin{tabular}{lcccc}
\toprule
\multicolumn{1}{l}{
\smash{
\begin{tabular}[t]{@{}l@{}}
\small margin $=0.15$ \\
\small $\alpha=\$0.05$
\end{tabular}
}
}
 & \multicolumn{3}{c}{Game $\Gcal$} & Auction $\widetilde{\Gcal}$ \\
\cmidrule(lr){2-4} \cmidrule(lr){5-5}
{} & $\beta=2\cdot10^2$ & $\beta=10^3$ & $\beta=10^5$ & \texttt{R1-D-Qwen-7B} \\
\midrule
Compute ($\theta$)   & 3    & 4    & 4    & 1     \\
User value           & 47.64 & 48.86 & 48.99 & 45.84 \\
Price                & 0.26  & 0.33  & 0.34  & 4.16  \\
Provider(s') utility & 0.05  & 0.07  & 0.07  & 3.92  \\
Social welfare       & 47.70 & 48.92 & 49.06 & 49.76 \\
\bottomrule
\end{tabular}
\end{subtable}

\vspace{1em}

\begin{subtable}{\textwidth}
\centering
\begin{tabular}{lcccc}
\toprule
\multicolumn{1}{l}{
\smash{
\begin{tabular}[t]{@{}l@{}}
\small margin $=0.05$ \\
\small $\alpha=\$0.05$
\end{tabular}
}
}
 & \multicolumn{3}{c}{Game $\Gcal$} & Auction $\widetilde{\Gcal}$ \\
\cmidrule(lr){2-4} \cmidrule(lr){5-5}
{} & $\beta=2\cdot10^2$ & $\beta=10^3$ & $\beta=10^5$ & \texttt{R1-D-Qwen-7B} \\
\midrule
Compute ($\theta$)   & 3    & 4    & 4    & 1     \\
User value           & 47.64 & 48.86 & 48.99 & 45.84 \\
Price                & 0.26  & 0.33  & 0.34  & 4.16  \\
Provider(s') utility & 0.05  & 0.07  & 0.07  & 3.92  \\
Social welfare       & 47.70 & 48.92 & 49.06 & 49.76 \\
\bottomrule
\end{tabular}
\end{subtable}

\vspace{1em}

\begin{subtable}{\textwidth}
\centering
\begin{tabular}{lcccc}
\toprule
\multicolumn{1}{l}{
\smash{
\begin{tabular}[t]{@{}l@{}}
\small margin $=0.25$ \\
\small $\alpha=\$0.1$
\end{tabular}
}
}
 & \multicolumn{3}{c}{Game $\Gcal$} & Auction $\widetilde{\Gcal}$ \\
\cmidrule(lr){2-4} \cmidrule(lr){5-5}
{} & $\beta=2\cdot10^2$ & $\beta=10^3$ & $\beta=10^5$ & \texttt{R1-D-Qwen-7B} \\
\midrule
Compute ($\theta$)   & 3    & 4    & 4    & 1     \\
User value           & 97.90 & 98.31 & 98.32 & 91.82 \\
Price                & 0.29  & 0.34  & 0.34  & 8.18  \\
Provider(s') utility & 0.06  & 0.07  & 0.07  & 7.95  \\
Social welfare       & 97.96 & 98.38 & 98.38 & 99.76 \\
\bottomrule
\end{tabular}
\end{subtable}

\vspace{1em}

\begin{subtable}{\textwidth}
\centering
\begin{tabular}{lcccc}
\toprule
\multicolumn{1}{l}{
\smash{
\begin{tabular}[t]{@{}l@{}}
\small margin $=0.25$ \\
\small $\alpha=\$0.025$
\end{tabular}
}
}
 & \multicolumn{3}{c}{Game $\Gcal$} & Auction $\widetilde{\Gcal}$ \\
\cmidrule(lr){2-4} \cmidrule(lr){5-5}
{} & $\beta=2\cdot10^2$ & $\beta=10^3$ & $\beta=10^5$ & \texttt{R1-D-Qwen-7B} \\
\midrule
Compute ($\theta$)   & 4    & 3    & 4    & 1     \\
User value           & 22.72 & 24.37 & 24.32 & 22.86 \\
Price                & 0.26  & 0.30  & 0.34  & 2.14  \\
Provider(s') utility & 0.05  & 0.06  & 0.07  & 1.91  \\
Social welfare       & 22.77 & 24.43 & 24.39 & 24.76 \\
\bottomrule
\end{tabular}
\end{subtable}

\end{table}


\subsubsection{Results of the auction mechanism on GPQA}

\begin{table}[H]
\centering
\captionsetup{font=small}
\setlength{\tabcolsep}{4pt}
\renewcommand{\arraystretch}{0.5}
\caption{\textbf{Comparison of the equilibrium between $\Gcal$ and the auction $\widetilde{\Gcal}$ on GPQA using majority voting.}
The table shows, for $\widetilde{\Gcal}$, the value received by the user (once the auction is conducted), the price they pay (according to Eq.~\ref{eq:second-price-payment}), the provider's utility, the social welfare, and the provider (LLM) that wins the auction. For $\Gcal$, the table shows the providers' average utility, the users' average value and price, all weighted by their equilibrium market shares, together with the social welfare. We vary the provider margin---used to derive the per-token cost from the per-token price---and the proportionality parameter $\alpha$, which linearly converts a percentage point of accuracy into a monetary value (see Appendix~\ref{app:experimental-details}). All quantities (except the compute) have units $(\$ \times 10^{-3})$.
}
\label{tab:auction-GPQA-maj}
\vspace{0.5cm}
\begin{subtable}{\textwidth}
\centering
\begin{tabular}{lcccc}
\toprule
\multicolumn{1}{l}{
\smash{
\begin{tabular}[t]{@{}l@{}}
\small margin $=0.25$ \\
\small $\alpha=\$0.02$
\end{tabular}
}
}
 & \multicolumn{3}{c}{Game $\Gcal$} & Auction $\widetilde{\Gcal}$ \\
\cmidrule(lr){2-4} \cmidrule(lr){5-5}
{} & $\beta=2\cdot10^2$ & $\beta=10^3$ & $\beta=10^5$ & \texttt{Llama-3-8B} \\
\midrule
Compute ($\theta$)   & 32   & 8    & 4    & 2     \\
User value           & 4.53 & 6.13 & 7.10 & 7.18  \\
Price                & 2.33 & 1.21 & 0.38 & 0.20  \\
Provider(s') utility & 0.47 & 0.24 & 0.08 & 0.05  \\
Social welfare       & 4.99 & 6.37 & 7.18 & 7.23  \\
\bottomrule
\end{tabular}
\end{subtable}

\vspace{1em}

\begin{subtable}{\textwidth}
\centering
\begin{tabular}{lcccc}
\toprule
\multicolumn{1}{l}{
\smash{
\begin{tabular}[t]{@{}l@{}}
\small margin $=0.15$ \\
\small $\alpha=\$0.02$
\end{tabular}
}
}
 & \multicolumn{3}{c}{Game $\Gcal$} & Auction $\widetilde{\Gcal}$ \\
\cmidrule(lr){2-4} \cmidrule(lr){5-5}
{} & $\beta=2\cdot10^2$ & $\beta=10^3$ & $\beta=10^5$ & \texttt{Llama-3-8B} \\
\midrule
Compute ($\theta$)   & 32   & 8    & 4    & 2     \\
User value           & 4.53 & 6.13 & 7.10 & 7.15  \\
Price                & 2.33 & 1.21 & 0.38 & 0.23  \\
Provider(s') utility & 0.30 & 0.16 & 0.05 & 0.07  \\
Social welfare       & 4.83 & 6.29 & 7.15 & 7.22  \\
\bottomrule
\end{tabular}
\end{subtable}

\vspace{1em}

\begin{subtable}{\textwidth}
\centering
\begin{tabular}{lcccc}
\toprule
\multicolumn{1}{l}{
\smash{
\begin{tabular}[t]{@{}l@{}}
\small margin $=0.05$ \\
\small $\alpha=\$0.02$
\end{tabular}
}
}
 & \multicolumn{3}{c}{Game $\Gcal$} & Auction $\widetilde{\Gcal}$ \\
\cmidrule(lr){2-4} \cmidrule(lr){5-5}
{} & $\beta=2\cdot10^2$ & $\beta=10^3$ & $\beta=10^5$ & \texttt{Llama-3-8B} \\
\midrule
Compute ($\theta$)   & 32   & 8    & 4    & 2     \\
User value           & 4.53 & 6.13 & 7.10 & 7.11  \\
Price                & 2.33 & 1.21 & 0.38 & 0.27  \\
Provider(s') utility & 0.11 & 0.06 & 0.02 & 0.10  \\
Social welfare       & 4.64 & 6.19 & 7.12 & 7.20  \\
\bottomrule
\end{tabular}
\end{subtable}

\vspace{1em}

\begin{subtable}{\textwidth}
\centering
\begin{tabular}{lcccc}
\toprule
\multicolumn{1}{l}{
\smash{
\begin{tabular}[t]{@{}l@{}}
\small margin $=0.25$ \\
\small $\alpha=\$0.04$
\end{tabular}
}
}
 & \multicolumn{3}{c}{Game $\Gcal$} & Auction $\widetilde{\Gcal}$ \\
\cmidrule(lr){2-4} \cmidrule(lr){5-5}
{} & $\beta=2\cdot10^2$ & $\beta=10^3$ & $\beta=10^5$ & \texttt{Qwen2-7B} \\
\midrule
Compute ($\theta$)   & 32   & 16   & 4    & 4     \\
User value           & 11.82 & 13.47 & 14.66 & 14.65 \\
Price                & 2.51  & 1.54  & 0.49  & 0.50  \\
Provider(s') utility & 0.50  & 0.31  & 0.10  & 0.10  \\
Social welfare       & 12.32 & 13.78 & 14.75 & 14.75 \\
\bottomrule
\end{tabular}
\end{subtable}

\vspace{1em}

\begin{subtable}{\textwidth}
\centering
\begin{tabular}{lcccc}
\toprule
\multicolumn{1}{l}{
\smash{
\begin{tabular}[t]{@{}l@{}}
\small margin $=0.25$ \\
\small $\alpha=\$0.01$
\end{tabular}
}
}
 & \multicolumn{3}{c}{Game $\Gcal$} & Auction $\widetilde{\Gcal}$ \\
\cmidrule(lr){2-4} \cmidrule(lr){5-5}
{} & $\beta=2\cdot10^2$ & $\beta=10^3$ & $\beta=10^5$ & \texttt{Llama-3-8B} \\
\midrule
Compute ($\theta$)   & 32   & 8    & 2    & 2     \\
User value           & 1.08 & 2.43 & 3.50 & 3.46  \\
Price                & 2.18 & 1.13 & 0.19 & 0.23  \\
Provider(s') utility & 0.44 & 0.23 & 0.04 & 0.08  \\
Social welfare       & 1.52 & 2.65 & 3.54 & 3.54  \\
\bottomrule
\end{tabular}
\end{subtable}

\end{table}


\begin{table}[H]
\centering
\captionsetup{font=small}
\setlength{\tabcolsep}{4pt}
\renewcommand{\arraystretch}{0.5}
\caption{\textbf{Comparison of the equilibrium between $\Gcal$ and the auction $\widetilde{\Gcal}$ on GPQA using best-of-n.}
The table shows, for $\widetilde{\Gcal}$, the value received by the user (once the auction is conducted), the price they pay (according to Eq.~\ref{eq:second-price-payment}), the provider's utility, the social welfare, and the provider (LLM) that wins the auction. For $\Gcal$, the table shows the providers' average utility, the users' average value and price, all weighted by their equilibrium market shares, together with the social welfare. We vary the provider margin---used to derive the per-token cost from the per-token price---and the proportionality parameter $\alpha$, which linearly converts a percentage point of accuracy into a monetary value (see Appendix~\ref{app:experimental-details}). All quantities (except the compute) have units $(\$ \times 10^{-3})$.
}
\label{tab:auction-GPQA-bon}
\vspace{0.5cm}
\begin{subtable}{\textwidth}
\centering
\begin{tabular}{lcccc}
\toprule
\multicolumn{1}{l}{
\smash{
\begin{tabular}[t]{@{}l@{}}
\small margin $=0.25$ \\
\small $\alpha=\$0.02$
\end{tabular}
}
}
 & \multicolumn{3}{c}{Game $\Gcal$} & Auction $\widetilde{\Gcal}$ \\
\cmidrule(lr){2-4} \cmidrule(lr){5-5}
{} & $\beta=2\cdot10^2$ & $\beta=10^3$ & $\beta=10^5$ & \texttt{Llama-3-8B} \\
\midrule
Compute ($\theta$)   & 32   & 16   & 4    & 2     \\
User value           & 4.15 & 5.94 & 6.93 & 6.84  \\
Price                & 2.35 & 1.15 & 0.37 & 0.39  \\
Provider(s') utility & 0.47 & 0.23 & 0.07 & 0.24  \\
Social welfare       & 4.62 & 6.17 & 7.01 & 7.08  \\
\bottomrule
\end{tabular}
\end{subtable}

\vspace{1em}

\begin{subtable}{\textwidth}
\centering
\begin{tabular}{lcccc}
\toprule
\multicolumn{1}{l}{
\smash{
\begin{tabular}[t]{@{}l@{}}
\small margin $=0.15$ \\
\small $\alpha=\$0.02$
\end{tabular}
}
}
 & \multicolumn{3}{c}{Game $\Gcal$} & Auction $\widetilde{\Gcal}$ \\
\cmidrule(lr){2-4} \cmidrule(lr){5-5}
{} & $\beta=2\cdot10^2$ & $\beta=10^3$ & $\beta=10^5$ & \texttt{Llama-3-8B} \\
\midrule
Compute ($\theta$)   & 32   & 16   & 4    & 2     \\
User value           & 4.15 & 5.94 & 6.93 & 6.81  \\
Price                & 2.35 & 1.15 & 0.37 & 0.43  \\
Provider(s') utility & 0.31 & 0.15 & 0.05 & 0.26  \\
Social welfare       & 4.46 & 6.09 & 6.98 & 7.07  \\
\bottomrule
\end{tabular}
\end{subtable}

\vspace{1em}

\begin{subtable}{\textwidth}
\centering
\begin{tabular}{lcccc}
\toprule
\multicolumn{1}{l}{
\smash{
\begin{tabular}[t]{@{}l@{}}
\small margin $=0.05$ \\
\small $\alpha=\$0.02$
\end{tabular}
}
}
 & \multicolumn{3}{c}{Game $\Gcal$} & Auction $\widetilde{\Gcal}$ \\
\cmidrule(lr){2-4} \cmidrule(lr){5-5}
{} & $\beta=2\cdot10^2$ & $\beta=10^3$ & $\beta=10^5$ & \texttt{Llama-3-8B} \\
\midrule
Compute ($\theta$)   & 32   & 16   & 4    & 2     \\
User value           & 4.15 & 5.94 & 6.93 & 6.77  \\
Price                & 2.35 & 1.15 & 0.37 & 0.47  \\
Provider(s') utility & 0.11 & 0.05 & 0.02 & 0.29  \\
Social welfare       & 4.27 & 5.99 & 6.95 & 7.06  \\
\bottomrule
\end{tabular}
\end{subtable}

\vspace{1em}

\begin{subtable}{\textwidth}
\centering
\begin{tabular}{lcccc}
\toprule
\multicolumn{1}{l}{
\smash{
\begin{tabular}[t]{@{}l@{}}
\small margin $=0.25$ \\
\small $\alpha=\$0.04$
\end{tabular}
}
}
 & \multicolumn{3}{c}{Game $\Gcal$} & Auction $\widetilde{\Gcal}$ \\
\cmidrule(lr){2-4} \cmidrule(lr){5-5}
{} & $\beta=2\cdot10^2$ & $\beta=10^3$ & $\beta=10^5$ & \texttt{Llama-3-8B} \\
\midrule
Compute ($\theta$)   & 32   & 8    & 4    & 2     \\
User value           & 11.20 & 13.07 & 14.24 & 14.09 \\
Price                & 2.54  & 1.50  & 0.37  & 0.38  \\
Provider(s') utility & 0.51  & 0.30  & 0.07  & 0.23  \\
Social welfare       & 11.71 & 13.37 & 14.31 & 14.32 \\
\bottomrule
\end{tabular}
\end{subtable}

\vspace{1em}

\begin{subtable}{\textwidth}
\centering
\begin{tabular}{lcccc}
\toprule
\multicolumn{1}{l}{
\smash{
\begin{tabular}[t]{@{}l@{}}
\small margin $=0.25$ \\
\small $\alpha=\$0.01$
\end{tabular}
}
}
 & \multicolumn{3}{c}{Game $\Gcal$} & Auction $\widetilde{\Gcal}$ \\
\cmidrule(lr){2-4} \cmidrule(lr){5-5}
{} & $\beta=2\cdot10^2$ & $\beta=10^3$ & $\beta=10^5$ & \texttt{Llama-3-8B} \\
\midrule
Compute ($\theta$)   & 32   & 16   & 4    & 2     \\
User value           & 0.89 & 2.31 & 3.28 & 3.27  \\
Price                & 2.19 & 1.12 & 0.37 & 0.34  \\
Provider(s') utility & 0.44 & 0.22 & 0.07 & 0.20  \\
Social welfare       & 1.32 & 2.53 & 3.36 & 3.47  \\
\bottomrule
\end{tabular}
\end{subtable}

\end{table}


\begin{table}[H]
\centering
\captionsetup{font=small}
\setlength{\tabcolsep}{4pt}
\renewcommand{\arraystretch}{0.5}
\caption{\textbf{Comparison of the equilibrium between $\Gcal$ and the auction $\widetilde{\Gcal}$ on GPQA using chain-of-thought.}
The table shows, for $\widetilde{\Gcal}$, the value received by the user (once the auction is conducted), the price they pay (according to Eq.~\ref{eq:second-price-payment}), the provider's utility, the social welfare, and the provider (LLM) that wins the auction. For $\Gcal$, the table shows the providers' average utility, the users' average value and price, all weighted by their equilibrium market shares, together with the social welfare. We vary the provider margin---used to derive the per-token cost from the per-token price---and the proportionality parameter $\alpha$, which linearly converts a percentage point of accuracy into a monetary value (see Appendix~\ref{app:experimental-details}). All quantities (except the compute) have units $(\$ \times 10^{-3})$.
}
\label{tab:auction-GPQA-cot}
\vspace{0.5cm}
\begin{subtable}{\textwidth}
\centering
\begin{tabular}{lcccc}
\toprule
\multicolumn{1}{l}{
\smash{
\begin{tabular}[t]{@{}l@{}}
\small margin $=0.25$ \\
\small $\alpha=\$0.02$
\end{tabular}
}
}
 & \multicolumn{3}{c}{Game $\Gcal$} & Auction $\widetilde{\Gcal}$ \\
\cmidrule(lr){2-4} \cmidrule(lr){5-5}
{} & $\beta=2\cdot10^2$ & $\beta=10^3$ & $\beta=10^5$ & \texttt{R1-D-Qwen-7B} \\
\midrule
Compute ($\theta$)   & 4    & 3    & 4    & 2     \\
User value           & 13.95 & 15.45 & 15.35 & 15.25 \\
Price                & 0.24  & 0.24  & 0.31  & 0.91  \\
Provider(s') utility & 0.05  & 0.05  & 0.06  & 0.72  \\
Social welfare       & 14.00 & 15.50 & 15.42 & 15.97 \\
\bottomrule
\end{tabular}
\end{subtable}

\vspace{1em}

\begin{subtable}{\textwidth}
\centering
\begin{tabular}{lcccc}
\toprule
\multicolumn{1}{l}{
\smash{
\begin{tabular}[t]{@{}l@{}}
\small margin $=0.15$ \\
\small $\alpha=\$0.02$
\end{tabular}
}
}
 & \multicolumn{3}{c}{Game $\Gcal$} & Auction $\widetilde{\Gcal}$ \\
\cmidrule(lr){2-4} \cmidrule(lr){5-5}
{} & $\beta=2\cdot10^2$ & $\beta=10^3$ & $\beta=10^5$ & \texttt{R1-D-Qwen-7B} \\
\midrule
Compute ($\theta$)   & 4    & 3    & 4    & 2     \\
User value           & 13.95 & 15.45 & 15.35 & 15.25 \\
Price                & 0.24  & 0.24  & 0.31  & 0.91  \\
Provider(s') utility & 0.05  & 0.05  & 0.06  & 0.72  \\
Social welfare       & 14.00 & 15.50 & 15.42 & 15.97 \\
\bottomrule
\end{tabular}
\end{subtable}

\vspace{1em}

\begin{subtable}{\textwidth}
\centering
\begin{tabular}{lcccc}
\toprule
\multicolumn{1}{l}{
\smash{
\begin{tabular}[t]{@{}l@{}}
\small margin $=0.05$ \\
\small $\alpha=\$0.02$
\end{tabular}
}
}
 & \multicolumn{3}{c}{Game $\Gcal$} & Auction $\widetilde{\Gcal}$ \\
\cmidrule(lr){2-4} \cmidrule(lr){5-5}
{} & $\beta=2\cdot10^2$ & $\beta=10^3$ & $\beta=10^5$ & \texttt{R1-D-Qwen-7B} \\
\midrule
Compute ($\theta$)   & 4    & 3    & 4    & 2     \\
User value           & 13.95 & 15.45 & 15.35 & 15.25 \\
Price                & 0.24  & 0.24  & 0.31  & 0.91  \\
Provider(s') utility & 0.05  & 0.05  & 0.06  & 0.72  \\
Social welfare       & 14.00 & 15.50 & 15.42 & 15.97 \\
\bottomrule
\end{tabular}
\end{subtable}

\vspace{1em}

\begin{subtable}{\textwidth}
\centering
\begin{tabular}{lcccc}
\toprule
\multicolumn{1}{l}{
\smash{
\begin{tabular}[t]{@{}l@{}}
\small margin $=0.25$ \\
\small $\alpha=\$0.04$
\end{tabular}
}
}
 & \multicolumn{3}{c}{Game $\Gcal$} & Auction $\widetilde{\Gcal}$ \\
\cmidrule(lr){2-4} \cmidrule(lr){5-5}
{} & $\beta=2\cdot10^2$ & $\beta=10^3$ & $\beta=10^5$ & \texttt{R1-D-Qwen-7B} \\
\midrule
Compute ($\theta$)   & 4    & 3    & 4    & 2     \\
User value           & 29.82 & 31.53 & 31.02 & 30.62 \\
Price                & 0.24  & 0.24  & 0.31  & 1.70  \\
Provider(s') utility & 0.05  & 0.05  & 0.06  & 1.51  \\
Social welfare       & 29.87 & 31.57 & 31.08 & 32.13 \\
\bottomrule
\end{tabular}
\end{subtable}

\vspace{1em}

\begin{subtable}{\textwidth}
\centering
\begin{tabular}{lcccc}
\toprule
\multicolumn{1}{l}{
\smash{
\begin{tabular}[t]{@{}l@{}}
\small margin $=0.25$ \\
\small $\alpha=\$0.01$
\end{tabular}
}
}
 & \multicolumn{3}{c}{Game $\Gcal$} & Auction $\widetilde{\Gcal}$ \\
\cmidrule(lr){2-4} \cmidrule(lr){5-5}
{} & $\beta=2\cdot10^2$ & $\beta=10^3$ & $\beta=10^5$ & \texttt{R1-D-Qwen-7B} \\
\midrule
Compute ($\theta$)   & 4    & 4    & 3    & 1     \\
User value           & 6.37 & 7.32 & 7.75 & 7.56  \\
Price                & 0.22 & 0.24 & 0.27 & 0.50  \\
Provider(s') utility & 0.04 & 0.05 & 0.05 & 0.32  \\
Social welfare       & 6.41 & 7.37 & 7.81 & 7.89  \\
\bottomrule
\end{tabular}
\end{subtable}

\end{table}


\clearpage
\newpage

\subsubsection{Results of the auction mechanism on TruthfulQA}

\begin{table}[H]
\centering
\captionsetup{font=small}
\setlength{\tabcolsep}{4pt}
\renewcommand{\arraystretch}{0.5}
\caption{\textbf{Comparison of the equilibrium between $\Gcal$ and the auction $\widetilde{\Gcal}$ on TruthfulQA using best-of-n.}
The table shows, for $\widetilde{\Gcal}$, the value received by the user (once the auction is conducted), the price they pay (according to Eq.~\ref{eq:second-price-payment}), the provider's utility, the social welfare, and the provider (LLM) that wins the auction. For $\Gcal$, the table shows the providers' average utility, the users' average value and price, all weighted by their equilibrium market shares, together with the social welfare. We vary the provider margin---used to derive the per-token cost from the per-token price---and the proportionality parameter $\alpha$, which linearly converts a percentage point of accuracy into a monetary value (see Appendix~\ref{app:experimental-details}). All quantities (except the compute) have units $(\$ \times 10^{-3})$.
}
\label{tab:auction-TruthfulQA-bon}
\vspace{0.5cm}
\begin{subtable}{\textwidth}
\centering
\begin{tabular}{lcccc}
\toprule
\multicolumn{1}{l}{
\smash{
\begin{tabular}[t]{@{}l@{}}
\small margin $=0.25$ \\
\small $\alpha=\$0.008$
\end{tabular}
}
}
 & \multicolumn{3}{c}{Game $\Gcal$} & Auction $\widetilde{\Gcal}$ \\
\cmidrule(lr){2-4} \cmidrule(lr){5-5}
{} & $\beta=2\cdot10^2$ & $\beta=10^3$ & $\beta=10^5$ & \texttt{Qwen2.5-3B} \\
\midrule
Compute ($\theta$)   & 32   & 32   & 32   & 16    \\
User value           & 1.17 & 1.25 & 1.58 & 1.58  \\
Price                & 0.21 & 0.22 & 0.10 & 0.10  \\
Provider(s') utility & 0.04 & 0.04 & 0.02 & 0.04  \\
Social welfare       & 1.22 & 1.30 & 1.60 & 1.62  \\
\bottomrule
\end{tabular}
\end{subtable}

\vspace{1em}

\begin{subtable}{\textwidth}
\centering
\begin{tabular}{lcccc}
\toprule
\multicolumn{1}{l}{
\smash{
\begin{tabular}[t]{@{}l@{}}
\small margin $=0.15$ \\
\small $\alpha=\$0.008$
\end{tabular}
}
}
 & \multicolumn{3}{c}{Game $\Gcal$} & Auction $\widetilde{\Gcal}$ \\
\cmidrule(lr){2-4} \cmidrule(lr){5-5}
{} & $\beta=2\cdot10^2$ & $\beta=10^3$ & $\beta=10^5$ & \texttt{Qwen2.5-3B} \\
\midrule
Compute ($\theta$)   & 32   & 32   & 32   & 4     \\
User value           & 1.17 & 1.25 & 1.58 & 1.58  \\
Price                & 0.21 & 0.22 & 0.10 & 0.06  \\
Provider(s') utility & 0.03 & 0.03 & 0.01 & 0.04  \\
Social welfare       & 1.20 & 1.28 & 1.59 & 1.62  \\
\bottomrule
\end{tabular}
\end{subtable}

\vspace{1em}

\begin{subtable}{\textwidth}
\centering
\begin{tabular}{lcccc}
\toprule
\multicolumn{1}{l}{
\smash{
\begin{tabular}[t]{@{}l@{}}
\small margin $=0.05$ \\
\small $\alpha=\$0.008$
\end{tabular}
}
}
 & \multicolumn{3}{c}{Game $\Gcal$} & Auction $\widetilde{\Gcal}$ \\
\cmidrule(lr){2-4} \cmidrule(lr){5-5}
{} & $\beta=2\cdot10^2$ & $\beta=10^3$ & $\beta=10^5$ & \texttt{Qwen2.5-3B} \\
\midrule
Compute ($\theta$)   & 32   & 32   & 32   & 4     \\
User value           & 1.17 & 1.25 & 1.58 & 1.58  \\
Price                & 0.21 & 0.22 & 0.10 & 0.06  \\
Provider(s') utility & 0.01 & 0.01 & 0.00 & 0.05  \\
Social welfare       & 1.18 & 1.26 & 1.58 & 1.62  \\
\bottomrule
\end{tabular}
\end{subtable}

\vspace{1em}

\begin{subtable}{\textwidth}
\centering
\begin{tabular}{lcccc}
\toprule
\multicolumn{1}{l}{
\smash{
\begin{tabular}[t]{@{}l@{}}
\small margin $=0.25$ \\
\small $\alpha=\$0.016$
\end{tabular}
}
}
 & \multicolumn{3}{c}{Game $\Gcal$} & Auction $\widetilde{\Gcal}$ \\
\cmidrule(lr){2-4} \cmidrule(lr){5-5}
{} & $\beta=2\cdot10^2$ & $\beta=10^3$ & $\beta=10^5$ & \texttt{Qwen2.5-3B} \\
\midrule
Compute ($\theta$)   & 32   & 32   & 32   & 32    \\
User value           & 2.62 & 2.83 & 3.29 & 3.19  \\
Price                & 0.22 & 0.23 & 0.13 & 0.23  \\
Provider(s') utility & 0.04 & 0.05 & 0.03 & 0.12  \\
Social welfare       & 2.66 & 2.87 & 3.32 & 3.32  \\
\bottomrule
\end{tabular}
\end{subtable}

\vspace{1em}

\begin{subtable}{\textwidth}
\centering
\begin{tabular}{lcccc}
\toprule
\multicolumn{1}{l}{
\smash{
\begin{tabular}[t]{@{}l@{}}
\small margin $=0.25$ \\
\small $\alpha=\$0.004$
\end{tabular}
}
}
 & \multicolumn{3}{c}{Game $\Gcal$} & Auction $\widetilde{\Gcal}$ \\
\cmidrule(lr){2-4} \cmidrule(lr){5-5}
{} & $\beta=2\cdot10^2$ & $\beta=10^3$ & $\beta=10^5$ & \texttt{Qwen2.5-3B} \\
\midrule
Compute ($\theta$)   & 32   & 32   & 8    & 2     \\
User value           & 0.47 & 0.49 & 0.79 & 0.78  \\
Price                & 0.21 & 0.21 & 0.03 & 0.03  \\
Provider(s') utility & 0.04 & 0.04 & 0.01 & 0.03  \\
Social welfare       & 0.51 & 0.54 & 0.79 & 0.81  \\
\bottomrule
\end{tabular}
\end{subtable}

\end{table}



\subsubsection{Attribution of the welfare gain of the auction mechanism.}
Compared to the test-time compute game $\mathcal{G}$, the auction $\widetilde{\mathcal{G}}$ improves social welfare through two mechanisms. First, it incentivizes each provider to select the compute $\theta_i^*$ that maximizes their individual contribution $W_i(\theta_i^*)$ to social welfare (Theorem~\ref{thm:second-price-dominant}). Second, it assigns the task to the provider with the highest such contribution. In $\mathcal{G}$, these two mechanisms are coupled: the user rationality $\beta$ shapes both the equilibrium compute levels $\boldsymbol{\theta}^\dagger$ and the market shares. To isolate the effect of compute from that of allocation, we consider the best-of-$n$ game on GSM8K from Section~\ref{sec:experiments} at its equilibrium $\boldsymbol{\theta}^\dagger$ for $\beta = 10^3$. In this instance, the same provider, Qwen2.5-7B, offers the highest value $V_i(\theta_i^\dagger)$, makes the highest contribution $W_i(\theta_i^\dagger)$ to social welfare, and wins the auction (Table~\ref{tab:auction-gsm8k-bon}). We compute the social welfare that would be attained if Qwen2.5-7B served all users at its equilibrium compute $\theta_i^\dagger$. We then compare it to the social welfare of the auction, in which the same provider serves all users at $\theta_i^*$. Since both outcomes share the same allocation, any difference in social welfare is due to compute alone. Table~\ref{tab:compute-vs-allocation} shows that, at equilibrium, Qwen2.5-7B uses $\theta_i^\dagger = 16$ samples and yields a social welfare of \$0.00643. Under the auction, it uses $\theta_i^* = 2$ samples and attains \$0.00736. Hence, in this instance, the auction's welfare gain over the optimally allocated equilibrium comes entirely from correcting over-compute.

\begin{table}[ht]
  \centering
  \caption{\xhdr{Effect of compute on social welfare in the best-of-$n$ game on GSM8K at $\beta = 10^3$, with the allocation held fixed}}
  \label{tab:compute-vs-allocation}
  \begin{tabular}{llcc}
    \toprule
    Setting & Provider & Compute & Social welfare \\
    \midrule
    $\mathcal{G}$, all users served at equilibrium compute & Qwen2.5-7B & $\theta_i^\dagger = 16$ & $\$0.00643$ \\
    Auction $\widetilde{\mathcal{G}}$ & Qwen2.5-7B & $\theta_i^* = 2$ & $\$0.00736$ \\
    \bottomrule
  \end{tabular}
\end{table}

\clearpage
\newpage

\end{document}